\documentclass[a4paper, final]{article}

\usepackage{a4wide}
\usepackage[utf8]{inputenc}
\usepackage{amsfonts}
\usepackage{amsmath, amsthm}
\usepackage{amssymb} 
\usepackage{enumitem}
\usepackage{microtype}

\theoremstyle{definition}

\newtheorem{theorem}{Theorem}
\newtheorem{lemma}[theorem]{Lemma}
\newtheorem{proposition}[theorem]{Proposition}
\newtheorem{corollary}[theorem]{Corollary}
\newtheorem{example}[theorem]{Example}
\newtheorem{definition}[theorem]{Definition}

\newtheorem*{lemma*}{Lemma}

\newcommand{\cA}{\mathcal{A}}
\newcommand{\cB}{\mathcal{B}}
\newcommand{\cC}{\mathcal{C}}
\newcommand{\fD}{\mathfrak{D}}
\newcommand{\cG}{\mathcal{G}}
\newcommand{\cL}{\mathcal{L}}

\newcommand{\cO}{\mathcal{O}}
\newcommand{\cS}{\mathcal{S}}
\newcommand{\bS}{\overline{S}}
\newcommand{\cT}{\mathcal{T}}

\newcommand{\<}{\langle}
\renewcommand{\>}{\rangle}

\newcommand{\vars}{\text{\upshape{vars}}}

\newcommand{\len}{\text{len}}

\newcommand{\hchi}{\widehat{\chi}}
\newcommand{\hphi}{\widehat{\varphi}}
\newcommand{\hpsi}{\widehat{\psi}}
\newcommand{\heta}{\widehat{\eta}}

\newcommand{\tX}{\widetilde{X}}

\newcommand{\subst}[1]{\bigl[#1\bigr]}

\newcommand{\semequiv}{\mathrel{|}\joinrel\Relbar\joinrel\mathrel{|}}

\newcommand{\tL}{\widetilde{\cL}}
\newcommand{\tchi}{\widetilde{\chi}}

\newcommand{\teta}{\widetilde{\eta}}

\newcommand{\At}{\text{At}}

\newcommand{\out}{\text{out}}
\newcommand{\Out}{\text{Out}}
\newcommand{\hsigma}{\widehat{\sigma}}
\newcommand{\Sk}{\text{Sk}}
\newcommand{\Mon}{\text{Mon}}

\newcommand{\bool}{\text{bool}}

\newcommand{\fU}{\mathfrak{U}}
\newcommand{\fa}{\mathfrak{a}}
\newcommand{\fb}{\mathfrak{b}}

\newcommand{\fc}{\mathfrak{c}}
\newcommand{\fd}{\mathfrak{d}}

\newcommand{\va}{\vec{\fa}}
\newcommand{\vb}{\vec{\fb}}
\newcommand{\vc}{\vec{\fc}}
\newcommand{\vd}{\vec{\fd}}

\newcommand{\vu}{\vec{\mathbf{u}}}
\newcommand{\vv}{\vec{\mathbf{v}}}
\newcommand{\vx}{\vec{\mathbf{x}}}
\newcommand{\vy}{\vec{\mathbf{y}}}
\newcommand{\vz}{\vec{\mathbf{z}}}

\newcommand{\fP}{\mathfrak{P}}

\newcommand{\twoup}[2]{{2^{\uparrow #1}(#2)}}

\renewcommand{\Cup}{\cC^{\uparrow}}

\newcommand{\im}{\text{im}}
\newcommand{\him}{\widehat{\im}}
\newcommand{\idx}{\text{idx}}

\newcommand{\Mapsto}{{\mathop{\mapsto}}}

\newcommand{\degree}{\partial}

\newcommand{\SF}{\text{SF}}
\newcommand{\GBSR}{\text{GBSR}}

\newcommand{\NEXPTIME}{\textsc{NExpTime}}

\begin{document}
\title{On Generalizing Decidable Standard Prefix Classes\\ of First-Order Logic}
\author{
	Marco Voigt\\
	\small\textit{Max Planck Institute for Informatics, Saarland Informatics Campus, Saarbr\"ucken, Germany,}\\
	\small\textit{Saarbr\"ucken Graduate School of Computer Science}
}	
\date{}
\maketitle

\begin{abstract}
	Recently, the separated fragment (SF) of first-order logic has been introduced. 
	Its defining principle is that universally and existentially quantified variables may not occur together in atoms.
	SF properly generalizes both the Bernays--Sch\"onfinkel--Ramsey (BSR) fragment and the relational monadic fragment.	
	In this paper the restrictions on variable occurrences in SF sentences are relaxed such that universally and existentially 
	quantified variables may occur together in the same atom under certain conditions.  
	Still, satisfiability can be decided. This result is established in two ways: firstly, by an effective equivalence-preserving translation into the BSR fragment, and, secondly, by a model-theoretic argument.
	
	Slight modifications to the described concepts facilitate the definition of other decidable classes of first-order sentences.
	The paper presents a second fragment which is novel, has a decidable satisfiability problem, and properly contains the Ackermann fragment 
	and---once more---the relational monadic fragment. 
	The definition is again characterized by restrictions on the occurrences of variables in atoms. More precisely, after certain 
	transformations, Skolemization yields only unary functions and constants, and every atom contains at most one universally quantified variable.
	An effective satisfiability-preserving translation into the monadic fragment is devised and employed to prove decidability of the associated satisfiability problem.
\end{abstract}

%%%%%%%%%%%%%%%%%%%%%%%%%%%%%%%%%%%%%%%%%%%%%%%%%%%%%%%%%%%%%%%%%%%%%%%%%%%%%%%%%%%
%%%%%%%%%%%%%%%%%%%%%%%%%%%%%%%%%%%%%%%%%%%%%%%%%%%%%%%%%%%%%%%%%%%%%%%%%%%%%%%%%%%
%%%%%%%%%%%%%%%%%%%%%%%%%%%%%%%%%%%%%%%%%%%%%%%%%%%%%%%%%%%%%%%%%%%%%%%%%%%%%%%%%%%

\section{Introduction}\label{section:Intro}

Quantifier prefix classes have for long been a dominating paradigm for the classification of first-order sentences into decidable and undecidable fragments. 
The \emph{Bernays--Sch\"onfinkel--Ramsey fragment (BSR)}---$\exists^*\forall^*$ sentences---and the \emph{Ackermann fragment}---$\exists^*\forall\exists^*$ sentences---are two prefix classes that are well-known for their decidable satisfiability problem. 
We show in this paper how both of them can be generalized to substantially larger, decidable fragments in which quantifier prefixes are not restricted anymore.
Instead, we formulate certain conditions on how existentially and universally quantified variables may be interlinked by joint occurrences in atoms. 
This means that the classes of relational first-order sentences characterized by the quantifier prefixes $\exists^*\forall^*$ and $\exists^*\forall\exists^*$ are merely the tips of two icebergs that we shall call the \emph{generalized BSR fragment (GBSR)} and the \emph{generalized Ackermann fragment (GAF)}, respectively. 
To the best knowledge of the author, both fragments are novel.
Interestingly, the \emph{relational monadic fragment}---relational first-order sentences in which all predicate symbols have arity one---is contained in both GBSR and GAF.

Recently, the \emph{separated fragment (SF)} has been introduced \cite{Voigt2016}.
It can be considered as an intermediate step between BSR and GBSR.
Its defining principle is that universally and existentially quantified variables may not occur together in atoms. 
In~\cite{Voigt2017} it is shown that the satisfiability problem for SF is non-elementary.
Yet, SF as well as GBSR enjoy the finite model property.
However, the size of smallest models of a satisfiable GBSR sentences grows at least $k$-fold exponentially with the length of the sentences for arbitrarily large $k \geq 0$.
This is in contrast with BSR, where the size of smallest models grows at most linearly with the length of satisfiable sentences.
Since every GBSR sentence is equivalent to some BSR sentence, 
one could say that GBSR is as expressive as BSR but admits a non-elementarily more compact representation of models.
It remains an open question whether a similar result holds for GAF relative to the Ackermann fragment.

The following example gives a taste of the kind of sentences treated in this paper and outlines a key method that we shall employ later. 
The sentence $\varphi_1$ is a sample taken from GAF and $\varphi_2$ belongs to GBSR and to GAF.

%%%%%%%%%%%%%%%%%%%%%%%%%%%%%%%%%%%%%%%%%%%%%%%%%%%%%%%%%%%%%%%%%%%%%%%%%%%%%%%%
%%%%%%%%%%%%%%%%%%%%%%%%%%%%%%%%%%%%%%%%%%%%%%%%%%%%%%%%%%%%%%%%%%%%%%%%%%%%%%%%
\begin{example}\label{example:GAFandGBSR}
	Consider the first-order sentence $\varphi_1 := \exists u \forall x \exists v \forall z \exists y_1 y_2 . \bigl(\neg P(u, x) \vee \bigl(Q(x,v) \wedge R(u,z,$ $y_1)\bigr)\bigr) \;\wedge\; \bigl(P(u, x) \vee \bigl(\neg Q(x,v) \wedge \neg R(u,z,y_2)\bigr)\bigr)$.
	Due to the Boolean structure of $\varphi_1$, the quantifiers $\exists y_2$, $\exists y_1$, and $\forall z$ can be moved inwards immediately but $\exists v$ cannot.
	Because of the two universal quantifiers $\forall x$ and $\forall z$, which are even interspersed with an existential one, $\varphi_1$ does not lie in the Ackermann fragment or in BSR. 
	Skolemization of $\varphi_1$ leads to 
	$\forall x z. 
			\bigl(\neg P(c, x) \vee \bigl(Q\bigl(x,f(x)\bigr) \wedge R\bigl(c,z,g(x,z)\bigr)\bigr)\bigr)
			 \wedge \bigl(P(c, x) \vee \bigl(\neg Q\bigl(x,f(x)\bigr) \wedge \neg R\bigl(c,z,h(x,z)\bigr)\bigr)\bigr)$
	and thus explicitly fixes the dependency of $y_1$ on the universally quantified variables $x$ and $z$, as $y_1$ is replaced with the term $g(x,z)$.
	However, the shape of the original $\varphi_1$ did not immediately indicate such a dependency of $y_1$ on $x$, since $x$ and $y_1$ do not occur together in any atom. 
	Moreover, there are no other variables that depend on $x$ and establish a connection between $x$ and $y_1$ by means of joint occurrences in atoms. 
	One may say that it is the Boolean structure of $\varphi_1$ alone which causes a dependency of $y_1$ on $x$, and that such a form of dependency has only a \emph{finite character}.
	(This finite character will be made explicit by the notions of \emph{fingerprints} and \emph{uniform} strategies that we introduce in Section~\ref{section:ModeltheoreticApproachToGBSRsat}.)
	
	The described point of view is supported by the existence of an equivalent sentence $\varphi_1'$, in which the dependency of $y_1$ on $x$ has vanished. The price we have to pay, however, is an increase in the size of the formula.\\
		\strut\hspace{2ex}
		$\varphi'_1 :=
			\exists u. 
			\bigl(\forall x.\bigl(\neg P(u, x) \vee \exists v. Q(x,v) \bigr)\bigr) 
				\wedge \bigl(\bigl(\forall x.\neg P(u, x)\bigr) \vee \forall z \exists y_1. R(u,z,y_1) \bigr)$ \\
		\strut\hspace{12ex}
			$\wedge \bigl(\forall x.\bigl(\exists v. \neg Q(x,v)\bigr) \vee P(u,x) \bigr) 
				\wedge \bigl(\bigl(\forall x.\exists v. \neg Q(x,v)\bigr) \vee \forall z \exists y_1. R(u,z,y_1) \bigr)$ \\
		\strut\hspace{12ex}
			$\wedge \bigl(\bigl(\forall z \exists y_2. \neg R(u,z,y_2)\bigr) \vee \forall x.P(u,x) \bigr)
				\wedge \bigl(\bigl(\forall z \exists y_2. \neg R(u,z,y_2)\bigr) \vee \forall x \exists v. Q(x,v) \bigr)$ \\
		\strut\hspace{12ex}
			$\wedge \bigl(\bigl(\forall z \exists y_2. \neg R(u,z,y_2)\bigr) \vee \forall z \exists y_1. R(u,z,y_1) \bigr)$\\
	Transforming $\varphi_1$ into $\varphi'_1$ requires only basic logical laws (details can be found in the appendix):  
	first, we push the quantifiers $\exists y_2, \exists y_1, \forall z$ inwards as far as possible.
	Then, we construct a disjunction of conjunctions of certain subformulas using distributivity.
	This allows us to move the quantifier $\exists v$ inwards. 
	Afterwards, we apply the laws of distributivity again to obtain a conjunction of disjunctions of certain subformulas.
	This step enables us to push the universal quantifier $\forall x$ inwards.
	In the resulting sentence every occurrence of an existential quantifier lies in the scope of at most one universal quantifier. 
	Moreover, every atom in the original formula $\varphi_1$ contains at most one universally quantified variable.
	
	Skolemization of $\varphi_1'$ leads to a sentence whose shape is quite close to the shape of a Skolemized sentence from the Ackermann fragment. 
	More precisely, every atom contains at most one variable, possibly with multiple occurrences. 
	The only difference is that we allow for more than only one universally quantified variable in the sentence as a whole, but at most one in every atom. 
	It is this particular form that one can exploit to construct an equisatisfiable  monadic sentence.

	As another example, consider the sentence $\varphi_2 := \exists u \forall x \exists y \forall z. \bigl(P(u,z) \wedge Q(u,x)\bigr) \vee \bigl(P(y,z) \wedge Q(u,y) \bigr)$. This sentence can be transformed in the same spirit, leading to the equivalent \\
		\noindent
		\strut\hspace{2ex}
		$\varphi'_2 := 
			\exists u \exists y \forall x z v.\,	
			\bigl(\bigl(P(u,x) \vee P(y,x)\bigr) \wedge P(u,x) \wedge Q(u,x) \bigr)$ \\
		\strut\hspace{20ex}
			$\vee \bigl(\bigl(P(u,z) \vee P(y,z)\bigr) \wedge Q(u,y) \wedge Q(u,z) \bigr)$ \\
		\strut\hspace{20ex}
			$\vee \bigl(\bigl( P(u,v) \vee P(y,v)\bigr) \wedge Q(u,y) \wedge P(y,v) \bigr)$.
				
	While Skolemization of $\varphi_2$ introduces terms $f(x)$ to replace $y$, $\varphi'_2$ is a much nicer target for Skolemization, since all introduced symbols are constants. This second approach is so attractive, because it leads to a BSR sentence.
\end{example}
As a matter of fact, the sentence $\varphi_2$ belongs to GBSR and GAF at the same time, while it does not belong to the Ackermann fragment, SF, BSR, or the monadic fragment. Hence, even the intersection of GBSR and GAF contains sentences which do not fall into the categories offered by standard fragments.

The transformation technique outlined in Example~\ref{example:GAFandGBSR} is one tool with which we establish the decidability of GBSR and GAF.
An interesting model-theoretic approach to establishing a small model property for GBSR sentences is presented as well.
Moreover, we employ a proof-theoretic result to argue that satisfiability of GAF sentences with equality is decidable.

%%%%%%%%%%%%%%%%%%%%%%%%%%%%%%%%%%%%%%%%%%%%%%%%%%%%%%%%%%%%%%%%%%%%%%%%%%%%%%%%
%%%%%%%%%%%%%%%%%%%%%%%%%%%%%%%%%%%%%%%%%%%%%%%%%%%%%%%%%%%%%%%%%%%%%%%%%%%%%%%%
In short, the main contributions of the present paper are the following.
In Section~\ref{section:GeneralizedBSR} we define GBSR and outline an effective equivalence-preserving transformation from GBSR into BSR (Lemma~\ref{lemma:TransformationGBSR}), which entails decidability of GBSR-satisfiability (Theorem~\ref{theorem:DecidabilityGBSR}). 
Using this translation, we moreover derive a Craig--Lyndon interpolation theorem for GBSR (Theorem~\ref{theorem:InterpolationGBSR}).
In Section~\ref{section:ModeltheoreticApproachToGBSRsat} we develop a model-theoretic point of view, which eventually leads to a small model property for GBSR (Theorem~\ref{theorem:GBSRsmallModelProperty}).
The computational hardness of GBSR-satisfiability is derived from the hardness of SF-satisfiability (cf.\ Theorems~\ref{theorem:ComputationalHardnessGBSR} and~\ref{theorem:ComputationalComplexityGBSR}).
In Section~\ref{section:GeneralizedAckermann} we introduce GAF.
Decidability of GAF-satisfiability is shown 
	(a) for GAF sentences with equality but without non-constant function symbols by employing a proof-theoretic result (Theorem~\ref{theorem:DecidabilityGAFwithEquality}) and
	(b) for GAF sentences without equality but with arbitrarily nested unary function symbols via an effective, (un)satisfiability-preserving transformation from GAF into the monadic fragment with unary function symbols (Theorem~\ref{theorem:DecidabilityGAFwithUnaryFunctions}).

In order to facilitate smooth reading, long proofs are only sketched in the main text and presented in full in the appendix.

%%%%%%%%%%%%%%%%%%%%%%%%%%%%%%%%%%%%%%%%%%%%%%%%%%%%%%%%%%%%%%%%%%%%%%%%%%%%%%%%%%%
%%%%%%%%%%%%%%%%%%%%%%%%%%%%%%%%%%%%%%%%%%%%%%%%%%%%%%%%%%%%%%%%%%%%%%%%%%%%%%%%%%%
%%%%%%%%%%%%%%%%%%%%%%%%%%%%%%%%%%%%%%%%%%%%%%%%%%%%%%%%%%%%%%%%%%%%%%%%%%%%%%%%%%%

%%%%%%%%%%%%%%%%%%%%%%%%%%%%%%%%%%%%%%%%%%%%%%%%%%%%%%%%%%%%%%%%%%%%%%%%%%%%%%%%
\section{Notation and preliminaries}\label{section:Preliminaries}
%%%%%%%%%%%%%%%%%%%%%%%%%%%%%%%%%%%%%%%%%%%%%%%%%%%%%%%%%%%%%%%%%%%%%%%%%%%%%%%%

We consider first-order logic formulas with equality. 
We call a first-order formula \emph{relational} if it contains neither function nor constant symbols.
We use $\varphi(x_1, \ldots, x_m)$ to denote a formula $\varphi$ whose free variables form a subset of $\{x_1, \ldots, x_m\}$.
In all formulas, if not explicitly stated otherwise, we tacitly assume that no variable occurs freely and bound at the same time and that no variable is bound by two different occurrences of quantifiers.
For convenience, we sometimes identify tuples $\vx$ of variables with the set containing all the variables that occur in $\vx$.
By $\vars(\varphi)$ we denote the set of all variables occurring in $\varphi$.
Similar notation is used for other syntactic objects.

A sentence $\varphi := \forall \vx_1 \exists \vy_1 \ldots \forall \vx_n \exists \vy_n. \psi$ is said to be in \emph{standard form} if and only if $\psi$  is quantifier free, in negation normal form, and exclusively contains the connectives $\wedge, \vee, \neg$. 
In addition, we assume that all variables bound in the quantifier prefix actually occur in $\psi$.
The tuples $\vx_1$ and $\vy_n$ may be empty, i.e.\ the quantifier prefix does not have to start with a universal quantifier, and it does not have to end with an existential quantifier.

As usual, we interpret a formula $\varphi$ with respect to given structures. A \emph{structure} $\cA$ consists of a nonempty universe $\fU_\cA$ and interpretations $f^\cA$ and $P^\cA$ of all considered function and predicate symbols, in the usual way. 
Given a formula $\varphi$, a structure $\cA$, and a variable assignment $\beta$, we write $\cA, \beta \models \varphi$ if $\varphi$ evaluates to \emph{true} under $\cA$ and $\beta$. We write $\cA \models \varphi$ if $\cA, \beta \models \varphi$ holds for every $\beta$.
The symbol $\models$ also denotes semantic entailment of formulas, i.e.\ $\varphi \models \psi$ holds whenever for every structure $\cA$ and every variable assignment $\beta$, $\cA,\beta \models \varphi$ entails $\cA,\beta \models \psi$.
The symbol $\semequiv$ denotes \emph{semantic equivalence} of formulas, i.e.\ $\varphi \semequiv \psi$ holds whenever $\varphi \models \psi$ and $\psi \models \varphi$.
We call two sentences $\varphi$ and $\psi$ \emph{equisatisfiable} if $\varphi$ has a model if and only if $\psi$ has one.
A structure $\cA$ is a \emph{substructure} of a structure $\cB$ (over the same signature) if (1) $\fU_\cA \subseteq \fU_\cB$, (2) $c^\cA = c^\cB$ for every constant symbol $c$, (3) $P^\cA = P^\cB \cap \fU_\cA^m$ for every $m$-ary predicate symbol $P$, and (4) $f^\cA(\fa_1, \ldots, \fa_m) = f^\cB(\fa_1, \ldots, \fa_m)$ for every $m$-ary function symbol $f$ and every $m$-tuple $\<\fa_1, \ldots, \fa_m\> \in \fU_\cA^m$. 
The following are standard lemmas, see, e.g., \cite{Ebbinghaus1994} for a proof.

%%%%%%%%%%%%%%%%%%%%%%%%%%%%%%%%%%%%%%%%%%%%%%%%%%%%%%%%%%%%%%%%%%%%%%%%%%%%%%%%
\begin{lemma}[Substructure Lemma]
	Let $\varphi$ be a first-order sentence in prenex normal form without existential quantifiers and let $\cA$ be a substructure of $\cB$.
	$\cB \models \varphi$ entails $\cA \models \varphi$.
\end{lemma}

%%%%%%%%%%%%%%%%%%%%%%%%%%%%%%%%%%%%%%%%%%%%%%%%%%%%%%%%%%%%%%%%%%%%%%%%%%%%%%%%
\begin{lemma}[Miniscoping]\label{lemma:Miniscoping}
	Let $\varphi, \psi, \chi$ be arbitrary first-order formulas, and assume that $x$ and $y$ do not occur freely in $\chi$.
	We have the following equivalences, where $\circ \in \{\wedge, \vee\}$:\\
	\begin{tabular}{cl@{\hspace{2ex}}c@{\hspace{2ex}}l@{\hspace{8ex}}cl@{\hspace{2ex}}c@{\hspace{2ex}}l}
			(i)	&	$\exists y. (\varphi \vee \psi)$		&$\semequiv$&		$(\exists y. \varphi) \vee (\exists y. \psi)$
		&	(ii)	&	$\forall x. (\varphi \wedge \psi)$		&$\semequiv$&		$(\forall x. \varphi) \wedge (\forall x. \psi)$ \\
			(iii)	&	$\exists y. (\varphi \circ \chi)$		&$\semequiv$&		$(\exists y. \varphi) \circ \chi$
		&	(iv)	&	$\forall x. (\varphi \circ \chi)$		&$\semequiv$&		$(\forall x. \varphi) \circ \chi$
	\end{tabular}	
\end{lemma}

%%%%%%%%%%%%%%%%%%%%%%%%%%%%%%%%%%%%%%%%%%%%%%%%%%%%%%%%%%%%%%%%%%%%%%%%%%%%%%%%%%%
%%%%%%%%%%%%%%%%%%%%%%%%%%%%%%%%%%%%%%%%%%%%%%%%%%%%%%%%%%%%%%%%%%%%%%%%%%%%%%%%%%%
%%%%%%%%%%%%%%%%%%%%%%%%%%%%%%%%%%%%%%%%%%%%%%%%%%%%%%%%%%%%%%%%%%%%%%%%%%%%%%%%%%%

%%%%%%%%%%%%%%%%%%%%%%%%%%%%%%%%%%%%%%%%%%%%%%%%%%%%%%%%%%%%%%%%%%%%%%%%%%%%%%%%
\section{The generalized Bernays--Sch\"onfinkel--Ramsey fragment}\label{section:GeneralizedBSR}
%%%%%%%%%%%%%%%%%%%%%%%%%%%%%%%%%%%%%%%%%%%%%%%%%%%%%%%%%%%%%%%%%%%%%%%%%%%%%%%%

In this section we generalize the Bernays--Sch\"onfinkel--Ramsey fragment---$\exists^* \forall^*$ sentences with equality but without non-constant function symbols. 
We show in two ways that the satisfiability problem for the new fragment is decidable.
The first approach (Section~\ref{section:SyntacticApproachToGBSRsat}) is of a syntactic nature, as it is based on an effective translation into BSR.
The second approach (Section~\ref{section:ModeltheoreticApproachToGBSRsat}) uses model-theoretic techniques to establish a small model property.

For the considerations in this section we fix a first-order sentence $\varphi := \forall \vx_1 \exists \vy_1 \ldots \forall \vx_n \exists \vy_n. \psi$ in standard form that may contain the distinguished equality predicate and constant symbols but no non-constant function symbols.
Let $\At$ be the set of all atoms occurring in $\varphi$ and let $\vx := \vx_1 \cup \ldots \cup \vx_n$ and $\vy := \vy_1 \cup \ldots \cup \vy_n$.

%%%%%%%%%%%%%%%%%%%%%%%%%%%%%%%%%%%%%%%%%%%%%%%%%%%%%%%%%%%%%%%%%%%%%%%%%%%%%%%%
\begin{definition}[GBSR fragment, axiomatically]\label{definition:GBSRaxiomatic}
	The sentence $\varphi$ belongs to the \emph{generalized Bernays--Sch\"onfinkel--Ramsey fragment (GBSR)} if and only if
	we can partition $\At$ into sets $\At_0, \At_1, \ldots, \At_n$ such that
	\begin{enumerate}[label=(\roman{*}), ref=(\roman{*})]
		\item\label{enum:GBSRaxiomatic:I} for every $i$, $0 \leq i \leq n$, it holds $\vars(\At_i) \subseteq \vy_1 \cup \ldots \cup \vy_i \cup \vx_{i+1} \cup \ldots \cup \vx_n$, and
		\item\label{enum:GBSRaxiomatic:II} for all distinct $i,j$, $0 \leq i < j \leq n$, it holds $\vars(\At_i) \cap \vars(\At_j) \cap \vx = \emptyset$.
	\end{enumerate}
\end{definition}
We shall see in Section~\ref{section:ModeltheoreticApproachToGBSRsat} how the described way of partitioning the atoms in GBSR sentences facilitates a model-theoretic approach to proving decidability of the satisfiability problem for GBSR sentences \emph{(GBSR-satisfiability)}.

We complement the axiomatic definition of GBSR with an alternative definition of an algorithmic flavor.
For one thing, the algorithmic definition shows that membership in GBSR is easily decidable.
Moreover, the notions used in the algorithmic definition will be of use in the syntactic approach to decidability of GBSR-satisfiability outlined in Section~\ref{section:SyntacticApproachToGBSRsat}. 
We shall see in Lemma~\ref{lemma:AlternativeDefinitionOfGBSR} that both definitions are equivalent.
But first, we need additional notation.

Given $\varphi$, we define the undirected graph $\cG_{\varphi} := \<V, E\>$ by setting $V := \vx$ and 
	$E := \{ \<x, x'\> \mid \text{there is an atom in $\varphi$ containing both $x$ and $x'$} \}$. 
A \emph{connected component in $\cG_{\varphi}$} is a maximal subset $C \subseteq V$ such that for all distinct variables $x, x' \in C$ the transitive closure of $E$ contains the pair $\<x,x'\>$. The set of all connected components in $\cG_{\varphi}$ forms a partition of $V$. 
For every connected component $C$ in $\cG_{\varphi}$ we denote by $\cL(C)$ the \emph{set of all literals in $\varphi$ which contain at least one variable from $C$}.

For every index $k$, $1\leq k\leq n$, we denote by $\cL_k$ the smallest set of literals such that
	$\cL_k$ contains all literals taken from $\varphi$ in which variables from $\vy_k$ occur, and
	for every connected component $C$ in $\cG_{\varphi}$ containing a variable $x \in \vars(\cL_k)$ it holds $\cL(C) \subseteq \cL_k$.
Intuitively, every $\cL_k$ constitutes a reasonably small superset of the literals in $\varphi$ which remain in the scope of the quantifier block $\exists \vy_k$ when the rules of miniscoping (cf.\ Lemma~\ref{lemma:Miniscoping}) are applied from left to right.

%%%%%%%%%%%%%%%%%%%%%%%%%%%%%%%%%%%%%%%%%%%%%%%%%%%%%%%%%%%%%%%%%%%%%%%%%%%%%%%%
\begin{definition}[GBSR fragment, algorithmically]\label{definition:GBSRalgorithmic}
	The sentence $\varphi$ is in \emph{GBSR} if and only if for all $k,\ell$ with $1\leq \ell \leq k \leq n$ it holds $\vars(\cL_k) \cap \vx_{\ell} = \emptyset$.
\end{definition}

By $\tL_k$ we denote the set $\cL_k \setminus \bigcup_{\ell > k} \cL_{\ell}$.
Moreover, $\tL_0$ stands for the set of all literals in $\varphi$ which do not belong to any of the $\tL_k$.
Note that the sets $\tL_0, \ldots, \tL_n$ form a partition of the set of all literals in $\varphi$.
For every $k$, $0 \leq k \leq n$, we write $\tX_{k}$ to address the set $\vars(\tL_k) \cap \vx$.

%%%%%%%%%%%%%%%%%%%%%%%%%%%%%%%%%%%%%%%%%%%%%%%%%%%%%%%%%%%%%%%%%%%%%%%%%%%%%%%%
\begin{lemma}\label{lemma:LiteralsDecomposition}~
	\begin{enumerate}[label=(\roman{*}), ref=(\roman{*})]
		\item\label{enum:LiteralsDecomposition:I} 
			For all distinct indices $k, \ell$ we have $\tX_{k} \cap \tX_{\ell} = \emptyset$.
		\item\label{enum:LiteralsDecomposition:II} 
			For every $k$ it holds $\vars(\tL_{k}) \cap \vy \subseteq \bigcup_{1 \leq \ell \leq k} \vy_{\ell}$.	
		\item\label{enum:LiteralsDecomposition:III} 
			If $\varphi$ satisfies Definition~\ref{definition:GBSRalgorithmic}, then
			for every $k$ it holds $\tX_{k} \subseteq \bigcup_{k < \ell \leq n} \vx_{\ell}$.
	\end{enumerate}
\end{lemma}
\begin{proof}~
	\begin{description}
		\item{Ad \ref{enum:LiteralsDecomposition:I}:} 
			Suppose there are distinct indices $k, \ell$, $k < \ell$, and a variable $x \in \vx$ such that $x \in \tX_k \cap \tX_{\ell}$. 
			But then, there must be literals $L_k \in \tL_k \subseteq \cL_k$ and $L_{\ell} \in \tL_{\ell} \subseteq \cL_{\ell}$, both containing $x$. 
			Let $C$ denote the (unique) connected component in $\cG_\psi$ to which $x$ belongs. 
			By definition of $\cL(C)$, both $L_k$ and $L_{\ell}$ belong to $\cL(C)$. 
			Therefore, it holds $\{L_k, L_{\ell}\} \subseteq \cL(C) \subseteq \cL_{\ell}$. 
			But since $\tL_k \subseteq \cL_k \setminus \cL_{\ell}$, $L_k$ cannot belong to $\cL_k$. A contradiction.
		
		\item{Ad \ref{enum:LiteralsDecomposition:II}:} 
			Let $k \leq n$ be some non-negative integer. Since for any $\ell > k$ $\cL_\ell$ contains all literals in which a variable $y \in \vy_\ell$ occurs, $\tL_k \subseteq \cL_k \setminus \cL_\ell$ cannot contain any occurrence of such a variable $y$.
		
		\item{Ad \ref{enum:LiteralsDecomposition:III}:}
			Let $k \leq n$ be some non-negative integer. $\varphi$'s belonging to GBSR means $\vars(\cL_k) \cap \bigcup_{\ell' \leq k} \vx_{\ell'} = \emptyset$. Because of $\tX_k = \vars(\tL_k) \cap \vx \subseteq \vars(\cL_k) \cap \vx$, we conclude $\tX_k \cap \bigcup_{\ell' \leq k} \vx_{\ell'} = \emptyset$. Hence, we have $\tX_k \subseteq \bigcup_{\ell > k} \vx_{\ell}$.
			\qedhere
	\end{description}
\end{proof}

We now have the right notions at hand to show that the axiomatic and the algorithmic definitions of GBSR sentences yield the same fragment of first-order logic. 
%%%%%%%%%%%%%%%%%%%%%%%%%%%%%%%%%%%%%%%%%%%%%%%%%%%%%%%%%%%%%%%%%%%%%%%%%%%%%%%%
\begin{lemma}\label{lemma:AlternativeDefinitionOfGBSR}
	The sentence $\varphi$ satisfies Definition~\ref{definition:GBSRaxiomatic} if and only if it satisfies Definition~\ref{definition:GBSRalgorithmic}.
\end{lemma}
\begin{proof}
	The \emph{if}-direction follows immediately from Lemma~\ref{lemma:LiteralsDecomposition}.
	We just define the $\At_k \subseteq \At$ such that $A \in \At_k$ if and only if either $A$ or $\neg A$ or both belong to $\tL_k$.
	
	The \emph{only if}-direction can be argued as follows.
	For every $i$, $0 \leq i \leq n$, let $X_i := \vars(\At_i) \cap \vx$.
	Consider the graph $\cG_\varphi$.
	Since the $X_1, \ldots, X_n$ are pairwise disjoint, they induce subgraphs of $\cG_\varphi$ that are not connected to one another.
	Moreover, for every connected component $C$ in $\cG_\varphi$, there is one $X_i$ such that $C \subseteq X_i$.
	When we write $\At(C)$ to denote all atoms in $\psi$ that contain variables from $C$, then the previous observation entails that for every connected component $C$ in $\cG_\varphi$ there is some $\At_i$ such that $\At(C) \subseteq \At_i$.
	
	Let $\At'_k$ be the set of atoms occurring in $\cL_k$.
	In other words, let $\At'_k$ be the smallest set of all atoms in $\psi$ that contain variables from $\vy_k$ and for every connected component $C$ in $\cG_\varphi$ containing a variable $x \in \vars(\At_k)$, $\At'_k$ contains $\At(C)$.
	Hence, $\At'_k \subseteq \At_k \cup \ldots \cup \At_n$ and, moreover, $\vars(\At'_k) \cap \vx \subseteq \vx_{k+1} \cup \ldots \cup \vx_n$.
	This means, $\vars(\At'_k) \cap \vx_\ell = \vars(\cL_k) \cap \vx_\ell = \emptyset$ for every $\ell$, $1 \leq \ell \leq k$. 
	It follows that Definition~\ref{definition:GBSRalgorithmic} is satisfied.
\end{proof}

%%%%%%%%%%%%%%%%%%%%%%%%%%%%%%%%%%%%%%%%%%%%%%%%%%%%%%%%%%%%%%%%%%%%%%%%%%%%%%%%
In \cite{Voigt2016} the \emph{separated fragment (SF)} is defined to be the set of first-order sentences $\chi := \exists \vz\, \forall\vu_1 \exists\vv_1 \ldots \forall\vu_n \exists\vv_n. \chi'$ with quantifier-free $\chi'$ that may contain equality but no non-constant function symbols. Moreover, for every atom $A$ in $\chi$ it is required that either $\vars(A) \cap \bigcup_i \vu_i = \emptyset$ or $\vars(A) \cap \bigcup_i \vv_i = \emptyset$ or both hold, i.e.\ variables $u \in \vu_k$ and $v \in \vv_\ell$ may not occur together in any atom in $\chi$ for arbitrary $k, \ell$.
We have advertised GBSR as an extension of SF. 
Indeed, we can partition the set of $\chi$'s atoms into two sets $\At_0, \At_n$ such that $\vars(\At_0) \subseteq \vz \cup \vu_1 \cup \ldots \cup \vu_n$ and $\vars(\At_n) \subseteq \vz \cup \vv_1 \cup \ldots \cup \vv_n$.
This partition obviously satisfies the requirements of Definition~\ref{definition:GBSRaxiomatic}.
Moreover, it is known that SF contains the BSR fragment as well as the relational monadic fragment without equality \cite{Voigt2016}.

%%%%%%%%%%%%%%%%%%%%%%%%%%%%%%%%%%%%%%%%%%%%%%%%%%%%%%%%%%%%%%%%%%%%%%%%%%%%%%%%
On the other hand, Example~\ref{example:GAFandGBSR} contains the sentence $\varphi_2$ which belongs to GBSR but not to SF, thus showing that GBSR is a proper extension of SF.
\begin{proposition}\label{theorem:InclusionGBSR}
	GBSR properly contains the separated fragment, the BSR fragment, and the relational monadic fragment without equality.
\end{proposition}

%%%%%%%%%%%%%%%%%%%%%%%%%%%%%%%%%%%%%%%%%%%%%%%%%%%%%%%%%%%%%%%%%%%%%%%%%%%%%%%%
One of the main results of the present paper is that GBSR-satisfiability is decidable.
\begin{theorem}\label{theorem:DecidabilityGBSR}
	Satisfiability of GBSR sentences is decidable.
\end{theorem}
We prove this theorem in two ways.
The first route is a syntactic one (Section~\ref{section:SyntacticApproachToGBSRsat}): we show that every GBSR sentence can be effectively transformed into an equivalent BSR sentence (see the proof Lemma~\ref{lemma:TransformationGBSR}).
Our second approach is a model-theoretic one: we devise a method that, given any model $\cA$ of a GBSR sentence $\varphi$, constructs a model $\cB$ of $\varphi$ with a domain of a bounded size. 
In other words, we show that GBSR enjoys a \emph{small model property} (Theorem~\ref{theorem:GBSRsmallModelProperty}).
Although already the existence of an effective translation of GBSR sentences into equivalent BSR sentences entails that GBSR inherits the small model property from BSR, our approach does not exploit this fact. 
It rather relies on a technique that emphasizes the finite character of the dependencies between universal variables and the existential variables that lie within their scope in a given GBSR sentence.

Concerning computational complexity, the hierarchy of $k$-\NEXPTIME-complete subproblems of \emph{SF-satisfiability} presented in \cite{Voigt2017} together with the containment of SF in GBSR leads to the observation that GBSR-satisfiability is non-elementary.
\begin{theorem}\label{theorem:ComputationalHardnessGBSR}
	GBSR-satisfiability is $k$-\NEXPTIME-hard for every positive integer $k$.
\end{theorem}
The small model property that we derive for GBSR sentences in Section~\ref{section:ModeltheoreticApproachToGBSRsat} entails the existence of a similar hierarchy of complete problems for GBSR as there exists for SF.
\begin{theorem}\label{theorem:ComputationalComplexityGBSR}
	We can divide the GBSR fragment into an increasing sequence of subclasses $\GBSR_k$ with $k = 1, 2, 3, \ldots$ such that for every $k$ the set of all satisfiable sentences from $\GBSR_k$ forms a $k$-\NEXPTIME-complete set.
\end{theorem}
\begin{proof}[Proof sketch]
	The division of GBSR into subfragments $\GBSR_1 \subseteq \GBSR_2 \subseteq \ldots$ is based on the degree $\degree_\varphi$ of GBSR sentences $\varphi$, which we shall define right after Corollary~\ref{corollary:UniformStrategies}.
	More precisely, we define $\GBSR_k$ to be the subfragment of GBSR that contains all GBSR sentences $\varphi$ with $\degree_\varphi = k-1$.
	By Theorem~\ref{theorem:GBSRsmallModelProperty}, every satisfiable GBSR sentence $\varphi$ has a model with at most $\len(\varphi)^2 \cdot \bigl( \twoup{\degree_\varphi+1}{\len(\varphi)} \bigr)^{\len(\varphi)^2}$ domain elements.
	Hence, for any $k \geq 1$ the satisfiability problem for $\GBSR_k$ is decidable in nondeterministic $k$-fold-exponential time.
	
	On the other hand, the lower bound proof for SF-satisfiability in \cite{Voigt2017} is based on SF-formulas that encode computationally hard \emph{domino problems}.
	One can use the same formulas to show that every set $\GBSR_k$ contains an infinite subset $\SF_k$ of SF formulas of degree\footnote{Notice that the notion of \emph{degree} used in the present paper differs from the notion of \emph{degree} used in \cite{Voigt2017}. The proof of Theorem~\ref{theorem:ComputationalComplexityGBSR} solely refers to the notions from the present paper.} $k$ such that there is a polynomial reduction from some $k$-\NEXPTIME-hard domino problem to the satisfiability problem for $\SF_k$ (see Lemma~20 in \cite{Voigt2017}).
	Hence, the satisfiability problem for $\GBSR_k$ is $k$-\NEXPTIME-hard as well.
\end{proof}
Every GBSR sentence with $n$ $\forall \exists$ quantifier alternations occurs in $\GBSR_n$ at the latest. 
It might occur earlier in the sequence $\GBSR_1 \subseteq \GBSR_2 \subseteq \ldots$, though. 
For instance, every relational monadic sentence $\varphi$ (in prenex form) without equality belongs to $\GBSR_1$ and thus also to every later set in the sequence, no matter what form $\varphi$'s quantifier prefix has.

%%%%%%%%%%%%%%%%%%%%%%%%%%%%%%%%%%%%%%%%%%%%%%%%%%%%%%%%%%%%%%%%%%%%%%%%%%%%%%%%%%%
%%%%%%%%%%%%%%%%%%%%%%%%%%%%%%%%%%%%%%%%%%%%%%%%%%%%%%%%%%%%%%%%%%%%%%%%%%%%%%%%%%%
%%%%%%%%%%%%%%%%%%%%%%%%%%%%%%%%%%%%%%%%%%%%%%%%%%%%%%%%%%%%%%%%%%%%%%%%%%%%%%%%%%%

%%%%%%%%%%%%%%%%%%%%%%%%%%%%%%%%%%%%%%%%%%%%%%%%%%%%%%%%%%%%%%%%%%%%%%%%%%%%%%%%
\subsection{A syntactic approach to decidability of GBSR-satisfiability}\label{section:SyntacticApproachToGBSRsat}
%%%%%%%%%%%%%%%%%%%%%%%%%%%%%%%%%%%%%%%%%%%%%%%%%%%%%%%%%%%%%%%%%%%%%%%%%%%%%%%%
The next lemma provides the key ingredient to show decidability of GBSR-satisfiability by a reduction to the satisfiability problem for the BSR fragment.
Moreover, it will help proving that relational GBSR without equality is closed under Craig--Lyndon interpolation.
\begin{lemma}\label{lemma:TransformationGBSR}
	Let $\varphi := \forall \vx_1 \exists \vy_1 \ldots \forall \vx_n \exists \vy_n.\,\psi$ be a GBSR sentence in standard form.
	There exists a quantifier-free first-order formula $\psi'(\vu, \vv)$ such that $\varphi' := \exists \vu\, \forall \vv.\, \psi'(\vu, \vv)$ is in standard form and semantically equivalent to $\varphi$ and all literals in $\varphi'$ also occur in $\varphi$ (modulo variable renaming).
\end{lemma}	
\begin{proof}[Proof sketch]
	The following transformations mainly use the standard laws of Boolean algebra and the miniscoping rules (Lemma~\ref{lemma:Miniscoping}) to (re-)transform $\psi$ into particular syntactic shapes. 
	Note that this does not change the set of literals occurring in the intermediate steps (modulo variable renaming), since we start from a formula in negation normal form restricted to the connectives $\wedge, \vee, \neg$. 
	We make use of the partition of the literals in $\psi$ into the sets $\tL_0, \tL_1, \ldots, \tL_n$ (defined right before Lemma~\ref{lemma:LiteralsDecomposition}) to group exactly those literals in each step of our transformation.
	
	To begin with, we transform the matrix $\psi$ into a disjunction of conjunctions of literals $\bigvee_i \psi_i$. 
	In addition, we rewrite every $\psi_i$ into $\psi_i = \tchi_{i,0}^{(1)} \wedge \ldots \wedge \tchi_{i,n}^{(1)}$. 
	Every $\tchi_{i,\ell}^{(1)}$ is a conjunction of literals and comprises exactly the literals from the $\psi_i$ which belong to $\tL_\ell$. 
	Moreover, by Lemma~\ref{lemma:LiteralsDecomposition}\ref{enum:LiteralsDecomposition:II} and \ref{enum:LiteralsDecomposition:III}, we know that $\vars(\tchi_{i,\ell}^{(1)}) \subseteq \vy_1 \cup \ldots \cup \vy_\ell \cup \vx_{\ell+1} \cup \ldots \cup \vx_n$.
	Using the rules of miniscoping, we move the existential quantifier block $\exists \vy_n$ inwards such that it binds the $\tchi_{i,n}^{(1)}$ alone. 
	The thus obtained sentence $\varphi''$ has the form $\forall\vx_1 \exists\vy_1 \ldots \forall\vx_n . \bigvee_i \tchi_{i,0}^{(1)} \wedge \ldots \wedge \tchi_{i,n-1}^{(1)} \wedge \exists\vy_n. \tchi_{i,n}^{(1)}$. 
	In further transformations, we treat the subformulas $\bigl(\exists\vy_n. \tchi_{i,n}^{(1)}\bigr)$ as indivisible units.
	
	Next, we transform the big disjunction in $\varphi''$ into a conjunction of disjunctions $\bigwedge_k \psi'_k$, group the disjunctions $\psi'_k$ into subformulas $\teta_{k,\ell}^{(1)}$ as before, and move the universal quantifier block $\forall \vx_n$ inwards. 
	Due to Lemma~\ref{lemma:LiteralsDecomposition}\ref{enum:LiteralsDecomposition:I}, we can split the quantifier block $\forall \vx_n$ so that universal quantifiers can be moved directly before the $\teta_{k,\ell}^{(1)}$. 
	The result is of the form $\forall\vx_1 \exists\vy_1 \ldots \forall\vx_{n-1} \exists\vy_{n-1} . \bigwedge_k \bigl(\forall(\vx_n \cap \tX_0).\, \teta_{k,0}^{(1)}\bigr) \vee \ldots \vee \bigl(\forall(\vx_n \cap \tX_{n-1}).\, \teta_{k,n-1}^{(1)}\bigr) \vee \teta_{k,n}^{(1)}$. 
	
	We reiterate this process until all quantifier blocks have been moved inwards as described.
	In the resulting formula we observe that no existential quantifier occurs within the scope of any universal quantifier.
	Using the miniscoping rules, we can move all quantifiers outwards again---existential quantifiers first---, renaming variables as necessary. In the end, we obtain a prenex formula of the form $\varphi' :=\;\exists \vu\, \forall \vv. \psi'$, where $\psi'$ is quantifier free and contains exclusively literals that are renamed variants of literals occurring in the original $\psi$.
\end{proof}

The just proven lemma shows that every GBSR sentence $\varphi$ is equivalent to some BSR sentence $\varphi' := \exists \vu\, \forall \vv.\, \psi'$. 
This immediately entails decidability of GBSR-satisfiability.
The number of leading existential quantifiers in BSR sentences induces an upper bound on the size of \emph{small models}---every satisfiable BSR sentences has such a small model.
One can adapt the methods applied in~\cite{Voigt2017} to facilitate the derivation of a tight upper on the number of leading existential quantifiers.
To this end, the notion of \emph{degree of interaction of existential variables} used in that paper needs to be extended so that it also covers the interaction of universally and existentially quantified variables caused by joint occurrences in atoms.

In Section~\ref{section:ModeltheoreticApproachToGBSRsat}, we present a different, a model-theoretic approach to deriving an upper bound on the size of small models.
In order to formulate this bound accurately, we introduce a related, yet complementary notion of \emph{degree} based on the interaction of universally quantified variables in atoms.

%%%%%%%%%%%%%%%%%%%%%%%%%%%%%%%%%%%%%%%%%%%%%%%%%%%%%%%%%%%%%%%%%%%%%%%%%%%%%%%%
%%%%%%%%%%%%%%%%%%%%%%%%%%%%%%%%%%%%%%%%%%%%%%%%%%%%%%%%%%%%%%%%%%%%%%%%%%%%%%%%

We conclude the present section by applying the above lemma to show that relational GBSR without equality is closed under Craig--Lyndon interpolation~\cite{Craig1957a, Lyndon1959}.
Hence, relational GBSR without equality additionally enjoys Beth's definability property, which is well-known to be a consequence of the Craig--Lyndon interpolation property (see, e.g., Chapter 20 in \cite{Boolos2002}).

Given a formula $\varphi$ that exclusively contains the connectives $\wedge, \vee, \neg$, we say that a predicate symbol $P$ \emph{occurs  positively in $\varphi$} if  there is an occurrence of some atom $P(\ldots)$ in $\varphi$ such that
the number of subformulas of $\varphi$ that contain this occurrence and have a negation sign as topmost connective is even.
Analogously, we say that a predicate symbol $P$ \emph{occurs  negatively in $\varphi$} if  there is an occurrence of some atom $P(\ldots)$ in $\varphi$ such that
the number of subformulas of $\varphi$ that contain this occurrence and have a negation sign as topmost connective is odd.

\begin{theorem}[Interpolation Theorem for GBSR]\label{theorem:InterpolationGBSR}
	Let $\varphi$ and $\psi$ be relational GBSR sentences in standard form without equality.
	If $\varphi \models \psi$, then there exists a relational BSR sentence $\chi$ without equality such that
	\begin{enumerate}[label=(\roman{*}), ref=(\roman{*})]
		\item $\varphi \models \chi$ and $\chi \models \psi$, and
		\item any predicate symbol $P$ occurs positively (negatively) in $\chi$ only if it occurs positively (negatively) in $\varphi$ and in $\psi$.
	\end{enumerate}
\end{theorem}

This is a consequence of Lemma~\ref{lemma:TransformationGBSR} combined with the following lemma.

%%%%%%%%%%%%%%%%%%%%%%%%%%%%%%%%%%%%%%%%%%%%%%%%%%%%%%%%%%%%%%%%%%%%%%%%%%%%%%%%
\begin{lemma}\label{lemma:LyndonInterpolationBSR}
	Let $\varphi$ and $\psi$ be relational BSR sentences without equality in which the only Boolean connectives are $\wedge, \vee, \neg$.
	If $\varphi \models \psi$, then there exists a relational BSR sentence $\chi$ without equality such that
	\begin{enumerate}[label=(\roman{*}), ref=(\roman{*})]
		\item\label{enum:LyndonInterpolationBSR:I} $\varphi \models \chi$ and $\chi \models \psi$, and
		\item\label{enum:LyndonInterpolationBSR:II} any predicate symbol $P$ occurs positively (negatively) in $\chi$ only if it occurs positively (negatively) in $\varphi$ and in $\psi$.
	\end{enumerate}
\end{lemma}
Since the proof of Lemma~\ref{lemma:LyndonInterpolationBSR} requires techniques that are largely unrelated to the rest of the present paper, the proof sketch has been moved to the appendix.

%%%%%%%%%%%%%%%%%%%%%%%%%%%%%%%%%%%%%%%%%%%%%%%%%%%%%%%%%%%%%%%%%%%%%%%%%%%%%%%%%%%
%%%%%%%%%%%%%%%%%%%%%%%%%%%%%%%%%%%%%%%%%%%%%%%%%%%%%%%%%%%%%%%%%%%%%%%%%%%%%%%%%%%
%%%%%%%%%%%%%%%%%%%%%%%%%%%%%%%%%%%%%%%%%%%%%%%%%%%%%%%%%%%%%%%%%%%%%%%%%%%%%%%%%%%

%%%%%%%%%%%%%%%%%%%%%%%%%%%%%%%%%%%%%%%%%%%%%%%%%%%%%%%%%%%%%%%%%%%%%%%%%%%%%%%%
\subsection{A model-theoretic approach to decidability of GBSR-satisfiability}\label{section:ModeltheoreticApproachToGBSRsat}
%%%%%%%%%%%%%%%%%%%%%%%%%%%%%%%%%%%%%%%%%%%%%%%%%%%%%%%%%%%%%%%%%%%%%%%%%%%%%%%%

In this section we investigate the \emph{finite character} of the dependency of existential variables $y$ on universal variables $x$ that appear earlier in the quantifier prefix $\ldots \forall x \ldots \exists y \ldots$ of a GBSR sentence.
This finite character will be made explicit with the following concepts:
	\emph{fingerprints} are sets of sets of \ldots sets of atoms that characterize certain tuples of domain elements by finite means; 
	\emph{fingerprint functions} $\mu$ assign fingerprints to tuples of domain elements; 
	\emph{$\mu$-uniform strategies} select domain elements for existentially quantified variables exclusively depending on the fingerprints of the domain elements that have been assigned to preceding universally quantified variables.

%%%%%%%%%%%%%%%%%%%%%%%%%%%%%%%%%%%%%%%%%%%%%%%%%%%%%%%%%%%%%%%%%%%%%%%%%%%%%%%%
Again, for the considerations in this section we fix a GBSR sentence $\varphi :=\, \forall \vx_1 \exists \vy_1 \ldots \forall \vx_n$ $\exists \vy_n.\psi$ in standard form in which $\psi$ is quantifier free.
Without loss of generality, we assume that $\varphi$ is relational and, hence, does not contain any function or constant symbols.
Let $\At$ denote the set of all atoms occurring in $\varphi$ and let $\vx := \vx_1 \cup \ldots \cup \vx_n$ and $\vy := \vy_1 \cup \ldots \cup \vy_n$.
By Definition~\ref{definition:GBSRaxiomatic}, we may assume that $\At$ can be partitioned into (possibly empty) sets $\At_0, \ldots, \At_n$ such
that Conditions~\ref{enum:GBSRaxiomatic:I} and~\ref{enum:GBSRaxiomatic:II} of Definition~\ref{definition:GBSRaxiomatic} are met.

Let $\cA$ be any structure over the vocabulary of $\psi$.
We next define the semantic equivalent of Skolem functions.
\begin{definition}
	A \emph{strategy $\sigma$} comprises a tuple of $n$ mappings $\<\sigma_1, \ldots, \sigma_n\>$ with signatures $\sigma_i : \fU_\cA^{|\vx_{1}|} \times \ldots \times \fU_\cA^{|\vx_i|} \to \fU_\cA^{|\vy_i|}$.
	A strategy $\sigma$ is \emph{satisfying for $\varphi$} if and only if \\
		\centerline{$\cA, [\vx_1 \Mapsto \vb_1, \ldots, \vx_n \Mapsto \vb_n, \vy_1 \Mapsto \sigma_1(\vb_1), \ldots, \vy_n \Mapsto \sigma_n(\vb_1, \ldots, \vb_n)] \models \psi$} 
	holds for every choice of tuples $\vb_1 \in \fU_\cA^{|\vx_1|}, \ldots, \vb_n \in \fU_\cA^{|\vx_n|}$.
\end{definition}

%%%%%%%%%%%%%%%%%%%%%%%%%%%%%%%%%%%%%%%%%%%%%%%%%%%%%%%%%%%%%%%%%%%%%%%%%%%%%%%%
\begin{definition}
	Let $\vb_1 \in \fU_\cA^{|\vx_1|}, \ldots, \vb_n \in \fU_\cA^{|\vx_n|}$ be tuples of domain elements. By $\out_{\At,\sigma}(\vb_1, \ldots, \vb_n)$ we denote the set $\bigl\{ A \in \At \bigm| \cA,[\vx_1 \Mapsto \vb_1, \ldots, \vx_n \Mapsto \vb_n, \vy_1 \Mapsto \sigma_1(\vb_1), \ldots,$ $\vy_n \Mapsto \sigma_n(\vb_1, \ldots, \vb_n)] \models A \bigr\}$, dubbed the \emph{outcome of $\vb_1, \ldots, \vb_n$ under $\sigma$ with respect to the atoms in $\At$}.
	By $\Out_{\At,\sigma}$ we denote the \emph{set of all possible outcomes of $\sigma$ with respect to the atoms in $\At$}, i.e.\ 
	$\Out_{\At,\sigma} := \bigl\{ \out_{\At,\sigma}(\vb_1, \ldots, \vb_n) \bigm| \vb_1 \in \fU_\cA^{|\vx_1|}, \ldots, \vb_n \in \fU_\cA^{|\vx_n|} \bigr\}$.
\end{definition}
The \emph{Boolean abstraction $\psi_\bool$ of $\psi$} is the propositional formula that results from $\psi$ if we conceive every atom $A$ in $\psi$ as propositional variable $p_A$. 
A subset $S \subseteq \At$ can be conceived as a valuation of $\psi_\bool$ by setting $S \models p_A$ if and only if $A \in S$.
Clearly, a strategy $\sigma$ is satisfying for $\varphi$ if and only if for every outcome $S \in \Out_{\At,\sigma}$ it holds $S \models \psi_\bool$.

%%%%%%%%%%%%%%%%%%%%%%%%%%%%%%%%%%%%%%%%%%%%%%%%%%%%%%%%%%%%%%%%%%%%%%%%%%%%%%%%
\begin{proposition}
	The structure $\cA$ is a model of $\varphi$ if and only if there exists a strategy $\sigma$ which is satisfying for $\varphi$.
\end{proposition}

If $\varphi$ is satisfied by $\cA$, we are interested in special satisfying strategies whose image only covers a finite portion of $\cA$'s domain.
Such a strategy induces a finite substructure of $\cA$ that also satisfies $\varphi$.
In order to find such strategies, we need to identify the key features of domain elements that make them distinguishable by the formula $\varphi$.
We express these features by the already mentioned \emph{fingerprints}.

We use $\fP$ to denote the power set operator, i.e.\ $\fP S$ stands for the set of all subsets of a given set $S$.
The iterated application of $\fP$ is denoted by $\fP^k$, meaning $\fP^0 S := S$ and $\fP^{k+1} S := \fP^k (\fP S)$ for every $k \geq 0$.

%%%%%%%%%%%%%%%%%%%%%%%%%%%%%%%%%%%%%%%%%%%%%%%%%%%%%%%%%%%%%%%%%%%%%%%%%%%%%%%%
\begin{definition}
	We define the family of \emph{fingerprint functions $\mu_{\ell, k}$} with $0 \leq \ell < k \leq n$ as follows
	\begin{description}
		\item $\mu_{\ell, n} : \fU_\cA^{|\vy_1|} \times \ldots \times \fU_\cA^{|\vy_\ell|} \times \fU_\cA^{|\vx_{\ell+1}|} \times \ldots \times \fU_\cA^{|\vx_n|} \to \fP \At_{\ell}$ 
			such that for all tuples $\va_1 \in \fU_\cA^{|\vy_1|}, \ldots, \va_{\ell} \in \fU_\cA^{|\vy_{\ell}|}, \vb_{\ell+1} \in \fU_\cA^{|\vx_{\ell+1}|}, \ldots, \vb_{n} \in \fU_\cA^{|\vx_{n}|}$ and every $A \in \At_{\ell}$ we have 
			$A \in \mu_{\ell,n}(\va_1, \ldots, \va_{\ell}, \vb_{\ell+1}, \ldots,$ $\vb_{n})$ if and only if 
			$\cA, [\vy_1 \Mapsto \va_1, \ldots, \vy_{\ell} \Mapsto \va_{\ell}, \vx_{\ell+1} \Mapsto \vb_{\ell+1}, \ldots,$ $\vx_{n} \Mapsto \vb_{n}] \models A$;
			
		\item $\mu_{\ell, n-1} :\fU_\cA^{|\vy_1|} \times \ldots \times \fU_\cA^{|\vy_\ell|} \times \fU_\cA^{|\vx_{\ell+1}|} \times \ldots \times \fU_\cA^{|\vx_{n-1}|} \to \fP^2 \At_{\ell}$
			such that for all tuples $\va_1 \in \fU_\cA^{|\vy_1|}, \ldots, \va_{\ell} \in \fU_\cA^{|\vy_{\ell}|}, \vb_{\ell+1} \in \fU_\cA^{|\vx_{\ell+1}|}, \ldots, \vb_{n-1} \in \fU_\cA^{|\vx_{n-1}|}$ and every $S \in \fP \At_{\ell}$ we have
			$S \in \mu_{\ell,n-1}(\va_1, \ldots, \va_{\ell}, \vb_{\ell+1}, \ldots, \vb_{n-1})$ if and only if there exists some $\vb_{n}$ such that $\mu_{\ell,n}(\va_1, \ldots, \va_{\ell},$ $\vb_{\ell+1}, \ldots, \vb_{n-1}, \vb_{n}) = S$;
			
		\item[] \qquad$\vdots$
		
		\item $\mu_{\ell, \ell+1} : \fU_\cA^{|\vy_1|} \times \ldots \times \fU_\cA^{|\vy_\ell|} \times \fU_\cA^{|\vx_{\ell+1}|} \to \fP^{n-\ell} \At_{\ell}$
			such that for all tuples $\va_1 \in \fU_\cA^{|\vy_1|}, \ldots, \va_{\ell} \in \fU_\cA^{|\vy_{\ell}|}, \vb_{\ell+1} \in \fU_\cA^{|\vx_{\ell+1}|}$ and every $S \in \fP^{n-\ell-1} \At_{\ell}$ we have
			$S \in \mu_{\ell,\ell+1}(\va_1, \ldots, \va_{\ell}, \vb_{\ell+1})$ if and only if there exists $\vb_{\ell+2}$ such that $\mu_{\ell,\ell+2}(\va_1, \ldots, \va_{\ell}, \vb_{\ell+1}, \vb_{\ell+2}) = S$.
	\end{description}
	We denote the \emph{image} of a fingerprint function $\mu_{\ell,k}$ under a strategy $\sigma = \<\sigma_1, \ldots, \sigma_n\>$ by \\
	\centerline{$\im_\sigma(\mu_{\ell, k}) := \bigl\{ \mu_{\ell,k}\bigl(\sigma_1(\vb_1), \ldots, \sigma_\ell(\vb_1, \ldots, \vb_\ell), \vb_{\ell+1}, \ldots, \vb_k\bigr) \bigm| \vb_1 \in \fU_\cA^{|\vx_1|}, \ldots, \vb_k \in \fU_\cA^{|\vx_k|} \bigr\}$.}
\end{definition}

%%%%%%%%%%%%%%%%%%%%%%%%%%%%%%%%%%%%%%%%%%%%%%%%%%%%%%%%%%%%%%%%%%%%%%%%%%%%%%%%
Having a suitable notion of fingerprints at hand, we next define strategies that induce finite substructures of $\cA$.
\begin{definition}
	A strategy $\sigma = \<\sigma_1, \ldots, \sigma_n\>$ is \emph{$\mu$-uniform} if for every $k$, $1 \leq k\leq n$, the following holds. 
	For all tuples $\vb_{1}, \vb'_{1} \in \fU_\cA^{|\vx_{1}|}, \ldots, \vb_k, \vb'_k \in \fU_\cA^{|\vx_k|}$ we have $\sigma_k(\vb_1, \ldots, \vb_k) = \sigma_k(\vb'_1, \ldots, \vb'_k)$ whenever for every $k'$, $1 \leq k' \leq k$, all of the following conditions are met:\\
	\centerline{$\mu_{0,k'}\bigl( \vb_1, \ldots, \vb_{k'}) = \mu_{0,k'}(\vb'_1, \ldots, \vb'_{k'} \bigr)$,}
	\centerline{$\mu_{1,k'}\bigl( \sigma_1(\vb_1), \vb_2, \ldots, \vb_{k'} \bigr) = \mu_{1,k'}\bigl( \sigma_1(\vb'_1), \vb'_2, \ldots, \vb'_{k'} \bigr)$,}
	\centerline{$\vdots$}
	\centerline{$\mu_{k'-1,k'}\bigl( \sigma_1(\vb_1), \ldots, \sigma_{k'-1}(\vb_1, \ldots, \vb_{k'-1}), \vb_{k'} \bigr) = \mu_{k'-1,k'}\bigl( \sigma_1(\vb'_1), \ldots, \sigma_{k'-1}(\vb'_1, \ldots, \vb'_{k'-1}), \vb'_{k'} \bigr)$.}
\end{definition}

%%%%%%%%%%%%%%%%%%%%%%%%%%%%%%%%%%%%%%%%%%%%%%%%%%%%%%%%%%%%%%%%%%%%%%%%%%%%%%%%
Intuitively, $\mu$-uniformity of a strategy $\sigma$ means that $\sigma$ reacts in the same way on inputs that have identical fingerprints.
The next lemma provides the key argument to infer the existence of some satisfying $\mu$-uniform strategy from the existence of any satisfying strategy.
\begin{lemma}\label{lemma:UniformStrategies}
	For every strategy $\sigma = \<\sigma_1, \ldots, \sigma_n\>$ there is a $\mu$-uniform strategy $\hsigma = \< \hsigma_1, \ldots, \hsigma_n \>$ such that $\Out_{\At,\hsigma} \subseteq \Out_{\At,\sigma}$.
\end{lemma}
\begin{proof}[Proof sketch]
	For $i = 1, \ldots, n$ we define $\fU_i$ as abbreviation of $\fU_\cA^{|\vx_1|} \times \ldots \times \fU_\cA^{|\vx_i|}$.
	We construct certain representatives $\alpha_{k, \<\bS_0^{(k)}, \ldots, \bS_{k-1}^{(k)}\>} \in \fU_k$ inductively as follows.
	The $\bS_i^{(k)}$ stand for sequences $S_{i,i+1}^{(k)} \ldots S_{i,k}^{(k)}$ of fingerprints satisfying $S_{i,k}^{(k)} \in S_{i,k-1}^{(k)} \in \ldots \in S_{i,i+1}^{(k)}$. \\
	Let $k=1$. 
		We partition
		$\fU_1$ into sets $\fU_{1,\<S_0^{(1)}\>}$ with $S_0^{(1)} \in \im_\sigma(\mu_{0,1})$ by setting
			$\fU_{1,\<S_0^{(1)}\>} := \bigl\{ \vb_1 \in \fU_\cA^{|\vx_1|} \mid \mu_{0,1}(\vb_1) = S_0^{(1)} \bigr\}$.
		We pick one representative $\alpha_{1,\<S_0^{(1)}\>}$ from each $\fU_{1,\<S_0^{(1)}\>}$. \\
	Let $k > 1$. 
		We construct subsets $\fU_{k,\<\bS_0^{(k)}, \ldots, \bS_{k-1}^{(k)}\>} \subseteq \fU_k$ with 
			$S_{0,j}^{(k)} \in \im_\sigma(\mu_{0,j}), 
				\ldots, 
				S_{k-1,j}^{(k)} \in \im_\sigma(\mu_{k-1,j})$, 
		for every $j \leq k$,		
		by setting $\fU_{k,\<\bS_0^{(k)}, \ldots, \bS_{k-1}^{(k)}\>} :=$\\
			$\bigl\{ \bigl\<\vc_1, \ldots, \vc_{k-1}, \vb_k\bigr\> \bigm|\; \text{$\vb_k \in \fU_\cA^{|\vx_k|}$ and there is some $\alpha_{k-1,\<\bS_0^{(k-1)}, \ldots, \bS_{k-2}^{(k-1)}\>} = \bigl\< \vc_1, \ldots, \vc_{k-1} \bigr\>$}$ \\
			\strut\hspace{2ex}				
				$\text{with $\vc_i \in \fU_\cA^{|\vx_i|}$, for every $i$, such that $\mu_{0,k}\bigl(\vc_1, \ldots, \vc_{k-1},  \vb_k\bigr) = S_0,\; \mu_{1,k}\bigl(\sigma_1(\vc_1), \vc_2, \ldots,$}$ \\
			\strut\hspace{2ex}				
			 	$\vc_{k-1}, \vb_k\bigr) = S_1, \ldots, \mu_{k-1,k}\bigl(\sigma_1(\vc_1), \ldots, \sigma_{k-1}(\vc_1, \ldots, \vc_{k-1}), \vb_k\bigr) = S_{k-1}$, and \\
			\strut\hspace{2ex}				
			 	$\bS_0^{(k)} = \bS_0^{(k-1)} S_0;\; \ldots ;\; \bS_{k-2}^{(k)} = \bS_{k-2}^{(k-1)} S_{k-2};\;\; \bS_{k-1}^{(k)} = S_{k-1}	\bigr\}$.
		
		We pick one representative $\alpha_{k,\<\bS_0^{(k)}, \ldots, \bS_{k-1}^{(k)}\>}$ from each nonempty $\fU_{k,\<\bS_0^{(k)}, \ldots, \bS_{k-1}^{(k)}\>}$.

	Having all the representatives $\alpha_{k,\<\bS_0^{(k)}, \ldots, \bS_{k-1}^{(k)}\>}$ at hand, we inductively construct $\hsigma$, starting from $\hsigma_1$ and going to $\hsigma_n$. \\
	Let $k=1$.
		For every $\vb_1 \in \fU_\cA^{|\vx_1|}$ we set $\hsigma_1(\vb_1) := \sigma_1(\alpha_{1,\<S_0\>})$, where $S_0 := \mu_{0,1}\bigl( \vb_1 \bigr)$.\\
	Let $k>1$.
		For all tuples $\vb_1 \in \fU_\cA^{|\vx_1|}, \ldots, \vb_k \in \fU_\cA^{|\vx_k|}$ we set $\hsigma_k(\vb_1, \ldots, \vb_k) := \sigma_k(\vc_1, \ldots, \vc_k)$, where 
		$\<\vc_1, \ldots, \vc_k\> : = \alpha_{k, \<\bS_0^{(k)}, \ldots, \bS_{k-1}^{(k)}\>}$ with $\vc_i \in \fU_\cA^{|\vx_i|}$, for every $i$, 
		and we have \\
		$S_{0,j}^{(k)} = \mu_{0,j}\bigl(\vb_1, \ldots, \vb_j\bigr)$ for every $j$, $0 < j \leq k$, \\
		$S_{1,j}^{(k)} = \mu_{1,j}\bigl(\hsigma_1(\vb_1), \vb_2, \ldots, \vb_j\bigr)$ for every $j$, $1 < j \leq k$, \\[-1ex]
		\strut\hspace{4ex} $\vdots$ \\[-1ex]
		$S_{k-2,j}^{(k)} = \mu_{k-2,j}\bigl(\hsigma_1(\vb_1), \ldots, \hsigma_{k-2}(\vb_1, \ldots, \vb_{k-2}), \vb_{k-1}, \ldots, \vb_j \bigr)$ for every $j$, $k-2 < j \leq k$, \\
		$S_{k-1,k}^{(k)} = \mu_{k-1,k}\bigl(\hsigma_1(\vb_1), \ldots, \hsigma_{k-1}(\vb_1, \ldots, \vb_{k-1}), \vb_k \bigr)$,\\
		if such an $\alpha_{k, \<\bS_0^{(k)}, \ldots, \bS_{k-1}^{(k)}\>}$ exists.

	\begin{description}
		\item\underline{Claim I:} For every $k$, $1 \leq k \leq n$, and all tuples $\vb_1 \in \fU_\cA^{|\vx_1|}, \ldots, \vb_k \in \fU_\cA^{|\vx_k|}$ there is a representative $\alpha_{k, \<\bS_0, \ldots, \bS_{k-1}\>}$ such that \\
			$S_{0,j} = \mu_{0,j}\bigl( \vb_1, \ldots, \vb_{j} \bigr)$ for every $j$, $0 < j \leq k$, \\
			$S_{1,j} = \mu_{1,j}\bigl(\hsigma_1(\vb_1), \vb_2, \ldots, \vb_{j}\bigr)$ for every $j$, $1 < j \leq k$, \\[-1ex]
			\strut\hspace{4ex} $\vdots$ \\[-1ex]
			$S_{k-2,j} = \mu_{k-2,j}\bigl(\hsigma_1(\vb_1), \ldots, \hsigma_{k-2}(\vb_1, \ldots, \vb_{k-2}), \vb_{k-1}, \ldots, \vb_j \bigr)$ for every $j$, $k-2 < j \leq k$, \\
			$S_{k-1,k} = \mu_{k-1,k}\bigl( \hsigma_1(\vb_1), \ldots, \hsigma_{k-1}(\vb_1, \ldots, \vb_{k-1}), \vb_k \bigr)$.

		\item\underline{Proof:} We proceed by induction on $k$. Details can be found in the appendix.	\hfill$\Diamond$				
	\end{description}
	By construction, $\hsigma$ is $\mu$-uniform.

	Now let $S \in \Out_{\At,\hsigma}$, i.e.\ there exist tuples $\vb_1 \in \fU_\cA^{|\vx_1|}, \ldots, \vb_n \in \fU_\cA^{|\vx_n|}$ such that
	$S = \out_{\At,\hsigma}(\vb_1, \ldots, \vb_n)$.
	We partition $S$ into sets $S_0 := S \cap \At_0, \ldots, S_n := S \cap \At_n$ and thus obtain the fingerprints 
		$S_\ell = \mu_{\ell,n}\bigl(\hsigma_{1}(\vb_1), \ldots, \hsigma_\ell(\vb_1, \ldots, \vb_\ell), \vb_{\ell+1}, \ldots, \vb_n \bigr) \subseteq \At_\ell$ for every $\ell$, $0 \leq \ell < n$.
	Claim I guarantees the existence of some representative $\alpha_{n,\<\bS'_0, \ldots, \bS'_{n-1}\>} = \<\vc_1, \ldots, \vc_n\>$, with $\vc_i \in \fU_\cA^{|\vx_i|}$, for every $i$, such that
		$S_\ell = \mu_{\ell,n}\bigl(\sigma_{1}(\vc_1), \ldots, \sigma_\ell(\vc_1, \ldots, \vc_\ell), \vc_{\ell+1}, \ldots, \vc_n \bigr)$ for every $\ell$, $0 \leq \ell < n$.
	
	Consider any $A \in \At$, and fix the $\ell$ for which $A \in \At_\ell$.
	We distinguish two cases.\\
	Suppose that $\ell < n$.
		The definition of $\alpha_{n, \<\bS'_0, \ldots, \bS'_{n-1}\>}$ and the fingerprint functions $\mu_{\ell, n}$ entail that $A \in S_\ell$ if and only if 
			$\cA,[\vy_1 \Mapsto \hsigma_1(\vb_1), \ldots, \vy_\ell \Mapsto \hsigma_\ell(\vb_1, \ldots, \vb_{\ell}), \vx_{\ell+1} \Mapsto \vb_{\ell+1}, \ldots, \vx_n \Mapsto \vb_n] \models A$
		if and only if 
			$\cA,[\vy_1 \Mapsto \sigma_1(\vc_1), \ldots, \vy_\ell \Mapsto \sigma_\ell(\vc_1, \ldots, \vc_{\ell}), \vx_{\ell+1} \Mapsto \vc_{\ell+1}, \ldots, \vx_n \Mapsto \vc_n] \models A$.\\
	In case of $\ell = n$, we have $A \in S_n$ if and only if 
			$\cA,[\vy_1 \Mapsto \hsigma_1(\vb_1), \ldots, \vy_n \Mapsto \hsigma_n(\vb_1, \ldots, \vb_n)] \models A$
		if and only if 
			$\cA,[\vy_1 \Mapsto \sigma_1(\vc_1), \ldots, \vy_n \Mapsto \sigma_n(\vc_1, \ldots, \vc_n)] \models A$.

	In both cases, we get $A \in \out_{\At,\hsigma}(\vb_1, \ldots, \vb_n)$ if and only if $A \in \out_{\At,\sigma}(\vc_1, \ldots, \vc_n)$.
	Consequently, we have $S = \out_{\At,\hsigma}(\vb_1, \ldots, \vb_n) = \out_{\At,\sigma}(\vc_1, \ldots,$ $\vc_n) \in \Out_{\At, \sigma}$. 
	
	Altogether, it follows that $\Out_{\At, \hsigma} \subseteq \Out_{\At, \sigma}$.
\end{proof}

%%%%%%%%%%%%%%%%%%%%%%%%%%%%%%%%%%%%%%%%%%%%%%%%%%%%%%%%%%%%%%%%%%%%%%%%%%%%%%%%
%%%%%%%%%%%%%%%%%%%%%%%%%%%%%%%%%%%%%%%%%%%%%%%%%%%%%%%%%%%%%%%%%%%%%%%%%%%%%%%%
\begin{corollary}\label{corollary:UniformStrategies}
	If there is a satisfying strategy $\sigma$ for $\varphi$, then there is also a $\mu$-uniform strategy $\hsigma$ that is satisfying for $\varphi$.
\end{corollary}
\begin{proof}
	Let $\sigma$ be a satisfying strategy for $\varphi$.
	By Lemma~\ref{lemma:UniformStrategies}, there is a $\mu$-uniform strategy $\hsigma$ such that for every $S \in \Out_{\At,\hsigma}$ we have $S \in \Out_{\At,\sigma}$.
	Since $\sigma$ is satisfying for $\varphi$, we get $S \models \psi_{\text{bool}}$ for every $S \in \Out_{\At,\sigma}$.
	Hence, we observe $S \models \psi_{\text{bool}}$ for every $S \in \Out_{\At,\hsigma} \subseteq \Out_{\At,\sigma}$.
	This means that $\hsigma$ is also satisfying for $\varphi$.
\end{proof}
In order to derive a small model property for GBSR sentences, it remains to show that any satisfying $\mu$-uniform strategy induces a finite substructure of $\cA$ that satisfies $\varphi$.
In what follows, we denote by $\kappa$ the smallest positive integer meeting the following condition.
For every $\At_i$, $0 \leq i < n$, there are at most $\kappa$ distinct indices $i+1 < j_1 < \ldots < j_\kappa \leq n$ such that $\vx_{j_\ell} \cap \vars(\At_i) \neq \emptyset$.
We call $\kappa$ the \emph{degree} of $\varphi$ and write $\degree_\varphi = \kappa$.
Notice that we have $0 \leq \kappa < n$.

To formulate the upper bound on the size of the domain, we need some notation for the \emph{tetration operation}. 
We define $\twoup{k}{m}$ inductively: $\twoup{0}{m} := m$ and $\twoup{k+1}{m} := 2^{\left(\twoup{k}{m}\right)}$.

%%%%%%%%%%%%%%%%%%%%%%%%%%%%%%%%%%%%%%%%%%%%%%%%%%%%%%%%%%%%%%%%%%%%%%%%%%%%%%%%
\begin{lemma}\label{lemma:FiniteModelsSemanticArgumentsGBSR}
	If there is a satisfying $\mu$-uniform strategy $\sigma$ for $\varphi$, then there is a model $\cB \models \varphi$ such that $\fU_\cB$ contains at most $n \cdot |\vy| \cdot \bigl( \twoup{\kappa+1}{|\text{\upshape{\At}}|} \bigr)^{n^2}$ domain elements.	
\end{lemma}
\begin{proof}[Proof sketch]~
	One can derive the following observations for all integers $\ell, k$, $0 \leq \ell < k < n$:
	(a) we have $|\im_\sigma(\mu_{\ell,k})| \leq 2^{|\im_\sigma(\mu_{\ell,k+1})|}$, and
	(b) for all tuples $\va_1, \ldots, \va_\ell, \vb_{\ell+1}, \ldots, \vb_k$ with $\va_i \in \fU_\cA^{|\vy_i|}$ and $\vb_i \in \fU_\cA^{|\vx_i|}$, for every $i$, we observe that, if $\vars(\At_\ell) \cap \vx_{k+1} = \emptyset$, then
				$\bigl| \mu_{\ell,k}(\va_1, \ldots, \va_\ell, \vb_{\ell+1}, \ldots, \vb_k) \bigr|$ $= 1$
			and, consequently,	
				$|\im_\sigma(\mu_{\ell,k})| \leq |\im_\sigma(\mu_{\ell,k+1})|$.
	
	Due to (a) and (b), our assumptions about $\kappa$ entail that \\
	(c) for all integers $\ell, k$ with $0 \leq \ell < k \leq n$ we have $|\im_\sigma(\mu_{\ell,k})| \leq \twoup{\kappa+1}{\At_\ell}$.
		
	%%%%%%%%%%%%%%%%%%%%%%%%%%%%%%%%%%%%%%%%%%%%%%%%%%%%%%%%%%	
	Let $\cT_\sigma$ be the \emph{target set} of $\sigma$, defined by $\cT_\sigma := \bigcup_{k = 1}^n \cT_k$, where \\
		$\cT_k := \bigl\{ \fa \in \fU_\cA \bigm|\; \text{there are tuples $\vb_1, \ldots, \vb_k$ with $\vb_i \in \fU_\cA^{|\vx_i|}$, for every $i$,}$ \\
		\strut\hspace{17ex}
			$\text{such that $\sigma_k(\vb_1, \ldots, \vb_k) = \<\ldots, \fa, \ldots\>$} \bigr\}$. \\
	Since $\sigma$ is $\mu$-uniform, $\cT_\sigma$ is a finite set.
	By definition of the fingerprint functions $\mu_{\ell, k}$, and by virtue of (c), we derive the following upper bounds, where we write $\him_\sigma(\mu_{i,j})$ to abbreviate $\im_\sigma(\mu_{i,i+1}) \times \im_\sigma(\mu_{i,i+2}) \times \ldots \times \im_\sigma(\mu_{i,j})$ for all $i,j$, $0\leq i < j \leq n$.\\
	\strut\hspace{1ex}\!
		$\bigl| \cT_1 \bigr| \;\;\leq\;\; |\vy_1| \cdot \bigl|\him_\sigma(\mu_{0,1})\bigr| \;\;\leq\;\; |\vy_1| \cdot \twoup{\kappa+1}{|\At|}$, \\
	\strut\hspace{1ex}
		$\bigl| \cT_2 \bigr| \;\;\leq\;\; |\vy_2| \cdot \bigl|\him_\sigma(\mu_{0,2}) \times \him_\sigma(\mu_{1,2})\bigr| \;\;\leq\;\; |\vy_2| \cdot \bigl( \twoup{\kappa+1}{|\At|} \bigr)^3$, \\[-1ex]
	\strut\hspace{8ex}
		$\vdots$ \\[-1ex]
	\strut\hspace{1ex}
		$\bigl| \cT_n \bigr| \;\;\leq\;\; |\vy_n| \cdot \bigl|\him_\sigma(\mu_{0,n}) \times \ldots \times \him_\sigma(\mu_{n-1,n})\bigr| \;\;\leq\;\; |\vy_n| \cdot \bigl( \twoup{\kappa+1}{|\At|} \bigr)^{n^2}$.\\
	Consequently, $\cT_\sigma$ contains at most $n \cdot |\vy| \cdot \bigl( \twoup{\kappa+1}{|\At|} \bigr)^{n^2}$ domain elements.

	Let $\varphi_\Sk$ be the result of exhaustive Skolemization of $\varphi$, i.e.\ every existentially quantified variable $y \in \vy_k$ in $\varphi$ is replaced by the Skolem function $f_y(\vx_1, \ldots, \vx_k)$.
	Clearly, $\sigma$ induces interpretations for all the Skolem functions $f_y$ such that $\cA$ can be extended to a model $\cA'$ of $\varphi_\Sk$ for which we have $f_y^{\cA'}(\vb_1, \ldots, \vb_k) \in \cT_k$ for every $f_y$ with $y \in \vy_k$ and all tuples $\vb_1, \ldots, \vb_k$.
	
	We define $\cB$ to be the substructure of $\cA'$ (with respect to $\varphi_\Sk$'s vocabulary) with universe $\fU_\cB := \cT_\sigma$.
	By the Substructure Lemma, $\cB$ satisfies $\varphi_\Sk$ and thus also the original $\varphi$.
	Moreover, we can bound the number of elements in $\cB$'s domain from above by $n \cdot |\vy| \cdot \bigl( \twoup{\kappa+1}{|\At|} \bigr)^{n^2}$.	
\end{proof}

%%%%%%%%%%%%%%%%%%%%%%%%%%%%%%%%%%%%%%%%%%%%%%%%%%%%%%%%%%%%%%%%%%%%%%%%%%%%%%%%

\begin{theorem}\label{theorem:GBSRsmallModelProperty}
	Every satisfiable GBSR sentence $\varphi := \forall \vx_1 \exists \vy_1 \ldots \forall \vx_n \exists \vy_n.\psi$ in standard form with quantifier-free $\psi$ and degree $\degree_\varphi \geq 1$ has a model with at most $\len(\varphi)^2 \cdot \bigl( \twoup{\degree_\varphi+1}{\len(\varphi)} \bigr)^{n^2}$ domain elements.
\end{theorem}

Note that the notion of degree used in this paper depends on a fixed partition of $\At$ into the sets $\At_0, \ldots, \At_n$.
Consider, for instance, a relational monadic formula $\varphi_\Mon$.
We could partition its atoms into two parts $\At_0, \At_n$, where $\At_0$ contains all atoms that contain a universally quantified variable and $\At_n$ comprises all atoms with existentially quantified variables. 
Clearly, $\At_0$ will cause the highest possible degree for $\varphi_\Mon$, since all universally quantified variables occur in atoms in $\At_0$.
To obtain a lower degree, we could partition $\At$ as follows.
For every $i$, $0 \leq i < n$, we set $\At_i := \bigl\{ P(x) \in \At \bigm| x \in \vx_{i+1} \bigr\}$, and the set $\At_n$ again contains all atoms with existentially quantified variables.
Clearly, this partition induces a degree $\degree_{\varphi_\Mon} = 0$ and thus a potentially much lower degree.

%%%%%%%%%%%%%%%%%%%%%%%%%%%%%%%%%%%%%%%%%%%%%%%%%%%%%%%%%%%%%%%%%%%%%%%%%%%%%%%%%%%
%%%%%%%%%%%%%%%%%%%%%%%%%%%%%%%%%%%%%%%%%%%%%%%%%%%%%%%%%%%%%%%%%%%%%%%%%%%%%%%%%%%
%%%%%%%%%%%%%%%%%%%%%%%%%%%%%%%%%%%%%%%%%%%%%%%%%%%%%%%%%%%%%%%%%%%%%%%%%%%%%%%%%%%

%%%%%%%%%%%%%%%%%%%%%%%%%%%%%%%%%%%%%%%%%%%%%%%%%%%%%%%%%%%%%%%%%%%%%%%%%%%%%%%%
\section{The generalized Ackermann fragment}\label{section:GeneralizedAckermann}
%%%%%%%%%%%%%%%%%%%%%%%%%%%%%%%%%%%%%%%%%%%%%%%%%%%%%%%%%%%%%%%%%%%%%%%%%%%%%%%%

In this section we generalize the Ackermann fragment (relational $\exists^* \forall \exists^*$ sentences without equality) in the same spirit as we have generalized the BSR fragment in Section~\ref{section:GeneralizedBSR}. We even go farther than Ackermann's original result, and, doing so, diverge in two directions:
In one direction we allow for unary function symbols to appear in formulas, even in a nested fashion, but do not allow equality.
In the other direction we allow equality but no non-constant function symbols.
In the former case, we devise an effective (un)satisfiability-preserving translation from the new fragment into the monadic fragment with unary function symbols.
In the latter case, we employ a result due to Ferm\"uller and Salzer~\cite{Fermuller1993b} to argue decidability of the satisfiability problem.

For the remainder of this section we fix a sentence $\varphi := \forall \vx_1 \exists \vy_1 \ldots \forall \vx_n \exists \vy_n. \psi$ in standard form without equality and without non-constant function symbols.
Before Lemma~\ref{lemma:GAFSkolemizedProperties} we do not have to pose any restrictions on the function symbols occurring in $\varphi$, and even the equality predicate would not do any harm.
Still, for the sake of clear definitions, we add the restrictions here and soften them where appropriate.

Let $\At$ be the set of all atoms occurring in $\varphi$ and let $\vx := \vx_1 \cup \ldots \cup \vx_n$ and $\vy := \vy_1 \cup \ldots \cup \vy_n$.
We define the \emph{index of a variable $v \in \vx \cup \vy$} by $\idx(v) := k$ if and only if $v \in \vx_k$ or $v \in \vy_k$.

%%%%%%%%%%%%%%%%%%%%%%%%%%%%%%%%%%%%%%%%%%%%%%%%%%%%%%%%%%%%%%%%%%%%%%%%%%%%%%%%
\begin{definition}[GAF, axiomatically]\label{definition:GAFaxiomatic}
	The sentence $\varphi$ belongs to the \emph{generalized Ackermann fragment (GAF)} if and only if
	we can partition $\At$ into sets $\At_0$ and $\At_x$, $x \in \vx$, such that the following conditions are met.
	\begin{enumerate}[label=(\alph{*}), ref=\alph{*}]
		\item\label{enum:definitionGAFaxiomatic:I} $\vars(\At_0) \cap \vx = \emptyset$.
		\item\label{enum:definitionGAFaxiomatic:II} For every $x \in \vx$ we have $\vars(\At_x) \cap \vx = \{x\}$.
		\item\label{enum:definitionGAFaxiomatic:III} For every $y \in \vy$ occurring in some $\At_x$ there are two (mutually exclusive) options
			\begin{enumerate}[label=(\ref{enum:definitionGAFaxiomatic:III}.\arabic{*}), ref=(\ref{enum:definitionGAFaxiomatic:III}.\arabic{*})]
				\item\label{enum:definitionGAFaxiomatic:III:I} either for every $\At_x$ in which $y$ occurs we have $\idx(y) < \idx(x)$,
				\item\label{enum:definitionGAFaxiomatic:III:II} or there is exactly one $\At_x$ in which $y$ occurs and $y$ does not occur in $\At_0$ and we have $\idx(y) \geq \idx(x)$.
			\end{enumerate}
	\end{enumerate}
\end{definition}
Intuitively speaking, Conditions~(\ref{enum:definitionGAFaxiomatic:II}) and~(\ref{enum:definitionGAFaxiomatic:III}) entail that, although a quantifier $\exists y$ may lie within the scope of two different quantifiers $\forall x$ and $\forall x'$ in $\varphi$, we can move $\exists y$ out of the scope of at least one of $\forall x$ and $\forall x'$ by suitable equivalence-preserving transformations. 
This will be the essence of the proof of Lemma~\ref{lemma:GAFTransformation}. 
In Example~\ref{example:GAFandGBSR} we have sketched the transformation procedure for the GAF sentence $\varphi_1$.

As in the case of GBSR, we complement the axiomatic definition of GAF with an alternative definition of an algorithmic flavor.
We shall see in Lemma~\ref{lemma:AlternativeDefinitionOfGAF} that both definitions are equivalent.
For the algorithmic point of view we need additional notation.

Let $\cG_\varphi := \<V,E\>$ be a directed graph such that $V := \vy$ and 
	$E := \bigl\{ \<y,y'\> \bigm| \idx(y) \leq \idx(y') \text{ and there is an atom $A$ in $\psi$ in which both $y$ and $y'$ occur}\bigr\}$.
For any variable $y \in \vy$ the \emph{upward closure} $\Cup_y$ is the smallest subset of $\vy$ such that $y \in \Cup_y$ and for every $y' \in \Cup_y$ the existence of an edge $\<y', y''\>$ in $\cG_\varphi$ entails $y'' \in \Cup_y$.
$\cL(\Cup_y)$ denotes the set of all literals in $\psi$, in which a variable from $\Cup_y$ occurs.
For any $x \in \vx$ let $\cL_x$ be the smallest set of literals such that
	(a) every literal in $\psi$ in which $x$ occurs belongs to $\cL_x$, and
	(b) for every $y \in \vars(\cL_x) \cap \vy$ with $\idx(y) \geq \idx(x)$ we have $\cL(\Cup_y) \subseteq \cL_x$.
Intuitively, $\cL_x$ collects all literals from $\psi$ which will remain in the scope of the quantifier $\forall x$ when we apply the laws of Boolean algebra and the rules of miniscoping to $\varphi$.
By $\cL_0$ we denote the set of all literals that occur in $\psi$ but in none of the $\cL_x$ with $x \in \vx$.
Moreover, we use the notation $Y_x := \vars(\cL_x) \cap \bigcup_{i\geq \idx(x)} \vy_i$.

%%%%%%%%%%%%%%%%%%%%%%%%%%%%%%%%%%%%%%%%%%%%%%%%%%%%%%%%%%%%%%%%%%%%%%%%%%%%%%%%
\begin{definition}[GAF, algorithmically]\label{definition:GAFalgorithmic}
	 The sentence $\varphi$ belongs to GAF if and only if
	\begin{enumerate}[label=(\roman{*}), ref=(\roman{*})]
		\item\label{enum:definitionGAFalgorithmic:II} every atom in $\psi$ contains at most one variable from $\vx$, and
		\item\label{enum:definitionGAFalgorithmic:III} for all distinct variables $x, x' \in \vx$ with $\idx(x) \leq \idx(x')$ and any variable $y \in \bigcup_{i\geq \idx(x)} \vy_i$ it holds $y \not\in \vars(\cL_x) \cap \vars(\cL_{x'})$.
	\end{enumerate}
\end{definition}

%%%%%%%%%%%%%%%%%%%%%%%%%%%%%%%%%%%%%%%%%%%%%%%%%%%%%%%%%%%%%%%%%%%%%%%%%%%%%%%%
\begin{lemma}\label{lemma:GAFPropertiesOne}
	If $\varphi$ satisfies Definition~\ref{definition:GAFalgorithmic}, then the following properties hold.
	\begin{enumerate}[label=(\roman{*}), ref=(\roman{*})]
		\item\label{enum:GAFPropertiesOne:I} 
			For all distinct $x,x' \in \vx$ it holds $\cL_x \cap \cL_{x'} = \emptyset$.
		\item\label{enum:GAFPropertiesOne:II} 
			For every $x \in \vx$ it holds $\vars(\cL_x) \cap \vx = \{x\}$.
		\item\label{enum:GAFPropertiesOne:III} 		
			For every $x \in \vx$ we have $Y_x \cap \vars(\cL_0) = \emptyset$. 
		\item\label{enum:GAFPropertiesOne:IV} 		
			For all distinct $x, x' \in \vx$ with $\idx(x) \leq \idx(x')$ we have $Y_x \cap (\vars(\cL_{x'})) = \emptyset$. 
	\end{enumerate}
\end{lemma}
\begin{proof}~
	\begin{description}
		\item{Ad \ref{enum:GAFPropertiesOne:I}:} 
			Suppose there are variables $x, x' \in \vx$ and there is a literal $L \in \cL_x \cap \cL_{x'}$. 
			$L$ must belong to $\cL(\Cup_y)$ for some variable $y \in \vy$ with $\idx(y) \geq \idx(x)$ or $\idx(y) \geq \idx(x')$, since otherwise it would hold $\{x,x'\} \subseteq \vars(L)$ which contradicts part \ref{enum:definitionGAFalgorithmic:II} of Definition~\ref{definition:GAFalgorithmic}. 
			This in turn means that some variable $y' \in \vy$ occurs in $L$ for which $\idx(y') \geq \idx(y)$.
			Hence, $y' \in \vars(L) \subseteq \bigl(\vars(\cL_x) \cap \vars(\cL_{x'})\bigr)$ with $\idx(y') \geq \idx(x)$ or $\idx(y') \geq \idx(x')$.
			This constitutes a contradiction to the second part of Definition~\ref{definition:GAFalgorithmic}.
			
		\item{Ad \ref{enum:GAFPropertiesOne:II}:} 
			This is a direct consequence of \ref{enum:GAFPropertiesOne:I} and the definition of $\cL_x$.

		\item{Ad \ref{enum:GAFPropertiesOne:III}:}
			From $y \in Y_x \subseteq \bigcup_{i\geq \idx(x)} \vy_i$ it follows that $\cL(\Cup_y) \subseteq \cL_x$, by definition of $\cL_x$.
			Suppose there is a $y \in Y_x \cap \vars(\cL_0)$, i.e.\ there is some $L \in \cL_0$ with $y \in \vars(L)$.
			Since $L \in \cL(\Cup_y) \subseteq \cL_x$, $L$ cannot occur in $\cL_0$, by definition of $\cL_0$.

		\item{Ad \ref{enum:GAFPropertiesOne:IV}:} 
			From $y \in Y_x \subseteq \bigcup_{i\geq \idx(x)} \vy_i$ it follows that $\cL(\Cup_y) \subseteq \cL_x$, by definition of $\cL_x$.
			Suppose there is some $y \in Y_x \cap \vars(\cL_{x'})$. Hence, there must be some literal $L \in \cL_{x'}$ in which $y$ occurs. But since $L$ belongs to $\cL(\Cup_y)$, we know that $L \in \cL_x$. This contradicts~\ref{enum:GAFPropertiesOne:I}.
			\qedhere
	\end{description}
\end{proof}
%%%%%%%%%%%%%%%%%%%%%%%%%%%%%%%%%%%%%%%%%%%%%%%%%%%%%%%%%%%%%%%%%%%%%%%%%%%%%%%%
We now have the right notions at hand to show equivalence of Definitions~\ref{definition:GAFaxiomatic} and~\ref{definition:GAFalgorithmic}.
\begin{lemma}\label{lemma:AlternativeDefinitionOfGAF}
	The sentence $\varphi$ satisfies Definition~\ref{definition:GAFaxiomatic} if and only if it satisfies Definition~\ref{definition:GAFalgorithmic}.
\end{lemma}
\begin{proof}
	Regarding the \emph{if}-direction, we define $\At_0$ to be the set of all atoms in $\cL_0$ and, analogously, for every $x \in \vx$ we define $\At_x$ to be the set of atoms in $\cL_x$.
	Condition~(\ref{enum:definitionGAFaxiomatic:I}) of Definition~\ref{definition:GAFaxiomatic} holds due to the definition of $\cL_0$ and the $\cL_x$.
	(\ref{enum:definitionGAFaxiomatic:II}) is an immediate consequence of Lemma~\ref{lemma:GAFPropertiesOne}\ref{enum:GAFPropertiesOne:II}.
	(\ref{enum:definitionGAFaxiomatic:III}) is entailed by Lemma~\ref{lemma:GAFPropertiesOne} \ref{enum:GAFPropertiesOne:III} together with \ref{enum:GAFPropertiesOne:IV}.
	More precisely, we observe that, if a variable $y$ occurs in a set $Y_x$, then condition \ref{enum:definitionGAFaxiomatic:III:II} applies.
	If, on the other hand, $y$ does not belong to any $Y_x$, then it satisfies condition \ref{enum:definitionGAFaxiomatic:III:I}.
	
	For the \emph{only if}-direction we argue as follows.
	Condition \ref{enum:definitionGAFalgorithmic:II} of Definition~\ref{definition:GAFalgorithmic} is certainly satisfied, if Condition~(\ref{enum:definitionGAFaxiomatic:I}) of Definition~\ref{definition:GAFaxiomatic} is met by $\varphi$.
	Consider any variable $y \in \vy$ to which Option \ref{enum:definitionGAFaxiomatic:III:II} applies and
	let $\At_{x_y}$ be the set in which $y$ occurs. 
	By definition of upward closure, every variable $y'$ in $\Cup$ must also be an Option-\ref{enum:definitionGAFaxiomatic:III:II} variable that occurs exclusively in $\At_{x_y}$.
	Hence, we conclude $\Cup_y \subseteq \vars(\At_{x_y})$. 
	This in turn entails that every atom occurring in $\cL(\Cup_y)$ belongs to $\At_{x_y}$.
	
	Consequently, for every $x \in \vx$ and for every literal $[\neg]A$ in $\cL_x$ we have $A \in \At_x$.
	Now consider two distinct variables $x, x' \in \vx$ with $\idx(x) \leq \idx(x')$ and let $y$ be some variable in $Y_x$.
	Since $y$ occurs in $\cL_x$ and thus also in $\At_x$, and since $\idx(y) \geq \idx(x)$, $y$ must be an Option-\ref{enum:definitionGAFaxiomatic:III:II} variable.
	Hence, it does not occur in any atom in $\At_{x'}$.
	And, due to the observation that for every literal $[\neg]A \in \cL_{x'}$ it also holds $A \in \At_{x'}$, $y$ cannot occur in any literal in $\cL_{x'}$.
	This shows that $\varphi$ satisfies Condition~\ref{enum:definitionGAFalgorithmic:III} of Definition~\ref{definition:GAFalgorithmic}.
\end{proof}

The name \emph{generalized} Ackermann fragment suggests that it properly contains the Ackermann fragment. 
This is confirmed by the next proposition.
But we also observe that the relational monadic fragment without equality is a proper subfragment of GAF.
Since neither the Ackermann fragment contains the monadic fragment nor vice versa, it is immediately clear that GAF constitutes a proper extension of both. Moreover, the sentence $\varphi_1$ treated in Example~\ref{example:GAFandGBSR} belongs to GAF but lies in neither of the two other  fragments.
\begin{proposition}\label{proposition:InclusionGAF}
	GAF properly contains the Ackermann fragment and the monadic first-order fragment, both without equality and non-constant function symbols.
\end{proposition}
\begin{proof}
	Let $\varphi' := \exists \vy\, \forall x \exists \vz. \psi'$ be an Ackermann sentence in standard form, which does neither contain equality nor any non-constant function symbols. 
	Any atom in $\varphi'$ contains at most one universally quantified variable, namely $x$.
	Let $\At_x$ be the set of all atoms occurring in $\varphi'$.
	Condition~(\ref{enum:definitionGAFaxiomatic:II}) of Definition~\ref{definition:GAFaxiomatic} is satisfied by $\At_x$.
	Moreover, every $y \in \vy$ is a $\ref{enum:definitionGAFaxiomatic:III:I}$ variable and every $z \in \vz$ is a $\ref{enum:definitionGAFaxiomatic:III:II}$ variable.
	Consequently, $\varphi'$ belongs to GAF.
	
	Let $\varphi'' := \forall \vx_1 \exists \vy_1 \ldots \forall \vx_n \exists \vy_n.\psi''$ be a monadic first-order sentence without equality that is in standard form. 
	For every $x \in \vx$ define $\At_x$ to be the set containing exactly the atoms in $\varphi''$ that contain $x$.
	Let $\At_0$ be the set of all atoms in $\varphi''$ that do not belong to any $\At_x$.
	Clearly, this partition of $\varphi''$'s atoms meets all the conditions posed in Definition~\ref{definition:GAFaxiomatic}.
	Hence, $\varphi''$ belongs to GAF.
\end{proof}

%%%%%%%%%%%%%%%%%%%%%%%%%%%%%%%%%%%%%%%%%%%%%%%%%%%%%%%%%%%%%%%%%%%%%%%%%%%%%%%%%%%
%%%%%%%%%%%%%%%%%%%%%%%%%%%%%%%%%%%%%%%%%%%%%%%%%%%%%%%%%%%%%%%%%%%%%%%%%%%%%%%%%%%
%%%%%%%%%%%%%%%%%%%%%%%%%%%%%%%%%%%%%%%%%%%%%%%%%%%%%%%%%%%%%%%%%%%%%%%%%%%%%%%%%%%

We shall now work towards showing that the satisfiability problem for GAF sentences (\emph{GAF-satisfiability}) is decidable, even if we extend it with equality or unary function symbols.
But first, we need additional notation.

For every $x \in \vx$ we refine the set $\cL_x$ into subsets $\cL_{x,0}, \cL_{x,\idx(x)}, \cL_{x,\idx(x)+1},\ldots, \cL_{x,n}$: \\
	\begin{tabular}{r@{\hspace{1ex}}c@{\hspace{1ex}}l}
		$\cL_{x,n}$	&$:=$		&$\bigcup_{y \in \vars(\cL_x) \cap \vy_n} \cL(\Cup_y)$,\\
		$\cL_{x,k}$	&$:=$		&$\bigcup_{y \in \vars(\cL_x) \cap \vy_{k}} \cL(\Cup_y) \setminus \bigcup_{\ell > k} \cL_{x,\ell}$ for every $k$ satisfying $\idx(x) \leq k < n$, and\\
		$\cL_{x,0}$ 	&$:=$ 	&$\cL_x \setminus \bigcup_{\ell \geq \idx(x)} \cL_{x,\ell}$.
	\end{tabular}\\	
Similarly, we define $Y_{x,k} := \vars(\cL_{x,k}) \cap Y_x$ for every $k$, $0\leq k\leq n$.

%%%%%%%%%%%%%%%%%%%%%%%%%%%%%%%%%%%%%%%%%%%%%%%%%%%%%%%%%%%%%%%%%%%%%%%%%%%%%%%%
\begin{lemma}\label{lemma:GAFPropertiesTwo}
	If $\varphi$ belongs to GAF, then the following properties hold for every $x \in \vx$.
	\begin{enumerate}[label=(\roman{*}), ref=(\roman{*})]
		\item\label{enum:GAFPropertiesTwo:I}
			For every $k$ we have $\cL_{x,k} \subseteq \cL_x$.
		\item\label{enum:GAFPropertiesTwo:II}
			For all distinct $k, \ell$ we have $\cL_{x,k} \cap \cL_{x,\ell} = \emptyset$.
		\item\label{enum:GAFPropertiesTwo:IV}
			For every $k > 0$ it holds $\vars(\cL_{x,k}) \cap \vy \subseteq \bigcup_{i \leq k} \vy_i$.
		\item\label{enum:GAFPropertiesTwo:V}
			We have $\vars(\cL_{x,0}) \cap \vy \subseteq \bigcup_{i < \idx(x)} \vy_i$.
	\end{enumerate}
\end{lemma}
\begin{proof}~
	\begin{description}
		\item{Ad \ref{enum:GAFPropertiesTwo:I}:} 
			The claim holds by definition of $\cL_x$ and $\cL_{x,k}$.
			
		\item{Ad \ref{enum:GAFPropertiesTwo:II}:} 
			Assume, without loss of generality, that $k < \ell$. 
			By definition of $\cL_{x,k}$ and by \ref{enum:GAFPropertiesTwo:I}, we have $\cL_{x,k} \subseteq \cL_x \setminus \cL_{x,\ell}$.

		\item{Ad \ref{enum:GAFPropertiesTwo:IV}:}
			Suppose there is some index $\ell > k$ and a variable $y \in \vy_\ell$ which occurs in some literal $L \in \cL_{x,k}$.
			By definition of $\cL(\Cup_y)$, we observe $L \in \cL(\Cup_y)$.
			But, by definition of $\cL_{x,\ell}$, we conclude $\cL(\Cup_y) \setminus \bigcup_{j > \ell} \cL_{x,j} \subseteq \cL_{x,\ell}$. Hence, it either holds $L \in \cL_{x,\ell}$ or $L \in \bigcup_{j > \ell} \cL_{x,j}$.
			This yields a contradiction, because $\cL_{x,k} \subseteq \cL_x \setminus  \bigcup_{j \geq \ell} \cL_{x,j}$ and thus $\cL_{x,k}$ cannot contain $L$. 

		\item{Ad \ref{enum:GAFPropertiesTwo:V}:}
			Suppose there is some index $\ell \geq k$ and a variable $y \in \vy_\ell$ which occurs in some literal $L \in \cL_{x,0}$.
			By definition of $\cL(\Cup_y)$, we observe $L \in \cL(\Cup_y)$.
			But, by definition of $\cL_{x,\ell}$, we conclude $\cL(\Cup_y) \setminus \bigcup_{j > \ell} \cL_{x,j} \subseteq \cL_{x,\ell}$. Hence, it either holds $L \in \cL_{x,\ell}$ or $L \in \bigcup_{j > \ell} \cL_{x,j}$.
			This yields a contradiction, because $\cL_{x,0} \subseteq \cL_x \setminus  \bigcup_{j \geq \ell} \cL_{x,j}$ and thus $\cL_{x,0}$ cannot contain $L$. 
			\qedhere
	\end{description}	
\end{proof}

%%%%%%%%%%%%%%%%%%%%%%%%%%%%%%%%%%%%%%%%%%%%%%%%%%%%%%%%%%%%%%%%%%%%%%%%%%%%%%%%
The next lemma provides the key ingredient to show decidability of GAF-satisfiability.
Similar to the transformation described in Lemma~\ref{lemma:TransformationGBSR}, the following lemma describes a transformation of GAF sentences into a nicer syntactic form. 
However, this transformation constitutes only the first stage of the decidability proof.
\begin{lemma}\label{lemma:GAFTransformation}
	If $\varphi$ belongs to GAF, we can effectively construct an equivalent sentence $\varphi'$ in standard form, in which every subformula lies within the scope of at most one universal quantifier.
	Moreover, all literals in $\varphi'$ already occur in $\varphi$ (modulo variable renaming).
\end{lemma}
\begin{proof}[Proof sketch]
	Similarly to the proof of Lemma~\ref{lemma:TransformationGBSR}, we (re-)transform parts of $\varphi$ repeatedly into a disjunction of conjunctions (or a conjunction of disjunctions) of subformulas which we treat as indivisible units. 
	The literals and indivisible units in the respective conjunctions (disjunctions) will be grouped in accordance with the sets $\cL_0, \cL_x$, and $\cL_{x,\idx(x)}, \ldots, \cL_{x,n}$, where needed. 
	For this purpose, it is important to note that Lemma~\ref{lemma:GAFPropertiesOne}\ref{enum:GAFPropertiesOne:I} and the definition of $\cL_0$ entail that $\cL_0$ together with the sets $\cL_x$ partition the set of all literals occurring in $\varphi$.
	Moreover, every $\cL_x$ is partitioned by the sets $\cL_{x,0}, \cL_{x, \idx(x)}, \ldots, \cL_{x,n}$, by virtue of Lemma~\ref{lemma:GAFPropertiesTwo}\ref{enum:GAFPropertiesTwo:I}, \ref{enum:GAFPropertiesTwo:II} and the definition of $\cL_{x,0}$.
	
	At the beginning, we transform $\psi$ into a disjunction of conjunctions of literals $\bigvee_i \psi_i$. 
	At this point, we move the existential quantifier block $\exists \vy_n$ inwards.
	Lemmas~\ref{lemma:GAFPropertiesOne} and~\ref{lemma:GAFPropertiesTwo} guarantee that the quantifiers from this quantifier block can be distributed over the constituents of the $\psi_i$ in a beneficial way.
	The thus obtained sentence $\varphi''$ has the form 
	 	$\forall \vx_1 \exists \vy_1 \ldots \forall \vx_n. \bigvee_{i} \bigl( \exists \vy_n. \chi_{i,0}^{(1)} \bigr)
			\wedge \bigwedge_{k=1}^{n} \bigwedge _{x \in \vx_k}$ $\bigl( \chi_{i,x,0}^{(1)}
			\wedge \bigl( \bigwedge_{j=\idx(x)}^{n-1} \chi_{i,x,j}^{(1)} \bigr)
			\wedge \exists(\vy_n \cap Y_{x,n}). \chi_{i,x,n}^{(1)} \bigr)$,
	where $\chi_{i,0}^{(1)}$ comprises all literals in $\psi_i$ which belong to $\cL_0$, and for every $k$ the $\chi_{i,x,k}^{(1)}$ group the literals which occur in $\psi_i$ and belong to $\cL_{x,k}$, respectively.
	
	Next, we transform the big disjunction in $\varphi''$ into a conjunction of disjunctions $\bigwedge_i \psi'_i$,
	and move the universal quantifier block $\forall \vx_n$ inwards. 
	The resulting formula has the shape $\forall \vx_1 \exists \vy_1 \ldots \forall \vx_{n-1}$ $\exists \vy_{n-1}. \bigwedge_{i} \eta_{i,0}^{(1)} \vee \bigl( \bigvee_{k=1}^{n-1} \bigvee _{x \in \vx_k} \eta_{i,x}^{(1)} \bigr) \vee \bigvee _{x \in \vx_n} \forall x.\, \eta_{i,x}^{(1)}$,
	where grouping of the constituents of each $\psi'_i$ is similar to what we have done above, but this time in accordance with the more coarse-grained sets $\cL_0$ and $\cL_x$. 

	We reiterate the described process until all the quantifiers have been moved inwards in the outlined way.
	The final result of this transformation is the sought $\varphi'$ and it does not contain any nested occurrences of universal quantifiers.
\end{proof}

The above proof still works if $\varphi$	 contains the equality predicate or non-constant function symbols.
Moreover, the sentence $\varphi'$ in Lemma~\ref{lemma:GAFTransformation} has a very particular shape. 
For one part, it does not contain nested universal quantifiers. 
In addition, the vocabulary in $\varphi'$ is identical to the vocabulary of the original $\varphi$. 
If, for instance, $\varphi$ does not contain function symbols of arity larger then one, then the same holds true for $\varphi'$.
These two properties have implications for the outcome $\varphi_\Sk$ of Skolemizing $\varphi'$.
%%%%%%%%%%%%%%%%%%%%%%%%%%%%%%%%%%%%%%%%%%%%%%%%%%%%%%%%%%%%%%%%%%%%%%%%%%%%%%%%
\begin{lemma}\label{lemma:GAFSkolemizedProperties}
	The Skolemized variant $\varphi_{\Sk}$ of $\varphi'$ satisfies the following properties:
	\begin{enumerate}[label=(\roman{*}), ref=(\roman{*})]
		\item\label{enum:propositionGAFSkolemizedProperties:I} $\varphi_{\Sk}$ does not contain any function symbol of arity larger than one.
		\item\label{enum:propositionGAFSkolemizedProperties:II} Every atom $A$ in $\varphi_{\Sk}$ is either ground or contains exactly one variable.
	\end{enumerate}	
\end{lemma}
\begin{proof}
	Property~\ref{enum:propositionGAFSkolemizedProperties:I} is a direct consequence of the fact that $\varphi'$ does neither contain free variables nor nested occurrences of universal quantifiers.
	Concerning \ref{enum:propositionGAFSkolemizedProperties:II}, consider an atom $A$ in $\varphi_\Sk$ and suppose that $A$ is not ground. 
	Since $\varphi_\Sk$ resulted from Skolemization, it can only contain universally quantified variables. 
	Suppose $A$ contains two distinct variables $x, x'$. 
	Because of $\varphi_\Sk$ being closed, $A$ must lie within the scope of two distinct universal quantifiers $\forall x$ and $\forall x'$. But this contradicts Lemma~\ref{lemma:GAFTransformation}.
\end{proof}

%%%%%%%%%%%%%%%%%%%%%%%%%%%%%%%%%%%%%%%%%%%%%%%%%%%%%%%%%%%%%%%%%%%%%%%%%%%%%%%%
These observations lead to the first decidability result with respect to GAF-satisfiability.
\begin{theorem}\label{theorem:DecidabilityGAFwithEquality}
	Satisfiability of a given GAF sentence $\varphi$ is decidable, even if $\varphi$ contains equality (but no non-constant function symbols).
\end{theorem}
This result follows from Theorem~2 in~\cite{Fermuller1993b}, where Ferm\"uller and Salzer show that the satisfiability for a clausal fragment called $\cA^=$ is decidable.
Roughly speaking, in any clause set in $\cA^=$ every literal of the form $P(t_1, \ldots, t_m)$ or $t_1 \approx t_2$ contains at most one variable (possibly with multiple occurrences) and  each $t_i$ is either a constant symbol, a variable, or a term of the form $f(v)$ for some variable $v$.
Clause sets corresponding to GAF sentences with equality but without non-constant function symbols, fall exactly into the syntactic category of $\cA^=$.
In order to see this more clearly, we can strengthen \ref{enum:propositionGAFSkolemizedProperties:I} in Lemma~\ref{lemma:GAFSkolemizedProperties} to the following property:
	\ref{enum:propositionGAFSkolemizedProperties:I}$'$ Every term $t$ in $\varphi_\Sk$ is either a constant symbol, a variable, or of the form $t = f(v)$ for some variable $v$.

%%%%%%%%%%%%%%%%%%%%%%%%%%%%%%%%%%%%%%%%%%%%%%%%%%%%%%%%%%%%%%%%%%%%%%%%%%%%%%%%
On the other hand, we can show decidability of the satisfiability problem for GAF sentences $\varphi$ without equality in which we allow unary function symbols to occur in an arbitrarily nested fashion.
Given the Skolemized variant $\varphi_\Sk$ of the result $\varphi'$ of Lemma~\ref{lemma:GAFTransformation}, the following lemma entails that we can effectively construct a sentence  $\varphi'_\Sk$ which is equisatisfiable to $\varphi_\Sk$ and belongs to the \emph{full monadic fragment} (monadic first-order sentences with unary function symbols but without equality). 
Decidability of the satisfiability problem for the full monadic fragment has first been shown by L\"ob and Gurevich~\cite{Lob1967, Gurevich1969}.
\begin{lemma}\label{lemma:TransformToFullMonadic:UnaryFunctions}
	Let $\varphi_\Sk$ be a first-order sentence without equality, in which
		\begin{enumerate}[label=(\alph{*}), ref=(\alph{*})]
			\item function symbol of arity larger than one do not occur,
			\item existential quantifiers do not occur, and
			\item every atom contains at most one variable (possibly with multiple occurrences).
		\end{enumerate}
	Then we can effectively construct an equisatisfiable sentence $\varphi'_\Sk$ in which all occurring predicate symbols are unary, and no other constant symbols and function symbols appear than the ones that occur in $\varphi_\Sk$. Moreover, the length of $\varphi'_\Sk$ lies in $\cO\bigl(\len(\varphi_\Sk)^4\bigr)$.
\end{lemma}
Section~\ref{section:proofGAFtoMonadic} is devoted to the proof of this lemma.
%
%%%%%%%%%%%%%%%%%%%%%%%%%%%%%%%%%%%%%%%%%%%%%%%%%%%%%%%%%%%%%%%%%%%%%%%%%%%%%%%%
Putting Lemmas~\ref{lemma:GAFTransformation}, \ref{lemma:GAFSkolemizedProperties} and \ref{lemma:TransformToFullMonadic:UnaryFunctions} together, we can prove decidability of GAF-satisfiability.
\begin{theorem}\label{theorem:DecidabilityGAFwithUnaryFunctions}
	Satisfiability of a given GAF sentence $\varphi$ without equality is decidable, even if $\varphi$ contains arbitrarily nested unary function symbols.
\end{theorem}

%%%%%%%%%%%%%%%%%%%%%%%%%%%%%%%%%%%%%%%%%%%%%%%%%%%%%%%%%%%%%%%%%%%%%%%%%%%%%%%%%
\subsection{Translating GAF sentences into monadic first-order sentences}\label{section:proofGAFtoMonadic}
%%%%%%%%%%%%%%%%%%%%%%%%%%%%%%%%%%%%%%%%%%%%%%%%%%%%%%%%%%%%%%%%%%%%%%%%%%%%%%%%%
We start with an auxiliary lemma.

Two atoms $A$ and $B$ are considered to be \emph{variable disjoint} if and only if $\vars(A) \cap \vars(B) = \emptyset$.
We say $A$ is \emph{more general} than $B$, denoted $A \lesssim B$, if and only if there is a substitution $\theta$ such that $A\theta = B$.
Moreover, we write $A\simeq B$ if and only if $A\lesssim B$ and $A \gtrsim B$.

A substitution $\theta$ for which $A\theta$ equals $B\theta$ is called a \emph{unifier of $A$ and $B$}.
If such a unifier exists, then we say \emph{$A$ and $B$ are unifiable}.
A unifier $\theta$ of $A$ and $B$ is considered to be \emph{most general} if and only if for every unifier $\theta'$ of $A$ and $B$ it holds $A\theta \lesssim A\theta'$. 
As usual, we shall abbreviate the term \emph{most general unifier} with the acronym \emph{mgu}.

%%%%%%%%%%%%%%%%%%%%%%%%%%%%%%%%%%%%%%%%%%%%%%%%%%%%%%%%%%%%%%%%%%%%%%%%%%%%%%%%
\begin{lemma}\label{lemma:UnifiabilityForSingleVariableAtoms}
	Let $A$ and $B$ be two variable-disjoint atoms, and assume each of them contains at most one variable (possibly with multiple occurrences).
	If $A$ and $B$ are unifiable, and $\theta$ is an mgu of the two, then either $A\theta \simeq A$ or $B\theta \simeq B$ or $A\theta = B\theta$ is ground.
\end{lemma}
\begin{proof}
	Suppose $A\theta \not\simeq A$ and $B\theta \not\simeq B$.
	Hence, there are distinct variables $x_1, x_2$ such that $\vars(A) = \{x_1\}$ and $\vars(B) = \{x_2\}$.	
	Because of $A\theta = B\theta$, $A\theta \not\simeq A$, and $B\theta \not\simeq B$, and since $\theta$ is most general, there must be two term positions $\pi_1, \pi_2$ such that
	\begin{itemize}
		\item $A|_{\pi_1} = x_1$ and $B|_{\pi_1} = t_2 \neq x_2$, and
		\item $A|_{\pi_2} = t_1 \neq x_1$ and $B|_{\pi_2} = x_2$.
	\end{itemize}
	Consequently, we know $x_1\theta = t_2\theta$ and $x_2\theta = t_1\theta$.
	
	If $A\theta = B\theta$ were not ground, then $\vars(t_1) = \{x_1\}$ and $\vars(t_2) = \{x_2\}$ would hold true.
	While $x_1\theta = t_2\theta$ thus entails that the term depth of $x_1\theta$ is strictly larger than that of $x_2\theta$, the observation of $x_2\theta = t_1\theta$ implies that the term depth of $x_1\theta$ is strictly smaller then that of $x_2\theta$. This contradiction means that $A\theta = B\theta$ must be ground.
\end{proof}

%%%%%%%%%%%%%%%%%%%%%%%%%%%%%%%%%%%%%%%%%%%%%%%%%%%%%%%%%%%%%%%%%%%%%%%%%%%%%%%%
%%%%%%%%%%%%%%%%%%%%%%%%%%%%%%%%%%%%%%%%%%%%%%%%%%%%%%%%%%%%%%%%%%%%%%%%%%%%%%%%

In order to show decidability of $\exists^* \forall \exists^*$ sentences, Ackermann translated $\exists^*\forall\exists^*$ sentences into equisatisfiable monadic ones \cite{Ackermann1954}. F\"urer adopted Ackermann's method to give an upper bound on the complexity of the decision problem for the Ackermann Fragment \cite{Furer1981}. We shall employ a generalization of F\"urer's reduction approach to prove Lemma~\ref{lemma:TransformToFullMonadic:UnaryFunctions}.

	Without loss of generality, we assume that $\varphi_\Sk$ contains at least the constant symbol $d$.
	Let $\At$ be the set of all atoms occurring in $\varphi_\Sk$.
	Consider the set $\At'$ which we define to be the smallest set of variable-disjoint atoms such that
		(a) for every $A\in \At$ there is some $B \in \At'$ such that $B \simeq A$, 
		(b) for all $A, B \in \At'$, for which there is an mgu $\theta$, we find some atom $C \in \At'$ such that $C \simeq A\theta = B\theta$, and
		(c) for all $A, B \in \At'$ we have $A \not\simeq B$.

	By Lemma~\ref{lemma:UnifiabilityForSingleVariableAtoms}, the set $\At'$ is finite. More precisely: $\At'$ contains at most $\tfrac{|\At| \cdot (|\At|-1)}{2} + |\At| \leq |\At|^2$ elements.
	Let $A_1, \ldots, A_q$ be an enumeration of all the atoms in $\At'$, and let $P_1, \ldots, P_q$ be distinct unary predicate symbols which do not occur in $\varphi_\Sk$.
	We construct the sentence $\varphi_\Mon$ from $\varphi_\Sk$ as follows:
		(a) replace every occurrence of a non-ground atom $B(x)$ in $\varphi_\Sk$ with the atom $P_i(x)$ which corresponds to the (unique) atom $A_i \in \At'$ with $A_i \simeq B$, and
		(b) replace every occurrence of a ground atom $B'$ in $\varphi_\Sk$ with $P_j(d)$ corresponding to the unique $A_j \in \At'$ with $A_j = B'$.
	
	Consider two distinct atoms $A_i, A_j \in \At'$. If there is a unifier $\theta$ of $A_i$ and $A_j$, we must make sure that the instances of $P_i(x)$ and $P_j(x')$ corresponding to $A_i\theta$ and $A_j\theta$ are interpreted in the same way by any model of $\varphi_\Mon$.
	In order to do so, we define the sets $\Psi, \Psi'$ of formulas as follows. 
	Let $x_*$ be a fresh variable which does not occur in $\varphi_\Mon$, and let $\tau_*$ be the substitution mapping all variables to $x_*$.
	We set \\
		\centerline{$\Psi := \bigl\{ P_i\bigl((x\theta)\tau_*\bigr) \leftrightarrow P_j\bigl(x_*\bigr) \!\bigm|\, 
				\text{$A_i, A_j\in \At'$ are distinct, $A_i \lesssim A_j$, $\vars(A_i) = \{x\}$, and $A_i\theta = A_j$} \bigr\}$.}
	In addition, we define $\Psi' := \bigl\{ P_j(x_*) \leftrightarrow P_j(d) \bigm| \text{$A_j \in \At'$ is ground} \bigr\}$.
			
	We now set $\varphi'_\Sk := \varphi_\Mon \wedge \forall x_*. \bigwedge_{\psi \in \Psi \cup \Psi'} \psi$. 
	By construction, $\varphi'_\Sk$ exclusively contains predicate symbols which are unary.
	Its constituent $\varphi_\Mon$ may contain the constant symbol $d$ but no non-constant function symbols.
	The formulas in $\Psi \cup \Psi'$, on the other hand, may contain non-constant function symbols. 
	However, all of these symbols have already occurred in  $\varphi_\Sk$.
	
	By virtue of Lemma~\ref{lemma:UnifiabilityForSingleVariableAtoms}, the number of formulas in $\Psi$ is at most $|\At| \cdot |\At'| \leq |\At|^3$. 
	The length of the formulas in $\Psi$ is upper bounded by the length of $x\theta$ in $\Psi$'s definition plus some constant value, and thus can exceed the length of the longest atom in $\varphi_\Sk$ only by this constant value.
	The cardinality of $\Psi'$ is upper bounded by $|\At'| \leq |\At|^2$ and the length of the formulas therein is constant.
	All in all, the length of $\varphi'_\Sk$ lies in $\cO\bigl(\len(\varphi_\Sk)^4\bigr)$.
	
	In what follows, we tacitly assume that the signatures underlying $\varphi_\Sk$ and $\varphi'_\Sk$ share the same constant symbols and function symbols---namely, the ones occurring in $\varphi_\Sk$. Consequently, when we refer to Herbrand structures with respect to $\varphi_\Sk$ and $\varphi'_\Sk$, we base these structures on exactly the same universe of ground terms. However, the sets of occurring predicate symbols are disjoint (as stipulated above).
	
	Since $\varphi_\Sk$ and $\varphi'_\Sk$ do neither contain equality nor existential quantifiers, we know that there are Herbrand models for them, if they are satisfiable at all.

%%%%%%%%%%%%%%%%%%%%%%%%%%%%%%%%%%%%%%%%%%%%%%%%%%%%%%%%%%%%%%%%%%%%%%%%%%%%%%%%
\begin{lemma}\label{lemma:GAFequisat:I}
	Given any Herbrand model $\cA \models \varphi_\Sk$, we can construct a  model $\cB \models \varphi'_\Sk$.
\end{lemma}
\begin{proof}
	Since $\cA$ is a Herbrand model of $\varphi_\Sk$, the universe $\fU_\cA$ contains all ground terms constructed from the constant symbols and  function symbols occurring in $\varphi_\Sk$.
	We define $\cB$ by taking over $\cA$'s universe and its interpretations of the constant and function symbols.
	Moreover, we set $P_i^\cB := \bigl\{ t\in \fU_\cB \bigm| \cA, [x \mapsto t] \models A_i \bigr\}$ if there is some $x \in \vars(A_i)$.
	For any $P_i$, for which $A_i$ is ground, we set $P_i^\cB := \fU^\cB$ if $\cA \models A_i$ and $P_i^\cB := \emptyset$ otherwise.
	
	We first prove $\cB \models \varphi_\Mon$. $\varphi_\Sk$ differs from $\varphi_\Mon$ only in the occurrences of atoms. It thus suffices to show that for two corresponding atom occurrences $A$ in $\varphi_\Sk$ and $B$ in $\varphi_\Mon$ and for an arbitrary variable assignment $\beta$ it holds $\cA, \beta \models A$ if and only if $\cB, \beta \models B$. But this is guaranteed by construction of $\varphi_\Mon$ and $\cB$:
	We need to consider two cases.
	\begin{description}
		\item If $A$ is ground, then there is some $j$ such that $A = A_j \in \At'$ and $B = P_j(d)$. By construction of $\cB$, we observe $\cA \models A_j$ if and only if $\cB \models P(d)$.
		
		\item If $A$ contains a variable $x$, then there is some $A_j \in \At'$ such that $A \simeq A_j$. Moreover, we know that $B = P_j(x)$. Because of $A \simeq A_j$, there must be some $x' \in \vars(A_j)$ such that for every $t\in \fU_\cA = \fU_\cB$ it holds $\cA, [x \mapsto t] \models A$ if and only if $\cA, [x' \mapsto t] \models A_j$.
			In addition, we have constructed $\cB$ in a way leading to $\cA, [x' \mapsto t] \models A_j$ if and only if $\cB, [x \mapsto t] \models P_j(x)$.
			
			Put togther, this yields $\cA, [x \mapsto t] \models A$ if and only if $\cB, [x \mapsto t] \models P_j(x)$.
	\end{description}	

	\noindent
	Next, we have to show $\cB \models \forall x_*. \psi$ for every $\psi \in \Psi \cup \Psi'$.
	\begin{description}
		\item If $\psi \in \Psi$, then $\psi = P_i\bigl( (x\theta)\tau_* \bigr) \leftrightarrow P_j(x_*)$ with $\vars(A_i) = \{x\}$ and $A_i\theta = A_j$. 
			Assume, without loss of generality, that $\vars(A_j) = \vars(A_j\theta) = \{x'\}$. (If $A_j$ is ground, the argument is still valid.)
					
			By construction of $\cB$, for any $t$ we have $\cB,[x_*\mapsto t] \models P_i\bigl( (x\theta)\tau_* \bigr)$ if and only if $\cA,[x' \mapsto t] \models A_i\theta$ and, moreover, $\cB,[x_*\mapsto t] \models P_j(x_*)$ if and only of $\cA,[x'\mapsto t] \models A_j$.
			Consequently, $A_i\theta = A_j$ leads to $\cB,[x_*\mapsto t] \models P_i\bigl( (x\theta)\tau_* \bigr)$ if and only if $\cB,[x_*\mapsto t] \models P_j(x_*)$.
			
		\item If $\psi \in \Psi'$, then $\psi = P_j(x_*) \leftrightarrow P_j(d)$ for some ground atom $A_j \in \At'$. The structure $\cB$ is constructed so that $P_j^\cB = \fU_\cB$ if $\cA \models A_j$ and $P_j^\cB = \emptyset$ otherwise. Hence, we have $\cB, \beta \models P_j(x_*)$ if and only if $\cB, \beta \models P_j(d)$.
	\end{description}
	Hence, $\cB \models \forall x_*. \psi$ follows in both cases.
\end{proof}	

%%%%%%%%%%%%%%%%%%%%%%%%%%%%%%%%%%%%%%%%%%%%%%%%%%%%%%%%%%%%%%%%%%%%%%%%%%%%%%%%
\begin{lemma}\label{lemma:GAFequisat:II}	
	Given any Herbrand model $\cB \models \varphi'_\Sk$, we can construct a model $\cA \models \varphi_\Sk$.
\end{lemma}
\begin{proof}	
	 Since $\cB$ is a Herbrand model of $\varphi'_\Sk$, and due to our assumption that the signature underlying $\varphi'_\Sk$ contains the same constant symbols and function symbols as the signature underlying $\varphi_\Sk$ does, the universe $\fU_\cB$ contains all ground terms constructed from the constant symbols and function symbols occurring in $\varphi_\Sk$. 

	We now define $\cA$ by taking over $\cB$'s universe and its interpretations of the constant and function symbols. 
	Moreover, for any predicate symbol $Q$ of arity $m$ occurring in $\varphi_\Sk$ we define $Q^\cA := S_1 \cup S_2$, where
		\begin{align*}
			\hspace{-3ex}
			S_1 := \bigl\{ \<t_1, \ldots, t_m\> \in \fU_\cA^m \bigm|\;
					&\text{there is a non-ground atom $Q(s_1, \ldots, s_m) \in \At$ and some}\\
					&\text{$A_j \in \At'$ with $A_j \simeq Q(s_1, \ldots, s_m)$ and there is a ground term $t$}\\
					&\text{ such that $Q(t_1, \ldots, t_m) = (Q(s_1, \ldots, s_1)\tau_*)\subst{x_*}{t}$}\\
					&\text{ and $\cB \models P_j(t)$} \bigr\}
		\end{align*}
		and
		\begin{align*}			
			S_2 :=\; \bigl\{ \<t_1, \ldots, t_m\> \in \fU_\cA^m \bigm|\; 
					&\text{there is a ground atom $A_j \in \At'$ such that}\\
					&\text{$A_j = Q(t_1, \ldots, t_m)$ and $\cB \models P_j(d)$} \bigr\} ~.
		\end{align*}	

	Again, since $\varphi_\Sk$ differs from $\varphi_\Mon$ only in the occurrences of atoms, it suffices to show that for two corresponding atom occurrences $A$ in $\varphi_\Sk$ and $B$ in $\varphi_\Mon$ and for an arbitrary variable assignment $\beta$ it holds $\cA, \beta \models A$ if and only if $\cB, \beta \models B$.

	Let $A = Q(s_1, \ldots, s_m)$. There must be some $A_j \in \At'$ such that $A_j \simeq A$ and $B = P_j(t)$.
	\begin{description}
		\item If $t = d$, then $A_j$ must be ground. 
			Hence, we have $\<s_1, \ldots, s_m\> \in S_2$ if and only if $\cB \models P_j(d)$.
			
			It remains to show that $\cB \not\models P_j(d)$ entails $\<s_1, \ldots, s_m\> \not\in S_1$.
			Suppose $\<s_1, \ldots, s_m\> \in S_1$, i.e.\ there is some non-ground atom $A_i \in \At'$, some $x \in \vars(A_i)$, and some ground term $t'$ such that $A_i\subst{x}{t'} = A_j$ and $\cB \models P_i(t)$.
			Because of $A_i \subst{x}{t'} = A_j$, we find the formula $P_i(t') \leftrightarrow P_j(x_*)$ in $\Psi$. Moreover, the formula $P_j(x_*) \leftrightarrow P_j(d)$ belongs to $\Psi'$.			
			Hence, $\cB$ is a model of both $\forall x_*. P_i(t') \leftrightarrow P_j(x_*)$ and $\forall x_*. P_j(x_*) \leftrightarrow P_j(d)$.
			But then $\cB \models P_i(t')$ contradicts $\cB \not\models P_j(d)$.
			Consequently, $\<s_1, \ldots, s_m\>$ cannot belong to $S_1$.
			
		\item If $t = x$ for some variable $x$, then $\vars(A) = \{x\}$ and there is some $x'\in \vars(A_j)$.
			Let $\theta := \subst{x}{\beta(x)}$.
			
			If $\cB \models P_j(x\theta)$, then we get $\<s_1, \ldots, s_m\>\theta \in S_1 \subseteq Q^\cA$, and thus $\cA,\beta \models Q(s_1, \ldots, s_m)$.
				
			If $\cB \not\models P_j(x\theta)$, then we may conclude $\<s_1, \ldots, s_m\>\theta \not\in S_1 \cup S_2 = Q^\cA$ and thus  $\cA,\beta \not\models Q(s_1, \ldots, s_m)$ due to the following arguments.
			
			Suppose $\<s_1, \ldots, s_m\>\theta \in S_1$.
			By definition of $\cB$, it holds $\cB \not\models A_j(x\theta)$.
			Hence, there must be some non-ground $A_i(x'') \in \At'$ with $i \neq j$, for which we can find a substitution $\rho$ such that $A_i\rho = Q(s_1, \ldots, s_m)\theta$ and it holds $\cB \models P_i(x''\rho)$. 			
			But then, $A_i$ and $A_j$ are unifiable, and there must exist an mgu $\sigma_{ij}$ of $A_i$ and $A_j$. 
			
			By Lemma~\ref{lemma:UnifiabilityForSingleVariableAtoms} and due to the definition of $At'$, we have to consider the following cases:
			\begin{enumerate}[label=(\arabic{*}), ref=(\arabic{*})]
				\item\label{enum:proofGAFequisat:I:I} $A_i \lesssim A_j$,
				\item\label{enum:proofGAFequisat:I:II} $A_j \lesssim A_i$, or
				\item\label{enum:proofGAFequisat:I:III} $A_i\sigma_{ij} = A_j\sigma_{ij} = Q(s_1, \ldots, s_m)\theta$.
			\end{enumerate}
			In case \ref{enum:proofGAFequisat:I:I} there must be some $\tau_1$ such that $A_i\tau_1 = A_j$ and $\Psi$ contains a formula $\psi_1 := P_i\bigl((x''\tau_1)\tau_*\bigr) \leftrightarrow P_j\bigl(x_*\bigr)$.
			Hence, we get $\cB \models P_j(x\theta)$ if and only if $\cB, [x_* \mapsto x\theta] \models P_j(x_*)$ if and only if $\cB, [x_* \mapsto x\theta] \models P_i((x''\tau_1)\tau_*)$ if and only if $\cB \models P_i(x''\rho)$. But this is contradictory to our assumptions.
			
			In case \ref{enum:proofGAFequisat:I:II} there must be some $\tau_2$ such that $A_j\tau_2 = A_i$ and $\Psi$ contains a formula $\psi_2 := P_j\bigl((x'\tau_2)\tau_*\bigr) \leftrightarrow P_i\bigl(x_*\bigr)$.
			Therefore, it holds $\cB \models P_i(x''\rho)$ if and only if $\cB, [x_* \mapsto x''\rho] \models P_i\bigl(x_*\bigr)$ if and only if $\cB, [x_* \mapsto x''\rho] \models P_j\bigl((x'\tau_2)\tau_*\bigr)$ if and only if $\cB \models P_j(x\theta)$. Again, this contradicts our assumptions.
			
			In case \ref{enum:proofGAFequisat:I:III} the set $\At'$ contains some atom $A_k := A\theta$ and $\Psi$ contains the formulas $\psi_3 := P_j\bigl(x\theta\bigr) \leftrightarrow P_k \bigl(x_*\bigr)$ and $\psi_4 := P_i\bigl(x''\rho\bigr) \leftrightarrow P_k \bigl(x_*\bigr)$. In addition, $\Psi'$ contains the formula $\psi_5 := P_k(x_*) \leftrightarrow P_k(d)$.
			Since $\cB$ is a model of $\forall x_*. \psi_3 \wedge \psi_4 \wedge \psi_5$, we get $\cB \models P_j(x\theta)$ if and only if $\cB \models P_i(x''\rho)$. This is on contradiction with our assumptions $\cB\not\models P_j(x\theta)$ and $\cB\models P_i(x''\rho)$ as well.
					
			In all three cases, we conclude $\<s_1, \ldots, s_m\>\theta \not\in S_1$.
			
			Suppose $\<s_1, \ldots, s_m\>\theta \in S_2$. Hence, there must be some ground atom $A_k \in \At'$ with $k \neq i$ and $A_k = Q(s_1, \ldots, s_m)\theta$.
			Moreover, $P_k(d)$ must hold true under$\cB$.
			Since we then observe $A_j \lesssim A_k$, $\Psi$ must contain the formula $\psi'_1 := P_j(x\theta) \leftrightarrow P_k(x_*)$. 
			Hence, it holds $\cB \models \forall x_*. P_k(x_*)$ if and only if $\cB \models P_j(x\theta)$, because $x\theta$ is ground. In particular, our assumption $\cB\not\models P_j(x\theta)$ entails $\cB,[x_*\mapsto d] \not\models P_k(d)$. Thus, we have reached a contradiction.
			Consequently, $\<s_1, \ldots, s_m\>\theta \not\in S_2$.
			\qedhere
	\end{description}
\end{proof}

This finishes the proof of Lemma~\ref{lemma:TransformToFullMonadic:UnaryFunctions}.

%%%%%%%%%%%%%%%%%%%%%%%%%%%%%%%%%%%%%%%%%%%%%%%%%%%%%%%%%%%%%%%%%%%%%%%%%%%%%%%%%%%
%%%%%%%%%%%%%%%%%%%%%%%%%%%%%%%%%%%%%%%%%%%%%%%%%%%%%%%%%%%%%%%%%%%%%%%%%%%%%%%%%%%
%%%%%%%%%%%%%%%%%%%%%%%%%%%%%%%%%%%%%%%%%%%%%%%%%%%%%%%%%%%%%%%%%%%%%%%%%%%%%%%%%%%

\section{Related and future work} \label{section:furthrelwork}

In \cite{Dreben1979} (page 65) Dreben and Goldfarb extend the relational monadic fragment to a certain extent and call the result \emph{Initially-extended Essentially Monadic Fragment}.
This class lies in the intersection of GBSR and GAF. 
Hence, this result could be considered as a first step away from the standard classification based  on quantifier prefixes or the arity of used predicate symbols. 
In \cite{Fermuller1993a} (page 152) Ferm\"uller et al.\ combine the just described formula class with the Ackermann fragment and call the result \emph{AM}. 
AM is properly contained in GAF. 
Moreover, Ferm\"uller et al.\ argue that AM itself is a special case of Maslov's fragment K~\cite{Maslov1968}. 
K is incomparable to GBSR and GAF.
In the last couple of years several new fragments have been discovered \cite{Kieronski2014, Barany2015, Mogavero2015}, all of which are incomparable to GBSR and GAF.

There exist a number of results showing decidability of $\exists^* \forall \exists^*$ sentences with arbitrary function symbols~\cite{Gurevich1973, Maslov1972, Gradel1990a}---see \cite{Borger1997}, chapter 6.3, for an overview.
The methods therein may help to also show decidability of GAF with arbitrary function symbols.

The methods and results described in the present paper may be extended in various directions: 
	generalizing other known decidable prefix classes, 
	more liberal conditions regarding function symbols,
	scenarios in interpreted theories such as arithmetic.
Moreover, we have not yet thoroughly investigated the complexity of deciding GAF-satisfiability.

%%%%%%%%%%%%%%%%%%%%%%%%%%%%%%%%%%%%%%%%%%%%%%%%%%%%%%%%%%%%%%%%%%%%%%%%%%%%%%%%%%%
%%%%%%%%%%%%%%%%%%%%%%%%%%%%%%%%%%%%%%%%%%%%%%%%%%%%%%%%%%%%%%%%%%%%%%%%%%%%%%%%%%%
%%%%%%%%%%%%%%%%%%%%%%%%%%%%%%%%%%%%%%%%%%%%%%%%%%%%%%%%%%%%%%%%%%%%%%%%%%%%%%%%%%%

\newpage

\appendix

\section{Appendix}
%%%%%%%%%%%%%%%%%%%%%%%%%%%%%%%%%%%%%%%%%%%%%%%%%%%%%%%%%%%%%%%%%%%%%%%%%%%%%%%%%%%
%%%%%%%%%%%%%%%%%%%%%%%%%%%%%%%%%%%%%%%%%%%%%%%%%%%%%%%%%%%%%%%%%%%%%%%%%%%%%%%%%%%
%%%%%%%%%%%%%%%%%%%%%%%%%%%%%%%%%%%%%%%%%%%%%%%%%%%%%%%%%%%%%%%%%%%%%%%%%%%%%%%%%%%

%%%%%%%%%%%%%%%%%%%%%%%%%%%%%%%%%%%%%%%%%%%%%%%%%%%%%%%%%%%%%%%%%%%%%%%%%%%%%%%%
\subsection{Details omitted in Example~\ref{example:GAFandGBSR} in Section~\ref{section:Intro}}
For convenience, we remove redundant subformulas at all stages.

\bigskip
Details regarding the transformation of $\varphi_1$ into $\varphi'_1$:
	\begin{align*}
		&\exists u \forall x \exists v \forall z \exists y_1 y_2.\; 
			\Bigl(\neg P(u, x) \vee \Bigl(Q(x,v) \wedge R(u,z,y_1)\Bigr)\Bigr)\\
			&\hspace{18ex} \wedge\; \Bigl(P(u, x) \vee \Bigl(\neg Q(x,v) \wedge \neg R(u,z,y_2)\Bigr)\Bigr)  \\
		&\semequiv\;\; 
			\exists u \forall x \exists v.\; 
			\Bigl(\neg P(u, x) \vee \Bigl(Q(x,v) \wedge \underbrace{\forall z \exists y_1. R(u,z,y_1)}_{=:\; A(u)}\Bigr)\Bigr)\\
			 &\hspace{15ex} \wedge\; \Bigl(P(u, x) \vee \Bigl(\neg Q(x,v) \wedge \underbrace{\forall z \exists y_2. \neg R(u,z,y_2)}_{=:\; B(u)}\Bigr)\Bigr)  \\
		&\semequiv\;\; 
			\exists u \forall x .\; 
			\Bigl(\neg P(u, x) \wedge \underbrace{\Bigl(\exists v. \neg Q(x,v)\Bigr)}_{=:\; C(x)} \wedge\; B(u)\Bigr) 
				\;\vee\; \Bigl(P(u, x) \wedge \underbrace{\Bigl(\exists v. Q(x,v)\Bigr)}_{=:\; D(x)} \wedge\; A(u)\Bigr)\\			 
		&\semequiv\;\; 
			\exists u.\; 
			\Bigl(\forall x.\Bigl(\neg P(u, x) \vee D(x) \Bigr)\Bigr) \wedge \Bigl(\Bigl(\forall x.\neg P(u, x)\Bigr) \vee A(u) \Bigr)\\
			&\hspace{10ex}\wedge \Bigl(\forall x.\Bigl( C(x) \vee P(u,x) \Bigr)\Bigr) \wedge \Bigl(\forall x.\underbrace{\Bigl(C(x) \vee D(x) \Bigr)}_{\semequiv\; \top}\Bigr) \wedge \Bigl(\Bigl(\forall x.C(x)\Bigr) \vee A(u) \Bigr)\\
			&\hspace{10ex}\wedge \Bigl(B(u) \vee \forall x.P(u,x) \Bigr) \wedge \Bigl(B(u) \vee \forall x.D(x) \Bigr) \wedge \Bigl(B(u) \vee A(u) \Bigr)\\
		&\semequiv\;\; 
			\exists u.\; 
			\Bigl(\forall x.\Bigl(\neg P(u, x) \vee \exists v. Q(x,v) \Bigr)\Bigr) 
				\;\wedge\; \Bigl(\Bigl(\forall x.\neg P(u, x)\Bigr) \vee \forall z \exists y_1. R(u,z,y_1) \Bigr)\\
			&\hspace{10ex}\wedge\;\; \Bigl(\forall x.\Bigl(\exists v. \neg Q(x,v)\Bigr) \vee P(u,x) \Bigr) 
				\;\wedge\; \Bigl(\Bigl(\forall x.\exists v. \neg Q(x,v)\Bigr) \vee \forall z \exists y_1. R(u,z,y_1) \Bigr)\\
			&\hspace{10ex}\wedge\;\; \Bigl(\Bigl(\forall z \exists y_2. \neg R(u,z,y_2)\Bigr) \vee \forall x.P(u,x) \Bigr) \\
			&\hspace{10ex}\wedge\;\; \Bigl(\Bigl(\forall z \exists y_2. \neg R(u,z,y_2)\Bigr) \vee \forall x \exists v. Q(x,v) \Bigr)\\
			&\hspace{10ex}\wedge\;\; \Bigl(\Bigl(\forall z \exists y_2. \neg R(u,z,y_2)\Bigr) \vee \forall z \exists y_1. R(u,z,y_1) \Bigr)
	\end{align*}		

\bigskip
\noindent
Details regarding the transformation of $\varphi_2$ into $\varphi'_2$:
	\begin{align*}
		&\exists u \forall x \exists y \forall z. \Bigl(P(u,z) \wedge Q(u,x)\Bigr) \vee \Bigl(P(y,z) \wedge Q(u,y) \Bigr) \\
		&\semequiv\;\;
			\exists u \forall x \exists y \forall z. 	\Bigl(P(u,z) \vee P(y,z)\Bigr)  \wedge \Bigl(P(u,z) \vee Q(u,y)\Bigr)\\
			&\hspace{17ex}				\wedge \Bigl(Q(u,x) \vee P(y,z)\Bigr) \wedge \Bigl(Q(u,x) \vee Q(u,y) \Bigr) \\
		&\semequiv\;\;
			\exists u \forall x \exists y. 	\underbrace{\Bigl(\forall z. P(u,z) \vee P(y,z)\Bigr)}_{=:\; A(u,y)}  \wedge\, \Bigl(\underbrace{\Bigl(\forall z. P(u,z)\Bigr)}_{=:\; B(u)} \vee\, Q(u,y)\Bigr)\\
			&\hspace{14ex}				\wedge \Bigl(Q(u,x) \vee \underbrace{\Bigl(\forall z. P(y,z)\Bigr)}_{=:\; C(y)}\Bigr) \wedge \Bigl(Q(u,x) \vee Q(u,y) \Bigr) 	
	\end{align*}
	\begin{align*}
		&\semequiv\;\;
			\exists u \forall x. 	\Bigl(\underbrace{\Bigl(\exists y. A(u,y)\Bigr)}_{=:\; D(u)} \wedge\, B(u) \wedge Q(u,x) \Bigr)
							\vee \Bigl(\underbrace{\Bigl(\exists y. A(u,y) \wedge Q(u,y) \Bigr)}_{=:\; E(u)} \wedge\, Q(u,x) \Bigr) \\ 
			&\hspace{12ex}		\vee \underbrace{\Bigl(\exists y. A(u,y) \wedge Q(u,y) \wedge C(y) \Bigr)}_{=:\; F(u)} \\
		&\semequiv\;\;
			\exists u. 	\Bigl( D(u) \vee E(u) \vee F(u) \Bigr)
						\wedge \Bigl( B(u) \vee E(u) \vee F(u) \Bigr)
						\wedge \Bigl(\underbrace{\Bigl(\forall x. Q(u,x)\Bigr)}_{=:\; G(u)} \vee\, F(u) \Bigr) \\
		&\semequiv\;\;
			\exists u. 	\Bigl( D(u) \wedge B(u) \wedge G(u) \Bigr)
						\vee \Bigl( E(u) \wedge G(u) \Bigr)
						\vee F(u) \\
		&\semequiv\;\;
			\exists u. 		\Bigl( \Bigl(\exists y. \Bigl(\forall z. P(u,z) \vee P(y,z)\Bigr)\Bigr) \wedge \Bigl(\forall z. P(u,z)\Bigr) \wedge \Bigl(\forall x. Q(u,x)\Bigr) \Bigr)\\
			&\hspace{9ex}		\vee \Bigl( \Bigl(\exists y. \Bigl(\forall z. P(u,z) \vee P(y,z)\Bigr) \wedge Q(u,y) \Bigr) \wedge \Bigl(\forall x. Q(u,x)\Bigr) \Bigr)\\
			&\hspace{9ex}		\vee \Bigl(\exists y. \Bigl(\forall z. P(u,z) \vee P(y,z)\Bigr) \wedge Q(u,y) \wedge \Bigl(\forall z. P(y,z)\Bigr) \Bigr) \\
		&\semequiv\;\;
			\exists u \exists y \forall x z v. 	\Bigl(\bigl(P(u,x) \vee P(y,x)\bigr) \wedge P(u,x) \wedge Q(u,x) \Bigr)\\
			&\hspace{17ex}				\vee \Bigl(\bigl(P(u,z) \vee P(y,z)\bigr) \wedge Q(u,y) \wedge Q(u,z) \Bigr)\\
			&\hspace{17ex}				\vee \Bigl(\bigl( P(u,v) \vee P(y,v)\bigr) \wedge Q(u,y) \wedge P(y,v) \Bigr)
	\end{align*}

%%%%%%%%%%%%%%%%%%%%%%%%%%%%%%%%%%%%%%%%%%%%%%%%%%%%%%%%%%%%%%%%%%%%%%%%%%%%%%%%
%%%%%%%%%%%%%%%%%%%%%%%%%%%%%%%%%%%%%%%%%%%%%%%%%%%%%%%%%%%%%%%%%%%%%%%%%%%%%%%%
\subsection{Proof details concerning Section~\ref{section:GeneralizedBSR}}			
%%%%%%%%%%%%%%%%%%%%%%%%%%%%%%%%%%%%%%%%%%%%%%%%%%%%%%%%%%%%%%%%%%%%%%%%%%%%%%%%

%%%%%%%%%%%%%%%%%%%%%%%%%%%%%%%%%%%%%%%%%%%%%%%%%%%%%%%%%%%%%%%%%%%%%%%%%%%%%%%%
\subsection*{Proof of Lemma~\ref{lemma:TransformationGBSR}}

\begin{lemma*}
	Let $\varphi :=\; \forall \vx_1 \exists \vy_1 \ldots \forall \vx_n \exists \vy_n.\psi$ be a GBSR sentence.
	There exists a quantifier-free first-order formula $\psi'(\vu, \vv)$ such that $\varphi' := \exists \vu\, \forall \vv. \psi'(\vu, \vv)$ is semantically equivalent to $\varphi$ and all literals in $\varphi'$ already occur in $\varphi$ (modulo variable renaming).
\end{lemma*}
\begin{proof}
	The following transformations will mainly use the laws of associativity, commutativity, and distributivity from Boolean algebra and will employ Lemma~\ref{lemma:Miniscoping} to (re-)transform $\psi$ into particular syntactic shapes. 
	Note that this does not change the set of literals occurring in the intermediate steps (modulo variable renaming), since we start from a formula in negation normal form restricted to the connectives $\wedge, \vee, \neg$. 
	We shall make use of the partition of the literals in $\psi$ into the sets $\tL_0, \tL_1, \ldots, \tL_n$ to group exactly those literals in each step of our transformation.
	
	First, we give a description and, afterwards, present the formal part below.
	
	To begin with, we transform the matrix $\psi$ into a disjunction of conjunctions of literals $\bigvee_i \psi_i$. 
	In addition, we rewrite every $\psi_i$ into $\psi_i = \tchi_{i,0}^{(1)} \wedge \ldots \wedge \tchi_{i,n}^{(1)}$. 
	Every $\tchi_{i,\ell}^{(1)}$ is a conjunction of literals and comprises exactly the literals from the $\psi_i$ which belong to $\tL_\ell$. 
	Moreover, by Proposition~\ref{lemma:LiteralsDecomposition}\ref{enum:LiteralsDecomposition:II} and \ref{enum:LiteralsDecomposition:III}, we know that $\vars(\tchi_{i,\ell}^{(1)}) \subseteq \vy_1 \cup \ldots \cup \vy_\ell \cup \vx_{\ell+1} \cup \ldots \cup \vx_n$.
	We move the existential quantifier block $\exists \vy_n$ inwards such that it binds the $\tchi_{i,n}^{(1)}$ alone (none of the variables in $\vy_n$ occurs in any of the $\tchi_{i,\ell}^{(1)}$ with $\ell < n$). 
	The thus obtained sentence $\varphi''$ has the form $\forall\vx_1 \exists\vy_1 \ldots \forall\vx_n . \bigvee_i \tchi_{i,0}^{(1)} \wedge \ldots \wedge \tchi_{i,n-1}^{(1)} \wedge \exists\vy_n. \tchi_{i,n}^{(1)}$. 
	In further transformations, we shall treat the subformulas $\bigl(\exists\vy_n. \tchi_{i,n}^{(1)}\bigr)$ as indivisible units.
	
	Next, we transform the big disjunction in $\varphi''$ into a conjunction of disjunctions $\bigwedge_k \psi'_k$, group the disjunctions $\psi'_k$ into subformulas $\teta_{k,\ell}^{(1)}$ (in accordance with the sets $\tL_0, \ldots, \tL_n$, as before), and move the universal quantifier block $\forall \vx_n$ inwards. 
	Due to Proposition~\ref{lemma:LiteralsDecomposition}\ref{enum:LiteralsDecomposition:I}, we can split the quantifier block $\forall \vx_n$ so that universal quantifiers can be moved directly before the $\teta_{k,\ell}^{(1)}$. 
	The resulting formula is of the form $\forall\vx_1 \exists\vy_1 \ldots \forall\vx_{n-1} \exists\vy_{n-1} . \bigwedge_k \bigl(\forall(\vx_n \cap \tX_0).\, \teta_{k,0}^{(1)}\bigr) \vee \bigl(\forall(\vx_n \cap \tX_1).\, \teta_{k,1}^{(1)}\bigr) \vee \ldots \vee \bigl(\forall(\vx_n \cap \tX_{n-1}).\, \teta_{k,n-1}^{(1)}\bigr) \vee \teta_{k,n}^{(1)}$. 
	In further transformations, we shall treat the subformulas $\bigl(\forall(\vx_n \cap \tX_\ell).\, \teta_{k,\ell}^{(1)}\bigr)$ as indivisible units, as well.
	
	We reiterate this process until all the quantifiers have been moved inwards in the described way.
	
	One notable difference between moving inward universal quantifier blocks and existential ones is that Proposition~\ref{lemma:LiteralsDecomposition}\ref{enum:LiteralsDecomposition:I} allows us to split universal quantifier blocks so that subformulas $\teta_{k,\ell}^{(j)}$ and $\teta_{k,\ell'}^{(j)}$ (for distinct $\ell, \ell'$) do not occur together in the scope of the same universal quantifier that has been moved inward as described.
	We do not observe such a property for existential quantifier blocks. Hence, subformulas of the form $\bigl(\exists \vy_\ell. \tchi_{i,\ell}^{(j)} \wedge \ldots \wedge \tchi_{i,n}^{(j)}\bigr)$ may appear and will be treated as indivisible units. For the sake of readability, we will drop the neat division into parts in accordance with $\tL_0, \ldots, \tL_n$ and rather use the shorthand $\bigl(\exists \vy_\ell. \tchi_{i,\geq\ell}^{(j)}\bigr)$ for such constructs.
	\begin{align*}
		\forall\vx_1 &\exists\vy_1 \ldots \forall\vx_n \exists\vy_n. \psi(\vx_1, \ldots, \vx_n, \vy_1, \ldots, \vy_n)\\
		\semequiv&\; \forall\vx_1 \exists\vy_1 \ldots \forall\vx_n \exists\vy_n. \bigvee_i \tchi_{i,0}^{(1)}(\vx_1, \ldots, \vx_n) \wedge \tchi_{i,1}^{(1)}(\vy_1, \vx_2, \ldots, \vx_n) \wedge \ldots \\[-1ex]
			&\hspace{24ex} \wedge \tchi_{i,n-1}^{(1)}(\vy_1, \ldots, \vy_{n-1}, \vx_n) \wedge \tchi_{i,n}^{(1)}(\vy_1, \ldots, \vy_{n}) \\
		\semequiv&\; \forall\vx_1 \exists\vy_1 \ldots \forall\vx_n . \bigvee_i \tchi_{i,0}^{(1)}(\vx_1, \ldots, \vx_n) \wedge \tchi_{i,1}^{(1)}(\vy_1, \vx_2, \ldots, \vx_n) \wedge \ldots \\[-1ex]
			&\hspace{20ex} \wedge \tchi_{i,n-1}^{(1)}(\vy_1, \ldots, \vy_{n-1}, \vx_n) \wedge \Bigl(\exists\vy_n. \tchi_{i,n}^{(1)}(\vy_1, \ldots, \vy_{n})\Bigr) \\
		\semequiv&\; \forall\vx_1 \exists\vy_1 \ldots \forall\vx_n . \bigwedge_k \teta_{k,0}^{(1)}(\vx_1, \ldots, \vx_n) \vee \teta_{k,1}^{(1)}(\vy_1, \vx_2, \ldots, \vx_n) \vee \ldots \\[-1ex]
			&\hspace{20ex} \vee \teta_{k,n-1}^{(1)}(\vy_1, \ldots, \vy_{n-1}, \vx_n) \vee \teta_{k,n}^{(1)}(\vy_1, \ldots, \vy_{n-1}) \\
		\semequiv&\; \forall\vx_1 \exists\vy_1 \ldots \exists\vy_{n-1} . \bigwedge_k \Bigl(\forall(\vx_n \cap \tX_0). \teta_{k,0}^{(1)}(\vx_1, \ldots, \vx_n)\Bigr)\\[-1ex]
			&\hspace{22ex} \vee \Bigl(\forall(\vx_n \cap \tX_1). \teta_{k,1}^{(1)}(\vy_1, \vx_2, \ldots, \vx_n)\Bigr) \vee \ldots \\[-1ex]
			&\hspace{22ex} \vee \Bigl(\forall(\vx_n \cap \tX_{n-1}). \teta_{k,n-1}^{(1)}(\vy_1, \ldots, \vy_{n-1}, \vx_n)\Bigr) \\[-1ex]
			&\hspace{22ex}\vee \teta_{k,n}^{(1)}(\vy_1, \ldots, \vy_{n-1}) \\
		\semequiv&\; \forall\vx_1 \exists\vy_1 \ldots \exists\vy_{n-1}. \bigvee_i \tchi_{i,0}^{(2)}(\vx_1, \ldots, \vx_{n-1}) \wedge \tchi_{i,1}^{(2)}(\vy_1, \vx_2, \ldots, \vx_{n-1}) \wedge \ldots \\[-1ex]
			&\hspace{22ex} \wedge \tchi_{i,n-1}^{(2)}(\vy_1, \ldots, \vy_{n-1}) \wedge \tchi_{i,n}^{(2)}(\vy_1, \ldots, \vy_{n-1}) 
	\end{align*}
	\begin{align*}			
		\semequiv&\; \forall\vx_1 \exists\vy_1 \ldots \forall\vx_{n-1}. \bigvee_i \tchi_{i,0}^{(2)}(\vx_1, \ldots, \vx_{n-1}) \wedge \tchi_{i,1}^{(2)}(\vy_1, \vx_2, \ldots, \vx_{n-1}) \wedge \ldots \\[-1ex]
			&\hspace{22ex} \wedge  \Bigl( \exists\vy_{n-1}. \tchi_{i,n-1}^{(2)}(\vy_1, \ldots, \vy_{n-1}) \wedge \tchi_{i,n}^{(2)}(\vy_1, \ldots, \vy_{n-1}) \Bigr) \\
		&\hspace{5ex}\vdots\\
		\semequiv&\; \forall\vx_1 \exists\vy_1 . \bigvee_i \tchi_{i,0}^{(n)}(\vx_1) \wedge \tchi_{i,\geq 1}^{(n)}(\vy_1) \\
		\semequiv&\; \forall\vx_1 . \bigvee_i \tchi_{i,0}^{(n)}(\vx_1) \wedge \Bigl( \exists\vy_1. \tchi_{i,\geq 1}^{(n)}(\vy_1) \Bigr) \\
		\semequiv&\; \bigwedge_k \Bigl( \forall\vx_1 . \teta_{k,0}^{(n)}(\vx_1) \Bigr) \vee \teta_{k,\geq 1}^{(n)}()
	\end{align*}
	Finally, we can move all quantifiers outwards again---existential quantifiers first, then universal ones---, renaming variables as necessary. In the end, we obtain the prenex formula
		$ \varphi' :=\;\exists \vu\, \forall \vv. \bigwedge_k \heta_{k,0}(\vv) \vee \heta_{k,\geq 1}(\vu,\vv)$, 
	where $\heta_{k,0}$ and $\heta_{k,\geq 1}$ are quantifier-free variants of $\teta_{k,0}^{(n)}$ and $\teta_{k,\geq 1}^{(n)}$, respectively, with appropriately renamed variables.
\end{proof}

%%%%%%%%%%%%%%%%%%%%%%%%%%%%%%%%%%%%%%%%%%%%%%%%%%%%%%%%%%%%%%%%%%%%%%%%%%%%%%%%
%%%%%%%%%%%%%%%%%%%%%%%%%%%%%%%%%%%%%%%%%%%%%%%%%%%%%%%%%%%%%%%%%%%%%%%%%%%%%%%%
\subsection*{Proof of Lemma~\ref{lemma:LyndonInterpolationBSR}}

\begin{lemma*}
	Let $\varphi$ and $\psi$ be relational BSR sentences without equality in which the only Boolean connectives are $\wedge, \vee, \neg$.
	If $\varphi \models \psi$, then there exists a relational BSR sentence $\chi$ without equality such that
	\begin{enumerate}[label=(\roman{*}), ref=(\roman{*})]
		\item $\varphi \models \chi$ and $\chi \models \psi$, and
		\item any predicate symbol $P$ occurs positively (negatively) in $\chi$ only if it occurs positively (negatively) in $\varphi$ and in $\psi$.
	\end{enumerate}
\end{lemma*}
%%%%%%%%%%%%%%%%%%%%%%%%%%%%%%%%%%%%%%%%%%%%%%%%%%%%%%%%%%%%%%%%%%%%%%%%%%%%%%%%

\medskip
\noindent
The general outline of the following proof based on ordered resolution with selection goes back to Harald Ganzinger (lecture notes ``Logic in Computer Science'', 2002).
\begin{proof}[Proof sketch]
	In the degenerate cases where $\varphi$ is unsatisfiable, i.e.\ $\varphi \models \bot$, or where $\psi$ is a tautology, i.e.\ $\top \models \psi$, we set $\chi := \bot$ and $\chi := \top$, respectively. In all other cases we proceed as follows.
	
	Let $\varphi'$ and $\psi'$ be quantifier-free formulas and let $\vu, \vv, \vx, \vy$ be tuples of variables such that $\varphi = \exists \vy\, \forall \vx.\, \varphi'$ and $\psi = \exists \vv\, \forall \vu.\, \psi'$.
	Without loss of generality, we assume that $\vu, \vv, \vx, \vy$ are pairwise disjoint and that $\varphi' := \bigwedge_i \varphi_i$ and $\psi' := \bigwedge_j \psi_j$ are in conjunctive normal form.
	
	Let $\Pi_1$ be the set of all predicate symbols that occur in $\varphi'$ but not in $\psi'$,
	let $\Pi_2$ be the set of all predicate symbols that occur positively in $\varphi'$ but not positively in $\psi'$ and that do not belong to $\Pi_1$,
	let $\Pi_3$ be the set of all predicate symbols that occur negatively in $\varphi'$ but not negatively in $\psi'$ and that do not belong to $\Pi_1$,
	let $\Pi_4$ be the set of all predicate symbols that occur in $\varphi'$ and in $\psi'$ but do not belong to $\Pi_1 \cup \Pi_2 \cup \Pi_3$.
	We construct the formulas $\hphi'$ and $\hpsi'$ from $\varphi'$ and $\psi'$, respectively, by simultaneously replacing every literal $\neg P(\vec{s}\,)$ by $P(\vec{s}\,)$ and every literal $P(\vec{s}\,)$ by $\neg P(\vec{s}\,)$ for every $P \in \Pi_2$.
	Hence, every $P \in \Pi_2$ occurs negatively in $\hphi'$ but not negatively in $\hpsi'$, and there are no predicate symbols that occur positively in $\hphi'$ but not positively in $\hpsi'$.
	Moreover, we observe that the above transformation preserves (un)satisfiability of $\varphi$, $\neg \psi$, and $\varphi \wedge \neg \psi$, i.e.
	\begin{itemize}
		\item $\exists \vy\, \forall \vx.\, \varphi' \models \bot$ if and only if $\exists \vy\, \forall \vx.\, \hphi' \models \bot$,
		\item $\neg \exists \vv\, \forall \vu.\, \psi' \models \bot$ if and only if $\neg \exists \vv\, \forall \vu.\, \hpsi' \models \bot$, and
		\item $\bigl( \exists \vy\, \forall \vx.\, \varphi' \bigr) \wedge \neg \bigl( \exists \vv\, \forall \vu.\, \psi' \bigr) \models \bot$ if and only if $\bigl( \exists \vy\, \forall \vx.\, \hphi' \bigr) \wedge \neg \bigl( \exists \vv\, \forall \vu.\, \hpsi' \bigr) \models \bot$.
	\end{itemize}	

	Let $\hphi_\Sk := \forall \vx.\, \hphi'\subst{y_1 / c_1, \ldots, y_{|\vy|} / c_{|\vy|}}$ where the $c_i$ are fresh Skolem constants.
	Moreover, let $\hpsi_\Sk := \forall \vv.\, \neg \hpsi'\subst{u_1 / f_1(\vv), \ldots, u_{|\vu|} / f_{|\vu|}(\vv)}$ where the $f_i$ are fresh Skolem functions of arity $|\vv|$.
	Hence, $\hphi_\Sk \wedge \hpsi_\Sk$ is a Skolemized variant of 
		$\bigl( \exists \vy\, \forall \vx.\, \hphi' \bigr) \wedge \bigl( \forall \vv\, \exists \vu.\, \neg \hpsi' \bigr)$, which is semantically equivalent to $\bigl( \exists \vy\, \forall \vx.\, \hphi' \bigr) \wedge \neg \bigl( \exists \vv\, \forall \vu.\, \hpsi' \bigr)$.
	Therefore, we observe 
		\begin{align*}
			\varphi \wedge \neg \psi \models \bot \quad
				&\text{ if and only if }\quad \bigl( \exists \vy\, \forall \vx.\, \hphi' \bigr) \wedge \neg \bigl( \exists \vv\, \forall \vu.\, \hpsi' \bigr) \models \bot \\
				&\text{ if and only if }\quad \hphi_\Sk \wedge \hpsi_\Sk \models \bot ~.
		\end{align*}	
	
	Let $N$ be a clause set corresponding to $\hphi_\Sk$  such that every $P$ occurring positively (negatively) in $N$ also occurs positively (negatively) in $\hphi_\Sk$. (We define $N$ to be the set containing all the implicitly universally quantified clauses $\hphi_i$ from $\hphi_\Sk$ whose variables are renamed so that the clauses in $N$ are pairwise variable disjoint).
	Analogously, let $M$ be the clause set corresponding to $\hpsi_\Sk$ such that every $P$ occurring positively (negatively) in $M$ occurs positively (negatively) in $\hpsi_\Sk$ and negatively (positively) in $\hpsi'$.
	
	We apply \emph{ordered resolution with selection} (cf.\ the calculus $O^\succ_S$ by Bachmair and Ganzinger in \cite{Bachmair2001}, page 41) 
	to $N$ until the clause set is saturated (without using any redundancy criterion) and call the result $N_*$.
	As underlying term ordering we apply some \emph{reduction ordering $\succ$} satisfying the following conditions.
	For all ground atoms $P(s_1, \ldots, s_m)$ and $R(t_1, \ldots, t_n)$ we require $P(s_1, \ldots, s_m) \;\succ\; \neg R(t_1, \ldots, t_n) \;\succ\; R(t_1, \ldots, t_n)$ whenever $P \in \Pi_1$ and $R \not\in \Pi_1$.
	To achieve this, we use a \emph{lexicographic path ordering} based on some precedence $\succ$ for which $P \succ R \succ f \succ c$ for any $P \in \Pi_1$, $R \not\in \Pi_1$, any Skolem function $f$ occurring in $\hpsi_\Sk$, and any Skolem constant $c$ occurring in $\hphi_\Sk$.
	We lift	the resulting (total) ordering on ground terms to a (partial) ordering on non-ground terms by stipulating $s \succ t$ if and only if for every substitution $\theta$ for which $s\theta$ and $t\theta$ are ground we have $s\theta \succ t\theta$.	
	The selection function that we use shall select exactly the literals $\neg P(\vec{s}\,)$ with $P \in \Pi_2 \cup \Pi_3$ in clauses that contain such literals. 
	In all other clauses nothing shall be selected.
	Let $M_*$ be the result of saturating $M$ in the same way as we have saturated $N$ to obtain $N_*$.
	
	Note that $N_*$ may be infinite, but may only contain clauses whose literals are instances of the literals in $N$ where variables are either instantiated with variables or constant symbols from $\vec{c}$.
	Since $\varphi$ (and thus also $\hphi_\Sk$) is satisfiable and since ordered resolution with selection is sound, $N_*$ does not contain the empty clause.
	The set $M_*$ may also be infinite.
	Due to our assumption that $\psi$ is not valid, $\neg \psi$ (and thus also $\hpsi_\Sk$) must be satisfiable. 
	Hence, $M_*$ does not contain the empty clause either.
	
	As our assumption $\varphi \models \psi$ is equivalent to $\varphi \wedge \neg \psi \models \bot$ and to $\hphi_\Sk \wedge \hpsi_\Sk \models \bot$, refutational completeness of ordered resolution with selection entails that there is a (finite) derivation $\fD$ of the empty clause $\Box$ (which at the same stands for \emph{falsity} $\bot$) from the unsatisfiable set of clauses $N_* \cup M_*$.
	We assume that $\fD$ is based on the same calculus and the same term ordering that we have used to saturate $N$ and $M$.
	Let $N'_*$ be the set of clauses from $N_*$ whose instances are used as premises in this derivation.
	Since $N_*$ and $M_*$ are both saturated and neither of them contains the empty clause, $\fD$ must indeed make use of clauses from $N_*$, and, hence, $N'_*$ is not empty.
	Since $N'_*$ is finite, we can define the sentence $\hchi_\Sk := \forall \vz.\, \bigwedge_{C \in N'_*} C$, where we set $\vz := \vars(N'_*)$.
	We observe the following properties for $\hchi_\Sk$ and the underlying clause set $N'_*$:
	\begin{enumerate}[label=(\arabic{*}), ref=\arabic{*}]
		\item\label{enum:proofLyndonInterpolationBSR:I} $\hphi_\Sk \models \hchi_\Sk$,
		\item\label{enum:proofLyndonInterpolationBSR:II} $\hchi_\Sk \wedge \hpsi_\Sk \models \bot$, 
		\item\label{enum:proofLyndonInterpolationBSR:III} for every $C \in N'_*$ we have
			\begin{enumerate}[label=(\ref{enum:proofLyndonInterpolationBSR:III}.\arabic{*}), ref=(\ref{enum:proofLyndonInterpolationBSR:III}.\arabic{*})]
				\item\label{enum:proofLyndonInterpolationBSR:III:I} for every literal $P(s_1, \ldots, s_m)$ in $C$ there is a clause $D \in M$ that contains some literal $\neg P(t_1, \ldots, t_m)$, and
				\item\label{enum:proofLyndonInterpolationBSR:III:II} for every literal $\neg P(s_1, \ldots, s_m)$ in $C$ there is a clause $D \in M$ that contains some literal $P(t_1, \ldots, t_m)$.				
			\end{enumerate}
					
	\end{enumerate}
	\begin{description}
		\item Ad~(\ref{enum:proofLyndonInterpolationBSR:I}) and~(\ref{enum:proofLyndonInterpolationBSR:II}).
			Both observations follow by soundness of ordered resolution.
			\hfill$\Diamond$

		\item Ad~(\ref{enum:proofLyndonInterpolationBSR:III}).
			Since $N_*$ and $M_*$ are both saturated and do not contain the empty clause, any inference step in $\fD$ that starts from two leaves of the derivation tree involves some clause taken from $N'_*$ and some clause taken from $M_*$.
			Consider any such resolution step between clauses $C \in N'_*$ and $D \in M_*$.
			By case distinction on the possible resolution steps we show that $C$ cannot contain any literal $[\neg] P(\vec{s}\,)$ with $P \in \Pi_1 \cup \Pi_2 \cup \Pi_3$.
			
			\begin{description}
				\item Suppose there is an ordered resolution step between two clauses $C = C' \vee R(\vec{t}\,) \in N'_*$ and $D = D' \vee \neg R(\vec{t'}) \in M_*$ over the literals $R(\vec{t}\,)$ and $\neg R(\vec{t'})$ such that $C$ contains some literal $[\neg] P(\vec{s}\,)$ with $P \in \Pi_1$.
					Since $R$ occurs in $N_*$ and in $M_*$, we have $R \not\in \Pi_1$.
					Hence, we get $P(\vec{s}\,) \succ R(\vec{t}\,)$.
					Due to the order restrictions in ordered resolution, $R(\vec{t}\,)\tau$ must be maximal in $C\tau$, where $\tau$ is the unifier that is used in the resolution step to unify $R(\vec{t}\,)$ and $R(\vec{t'})$.
					But this contradicts $P(\vec{s}\,) \succ R(\vec{t}\,)$, as the latter entails $P(\vec{s}\,)\tau \succ R(\vec{t}\,)\tau$.
				
				\item Suppose there is an ordered resolution step between two clauses $C = C' \vee \neg R(\vec{t}\,) \in N'_*$ and $D = D' \vee R(\vec{t'}) \in M_*$ over the literals $\neg R(\vec{t}\,)$ and $R(\vec{t'})$ such that $C$ contains some literal $[\neg] P(\vec{s}\,)$ with $P \in \Pi_1$.
					Since $R$ occurs negatively in $N_*$ and positively in $M_*$, we conclude $R \not\in \Pi_1 \cup \Pi_2 \cup \Pi_3$.
					Hence, we have that $P(\vec{s}\,) \succ R(\vec{t}\,)$, which entails $P(\vec{s}\,) \succ \neg R(\vec{t}\,)$, and $\neg R(\vec{t}\,)$ is not selected in $C$.
					But then, due to the order restrictions in ordered resolution, $\neg R(\vec{t}\,)\tau$ must be maximal in $C\tau$, where $\tau$ is the unifier that is used  to unify $R(\vec{t}\,)$ and $R(\vec{t'})$.
					But this contradicts $P(\vec{s}\,) \succ \neg R(\vec{t}\,)$, as the latter entails $P(\vec{s}\,)\tau \succ \neg R(\vec{t}\,)\tau$.
				
				\item Suppose there is an ordered resolution step between two clauses $C = C' \vee R(\vec{t}\,) \in N'_*$ and $D = D' \vee \neg R(\vec{t'}) \in M_*$ over the literals $R(\vec{t}\,)$ and $\neg R(\vec{t'})$ such that $C$ contains some literal $\neg P(\vec{s}\,)$ with $P \in \Pi_2 \cup \Pi_3$.
					Since $\neg P(\vec{s}\,)$ is selected in $C$, the resolution step cannot be performed.
				
				\item Suppose there is an ordered resolution step between two clauses $C = C' \vee \neg R(\vec{t}\,) \in N'_*$ and $D = D' \vee R(\vec{t'}) \in M_*$ over the literals $\neg R(\vec{t}\,)$ and $R(\vec{t'})$ such that $C$ contains some literal $\neg P(\vec{s}\,)$ with $P \in \Pi_2 \cup \Pi_3$.
					Since $R$ occurs negatively in $N_*$ and positively in $M_*$, it must occur negatively in $\hpsi'$, and thus $R \not\in \Pi_1 \cup \Pi_2 \cup \Pi_3$.
					Hence, the literal $\neg R(\vec{t}\,)$ is not selected in $C$.
					Since, on the other hand, there is a selected literal in $C$, namely $\neg P(\vec{s}\,)$, the resolution step cannot be performed.
			\end{description}	

			Consequently, the result of any inference starting from two leaf nodes of the derivation tree of $\fD$ cannot contain any predicate symbol $P \in \Pi_1$ and it cannot contain any literal $\neg R(\ldots)$ with $R \in \Pi_2 \cup \Pi_3$.
			
			By an inductive argument (over the height of derivation trees), this leads to the observation that none of the clauses from $N_*$ that are involved in the derivation $\fD$ can contain any predicate symbol from $\Pi_1$ or any negative literal $\neg R(\ldots)$ with $R \in \Pi_2 \cup \Pi_3$. 
			Since $N'_*$ contains only clauses that are involved in $\fD$, \ref{enum:proofLyndonInterpolationBSR:III:II} is satisfied.
			By construction of $N_*$ from $\varphi = \exists \vy\, \forall \vx.\, \varphi'$ via $\exists \vy\, \forall \vx.\, \hphi'$ and $\hphi_\Sk$, Condition \ref{enum:proofLyndonInterpolationBSR:III:I} is satisfied as well.
			\hfill$\Diamond$
	\end{description}
	Since $\chi_\Sk$ contains exclusively constant symbols from $\vec{c}$, we can easily construct $\hchi'$ from $\hchi_\Sk$'s matrix by de-Skolemization, i.e.\ $\hchi_\Sk = \forall \vz.\, \hchi'\subst{y_1 / c_1, \ldots,$ $y_{|\vy|} / c_{|\vy|}}$.
	Furthermore, we construct the formula $\chi'$ from $\hchi'$ by simultaneously replacing every literal $\neg P(\vec{s}\,)$ by $P(\vec{s}\,)$ and every literal $P(\vec{s}\,)$ by $\neg P(\vec{s}\,)$ for every $P \in \Pi_2$.
	Finally, we set $\chi := \exists \vy\, \forall \vz.\, \chi'$.
	
	It remains to prove the following properties:
	\newcounter{counter:proofBSRInterpolation}
	\setcounter{counter:proofBSRInterpolation}{3}
	\begin{enumerate}[label=(\arabic{counter:proofBSRInterpolation}), ref=(\arabic{counter:proofBSRInterpolation})]
		\addtocounter{counter:proofBSRInterpolation}{1}
		\item\label{enum:proofLyndonInterpolationBSR:IV} $\varphi \models \chi$, and
		\addtocounter{counter:proofBSRInterpolation}{1}
		\item\label{enum:proofLyndonInterpolationBSR:V} $\chi \wedge \neg \psi \models \bot$.
	\end{enumerate}
	\begin{description}
		\item Ad~\ref{enum:proofLyndonInterpolationBSR:IV}.
			For every model $\cA \models \exists \vy\, \forall \vx.\, \hphi'$ there is some model $\cB \models \hphi_\Sk$ such that $\cA$ and $\cB$ differ only in their interpretation of the Skolem constants $c_1, \ldots, c_{|\vy|}$.
			By~(\ref{enum:proofLyndonInterpolationBSR:I}) and because of $\hchi_\Sk \models \exists \vy\, \forall \vx.\, \hchi'$, we get $\cB \models \exists \vy\, \forall \vx.\, \hchi'$.
			Since $\cB$ differs from $\cA$ only in the interpretation of symbols that do not occur in $\exists \vy\, \forall \vx.\, \hchi'$, $\cA$ is also a model of $\exists \vy\, \forall \vx.\, \hchi'$.
			Hence, $\exists \vy\, \forall \vx.\, \hphi' \models \exists \vy\, \forall \vx.\, \hchi'$, which can equivalently be written as $\bigl( \exists \vy\, \forall \vx.\, \hphi' \bigr) \wedge \neg \bigl( \exists \vy\, \forall \vx.\, \hchi' \bigr) \models \bot$.
			
			Since $\bigl( \exists \vy\, \forall \vx.\, \hphi' \bigr) \wedge \neg \bigl( \exists \vy\, \forall \vx.\, \hchi' \bigr) \models \bot$ holds if and only if $\bigl( \exists \vy\, \forall \vx.\, \varphi' \bigr) \wedge \neg \bigl( \exists \vy\, \forall \vx.\, \chi' \bigr) \models \bot$, and since the latter is equivalent to $\bigl( \exists \vy\, \forall \vx.\, \varphi' \bigr) \models \bigl( \exists \vy\, \forall \vx.\, \chi' \bigr)$ we in the end get 
				$\varphi \models \chi$.
			\strut\hfill$\Diamond$
			
		\item Ad~\ref{enum:proofLyndonInterpolationBSR:V}.
			The formula $\hchi_\Sk \wedge \hpsi_\Sk$ can be conceived as a Skolemized variant of $\bigl( \exists \vy\, \forall \vx.\, \hchi' \bigr) \wedge \bigl( \forall \vv\, \exists \vu.\, \neg \hpsi' \bigr)$, which is semantically equivalent to $\bigl( \exists \vy\, \forall \vx.\, \hchi' \bigr) \wedge \neg \bigl( \exists \vv\, \forall \vu.\, \hpsi' \bigr)$.
			Hence, we have $\hchi_\Sk \wedge \hpsi_\Sk \models \bot$ if and only if $\bigl( \exists \vy\, \forall \vx.\, \hchi' \bigr) \wedge \neg \bigl( \exists \vv\, \forall \vu.\, \hpsi' \bigr) \models \bot$.
			As we observe that $\bigl( \exists \vy\, \forall \vx.\, \hchi' \bigr) \wedge \neg \bigl( \exists \vv\, \forall \vu.\, \hpsi' \bigr) \models \bot$ holds if and only if $\bigl( \exists \vy\, \forall \vx.\, \chi' \bigr) \wedge \neg \bigl( \exists \vv\, \forall \vu.\, \psi' \bigr) \models \bot$, we in the end get
		\begin{align*}
			\hchi_\Sk \wedge \hpsi_\Sk \models \bot \quad
				&\text{ if and only if }\quad \bigl( \exists \vy\, \forall \vx.\, \hchi' \bigr) \wedge \neg \bigl( \exists \vv\, \forall \vu.\, \hpsi' \bigr) \models \bot \\
				&\text{ if and only if }\quad  \chi \wedge \neg \psi \models \bot ~.
		\end{align*}	
		By~(\ref{enum:proofLyndonInterpolationBSR:II}), this yields $\chi \wedge \neg \psi \models \bot$.
		\strut\hfill$\Diamond$			
	\end{description}
	
	Because of the equivalence of $\chi \wedge \neg \psi \models \bot$ and $\chi \models \psi$, we have shown that $\chi$ satisfies Requirement~\ref{enum:LyndonInterpolationBSR:I} of the lemma.	
	
	Due to (\ref{enum:proofLyndonInterpolationBSR:III}) and due to the way $\chi$ is constructed from $N'_*$, every positive occurrence of a predicate symbol $P$ in $\chi$ entails the existence of a negative occurrence of $P$ in $\neg \psi$, and every negative occurrence of a predicate symbol $P$ in $\chi$ entails the existence of a positive occurrence of $P$ in $\neg \psi$.
	Consequently, $\chi$ satisfies Requirement~\ref{enum:LyndonInterpolationBSR:II} as well.
\end{proof}

%%%%%%%%%%%%%%%%%%%%%%%%%%%%%%%%%%%%%%%%%%%%%%%%%%%%%%%%%%%%%%%%%%%%%%%%%%%%%%%%
%%%%%%%%%%%%%%%%%%%%%%%%%%%%%%%%%%%%%%%%%%%%%%%%%%%%%%%%%%%%%%%%%%%%%%%%%%%%%%%%
\subsection*{Proof of Lemma~\ref{lemma:UniformStrategies}}

\begin{lemma*}
	For every strategy $\sigma = \<\sigma_1, \ldots, \sigma_n\>$ there is a $\mu$-uniform strategy $\hsigma = \< \hsigma_1, \ldots, \hsigma_n \>$ such that $\Out_{\At,\hsigma} \subseteq \Out_{\At,\sigma}$.
\end{lemma*}
\begin{proof}
	We start with two preliminary results.
	\begin{description}
		%%%%%%%%%%%%%%%%%%%%%%%%%%%%	
		\item \underline{Claim I:} 
			Let $\ell, k$ be two integers such that $0 \leq \ell < k \leq n$.
			Let $\va_1, \ldots, \va_\ell, \vb_{\ell+1}, \ldots, \vb_k$ be tuples of domain elements, where $\va_i \in \fU_\cA^{|\vy_i|}$, $\vb_i \in \fU_\cA^{|\vx_i|}$ for every $i$.
			Consider two sequences of tuples $\vc_{k+1}, \ldots, \vc_{n}$ and $\vd_{k+1}, \ldots, \vd_{n}$ that coincide in all positions that correspond to variables $x$ occurring in $\At_\ell$, where $\vc_i, \vd_i \in \fU_\cA^{|\vx_i|}$ for every $i$.
			We have\\
				\centerline{\hspace{-2ex}$\mu_{\ell,k'}(\va_1, \ldots, \va_\ell, \vb_{\ell+1}, \ldots, \vb_k, \vc_{k+1}, \ldots, \vc_{k'}) =  \mu_{\ell,k'}(\va_1, \ldots, \va_\ell, \vb_{\ell+1}, \ldots, \vb_k, \vd_{k+1}, \ldots, \vd_{k'})$}
			for every $k'$, $k \leq k' \leq n$.	
						
		\item \underline{Proof:} 
			We proceed inductively from $k' = n$ downwards.
			
			\medskip
			Let $k' = n$.	
			By definition of the sequences $\vc_{k+1}, \ldots, \vc_{n}$ and $\vd_{k+1}, \ldots, \vd_{n}$ it holds\\
				\centerline{$\cA, [\vy_1 \Mapsto \va_1, \ldots, \vy_\ell \Mapsto \va_\ell, \vx_{\ell+1} \Mapsto \fb_{\ell+1}, \ldots, \vx_k \Mapsto \fb_k, \vx_{k+1} \Mapsto \fc_{k+1}, \ldots, \vx_n \Mapsto \vc_n ] \models A$}
			if and only if\\ 	
				\centerline{$\cA, [\vy_1 \Mapsto \va_1, \ldots, \vy_\ell \Mapsto \va_\ell, \vx_{\ell+1} \Mapsto \fb_{\ell+1}, \ldots, \vx_k \Mapsto \fb_k, \vx_{k+1} \Mapsto \fd_{k+1}, \ldots, \vx_n \Mapsto \vd_n ] \models A$}
			for every atom $A \in \At_\ell$.
			Hence, we have $\mu_{\ell,n}(\va_1, \ldots, \va_\ell, \vb_{\ell+1}, \ldots, \vb_k, \vc_{k+1}, \ldots, \vc_n) = $\linebreak $\mu_{\ell,n}(\va_1, \ldots, \va_\ell, \vb_{\ell+1}, \ldots, \vb_k, \vd_{k+1}, \ldots, \vd_n)$.
		
			\medskip
			Let $k' < n$.
			Consider any set $S \in \mu_{\ell, k'}(\va_1, \ldots, \va_\ell, \vb_{\ell+1}, \ldots, \vb_k, \vc_{k+1}, \ldots, \vc_{k'})$.
			By definition of $\mu_{\ell,k'}$, there must be some tuple $\vc_{k'+1}$ such that $S = \mu_{\ell, k'+1}(\va_1, \ldots, \va_\ell, \vb_{\ell+1}, \ldots, \vb_k,$ $\vc_{k+1}, \ldots, \vc_{k'}, \vc_{k'+1})$.
			By induction, $S = \mu_{\ell, k'+1}(\va_1, \ldots, \va_\ell, \vb_{\ell+1}, \ldots, \vb_k, \vd_{k+1}, \ldots, \vd_{k'}, \vc_{k'+1})$ and thus we have $S \in \mu_{\ell, k'}(\va_1, \ldots, \va_\ell, \vb_{\ell+1}, \ldots, \vb_k, \vd_{k+1}, \ldots, \vd_{k'})$.
						
			Since this argument is symmetric, we obtain\\
				\strut\hspace{2ex}
					$\mu_{\ell, k'}(\va_1, \ldots, \va_\ell, \vb_{\ell+1}, \ldots, \vb_k, \vc_{k+1}, \ldots, \vc_{k'})$\\
				\strut\hspace{5ex}	
						$= \mu_{\ell, k'}(\va_1, \ldots, \va_\ell, \vb_{\ell+1}, \ldots, \vb_k, \vd_{k+1}, \ldots, \vd_{k'})$.	
			\strut\hfill$\Diamond$	
			
		%%%%%%%%%%%%%%%%%%%%%%%%%%%%	
		\item \underline{Claim II:} 
			Let $\ell, k, k'$ be three integers such that $0 \leq \ell < k < k' \leq n$.
			Let $\va_1, \ldots, \va_\ell, \vb_{\ell+1}, \ldots, \vb_k$ be tuples of domain elements, where $\va_i \in \fU_\cA^{|\vy_i|}$, $\vb_i \in \fU_\cA^{|\vx_i|}$ for every $i$.
			Consider two sequences of tuples $\vc_{k+1}, \ldots, \vc_{k'}$ and $\vd_{k+1}, \ldots, \vd_{k'}$ that coincide in all positions that correspond to variables $x$ occurring in $\At_\ell$, where $\vc_i, \vd_i \in \fU_\cA^{|\vx_i|}$ for every $i$.
			We have\\
				\centerline{\hspace{-2ex}$\mu_{\ell,k'}(\va_1, \ldots, \va_\ell, \vb_{\ell+1}, \ldots, \vb_k, \vc_{k+1}, \ldots, \vc_{k'}) =  \mu_{\ell,k'}(\va_1, \ldots, \va_\ell, \vb_{\ell+1}, \ldots, \vb_k, \vd_{k+1}, \ldots, \vd_{k'})$.}
			
		\item \underline{Proof:} 
			For $k' = n$ Claim~II follows immediately from Claim~I.
			For $k' < n$ we simply pad the sequences $\vc_{k+1}, \ldots, \vc_{k'}$ and $\vd_{k+1}, \ldots, \vd_{k'}$ with tuples $\vc\hspace{0.075ex}'_{k'+1}, \ldots, \vc\hspace{0.075ex}'_n$ and $\vd'_{k'+1}, \ldots, \vd'_n$ for which we ensure $\vc\hspace{0.075ex}'_i = \vd'_i$ and $\vc\hspace{0.075ex}'_i, \vd'_i \in \fU_\cA^{|\vx_i|}$ for every $i$.
			Then, Claim~II follows from Claim~I applied to the sequences $\va_1, \ldots, \va_\ell, \vb_{\ell+1}, \ldots, \vb_{k}, \vc_{k+1}, \ldots, \vc_{k'}, \vc\hspace{0.075ex}'_{k'+1}, \ldots, \vc\hspace{0.075ex}'_n$ and $\va_1, \ldots, \va_\ell, \vb_{\ell+1}, \ldots, \vb_{k},$ $\vd_{k+1}, \ldots, \vd_{k'}, \vd'_{k'+1}, \ldots, \vd'_n$.
			\strut\hfill$\Diamond$	
	\end{description}	

	For $i = 1, \ldots, n$ we define $\fU_i$ as abbreviation of $\fU_\cA^{|\vx_1|} \times \ldots \times \fU_\cA^{|\vx_i|}$.
	We construct certain representatives $\alpha_{k, \<\bS_0^{(k)}, \ldots, \bS_{k-1}^{(k)}\>} \in \fU_k$ inductively as follows.
	The $\bS_i^{(k)}$ stand for sequences $S_{i,i+1}^{(k)} \ldots S_{i,k}^{(k)}$ of fingerprints satisfying $S_{i,k}^{(k)} \in S_{i,k-1}^{(k)} \in \ldots \in S_{i,i+1}^{(k)}$.
	\begin{description}
		\item Let $k=1$. 
			We partition
			$\fU_1$ into sets $\fU_{1,\<S_0^{(1)}\>}$ with $S_0^{(1)} \in \im_\sigma(\mu_{0,1})$ by setting
				$\fU_{1,\<S_0^{(1)}\>} := \bigl\{ \vb_1 \in \fU_\cA^{|\vx_1|} \mid \mu_{0,1}(\vb_1) = S_0^{(1)} \bigr\}$.
			We pick one representative $\alpha_{1,\<S_0^{(1)}\>} \in \fU_{1,\<S_0^{(1)}\>}$ from every part.
		
		\item Let $k > 1$. 
			We construct subsets $\fU_{k,\<\bS_0^{(k)}, \ldots, \bS_{k-1}^{(k)}\>} \subseteq \fU_k$ with 
				$S_{0,j}^{(k)} \in \im_\sigma(\mu_{0,j}), 
					\ldots, 
					S_{k-1,j}^{(k)} \in \im_\sigma(\mu_{k-1,j})$, 
			for every $j \leq k$,		
			by setting $\fU_{k,\<\bS_0^{(k)}, \ldots, \bS_{k-1}^{(k)}\>} :=$
				\begin{align*}
					 	\bigl\{ \bigl\<\vc_1, \ldots, \vc_{k-1}, \vb_k\bigr\> \bigm|\;
					 	&\text{$\vb_k \in \fU_\cA^{|\vx_k|}$ and there is some $\alpha_{k-1,\<\bS_0^{(k-1)}, \ldots, \bS_{k-2}^{(k-1)}\>} =$}\\
					 	&\text{$\bigl\< \vc_1, \ldots, \vc_{k-1} \bigr\>$, $\vc_i \in \fU_\cA^{|\vx_i|}$, for every $i$, such that} \\
					 	&\text{$\mu_{0,k}\bigl(\vc_1, \ldots, \vc_{k-1},  \vb_k\bigr) = S_0$,} \\
					 	&\text{$\mu_{1,k}\bigl(\sigma_1(\vc_1), \vc_2, \ldots, \vc_{k-1},  \vb_k\bigr) = S_1$,}\\
					 	&\qquad\vdots \\
					 	&\text{$\mu_{k-2,k}\bigl(\sigma_1(\vc_1), \ldots, \sigma_{k-2}(\vc_1, \ldots, \vc_{k-2}), \vc_{k-1}, \vb_k\bigr) = S_{k-2}$,} \\
					 	&\text{$\mu_{k-1,k}\bigl(\sigma_1(\vc_1), \ldots, \sigma_{k-1}(\vc_1, \ldots, \vc_{k-1}), \vb_k\bigr) = S_{k-1}$, }\\
					 	&\bS_0^{(k)} = \bS_0^{(k-1)} S_0, \\
					 	&\qquad\vdots \\
					 	&\bS_{k-2}^{(k)} = \bS_{k-2}^{(k-1)} S_{k-2}\text{, and }\\
					 	&\bS_{k-1}^{(k)} = S_{k-1} 									
					 	 	 \bigr\} ~.
				\end{align*}
			
			We pick one representative $\alpha_{k,\<\bS_0^{(k)}, \ldots, \bS_{k-1}^{(k)}\>}$ from each nonempty $\fU_{k,\<\bS_0^{(k)}, \ldots, \bS_{k-1}^{(k)}\>}$.
	\end{description}

	Having all the representatives $\alpha_{k,\<\bS_0^{(k)}, \ldots, \bS_{k-1}^{(k)}\>}$ at hand, we inductively construct $\hsigma$, starting from $\hsigma_1$ and going to $\hsigma_n$.
	\begin{description}
		\item Let $k=1$.
			For every $\vb_1 \in \fU_\cA^{|\vx_1|}$ we set $\hsigma_1(\vb_1) := \sigma_1(\alpha_{1,\<S_0\>})$, where $S_0 := \mu_{0,1}\bigl( \vb_1 \bigr)$.
	
		\item Let $k>1$.
			For all tuples $\vb_1 \in \fU_\cA^{|\vx_1|}, \ldots, \vb_k \in \fU_\cA^{|\vx_k|}$ we set $\hsigma_k(\vb_1, \ldots, \vb_k) := \sigma_k(\vc_1, \ldots, \vc_k)$, where 
			$\<\vc_1, \ldots, \vc_k\> : = \alpha_{k, \<\bS_0^{(k)}, \ldots, \bS_{k-1}^{(k)}\>}$ with $\vc_i \in \fU_\cA^{|\vx_i|}$, for every $i$, 
			and we have
			\begin{itemize}
				\item $S_{0,j}^{(k)} = \mu_{0,j}\bigl(\vb_1, \ldots, \vb_j\bigr)$ for every $j$, $0 < j \leq k$,
				\item $S_{1,j}^{(k)} = \mu_{1,j}\bigl(\hsigma_1(\vb_1), \vb_2, \ldots, \vb_j\bigr)$ for every $j$, $1 < j \leq k$,
				\item[] $\vdots$
				\item $S_{k-2,j}^{(k)} = \mu_{k-2,j}\bigl(\hsigma_1(\vb_1), \ldots, \hsigma_{k-2}(\vb_1, \ldots, \vb_{k-2}), \vb_{k-1}, \ldots, \vb_j \bigr)$ for every $j$, $k-2 < j \leq k$,
				\item $S_{k-1,k}^{(k)} = \mu_{k-1,k}\bigl(\hsigma_1(\vb_1), \ldots, \hsigma_{k-1}(\vb_1, \ldots, \vb_{k-1}), \vb_k \bigr)$,
			\end{itemize}	
			if such an $\alpha_{k, \<\bS_0^{(k)}, \ldots, \bS_{k-1}^{(k)}\>}$ exists. (We show in Claim IV that this is always the case.)
	\end{description}

	\begin{description}
		%%%%%%%%%%%%%%%%%%%%%%%%%%%%	
		\item\underline{Claim III:} For all $\ell, k$, $0 \leq \ell <  k \leq n$, we have $\im_{\hsigma}\bigl( \mu_{\ell,k} \bigr) \subseteq \im_{\sigma}\bigl( \mu_{\ell,k} \bigr)$. 
		
		\item\underline{Proof:}
			Fix some $\mu_{\ell,k}$ and let $S \in \im_{\hsigma}(\mu_{\ell,k})$.
			Hence, there are tuples $\vb_1 \in \fU_\cA^{|\vx_1|}, \ldots, \vb_k \in \fU_\cA^{|\vx_k|}$ such that $\hsigma_1(\vb_1), \ldots, \hsigma_k(\vb_1, \ldots, \vb_k)$ are defined and we have \\
				\centerline{$S = \mu_{\ell,k}\bigl( \hsigma_1(\vb_1), \ldots, \hsigma_\ell(\vb_1, \ldots, \vb_\ell), \vb_{\ell+1}, \ldots, \vb_k \bigr)$.}
			By definition of $\hsigma$, there are representatives 
			$\alpha_{j,\<\bS_0^{(j)}, \ldots, \bS_{j-1}^{(j)}\>} = \<\vc_1^{(j)}, \ldots, \vc_j^{(j)}\>$, $1 \leq j \leq \ell$, for which we observe the following properties.
			\begin{enumerate}[label=(\alph{*}), ref=(\alph{*})]
				\item\label{enum:proofUniformStrategies:I} 
					For every $i$, $0 \leq i < \ell$, all the $\bS_i^{(j)}$ are prefixes of $\bS_i^{(\ell)}$.
					This means, if we abbreviate $\bS_i^{(\ell)}$ to $\bS_i$, we have\\
						$S_{i,i+1} \ldots S_{i,\ell} 
							\;=\; \bS_{i}^{(1)} S_{i,i+2} \ldots S_{i,\ell} 
							\;=\; \bS_{i}^{(2)} S_{i,i+3} \ldots S_{i,\ell} 
							\;= \ldots 
							=\; \bS_{i}^{(\ell-1)} S_{i,\ell} 
							\;=\; \bS_{i}^{(\ell)}$.
					
				\item\label{enum:proofUniformStrategies:II} 
					For every $j$, $1 \leq j \leq \ell$, we have
						$\hsigma_j(\vb_1, \ldots, \vb_j) = \sigma_j(\vc_1^{(j)}, \ldots, \vc_j^{(j)})$.
			\end{enumerate}
			Because of~\ref{enum:proofUniformStrategies:I} and due to the construction of the $\alpha_{j,\<\bS_0^{(j)}, \ldots, \bS_{j-1}^{(j)}\>} = \<\vc_1^{(j)}, \ldots, \vc_j^{(j)}\>$, we have $\vc_i^{(j)} = \vc_i^{(j')}$ for every $i$, $1 \leq i \leq \ell$, and all $j,j'$, $1 \leq j,j' \leq \ell$.
			Hence, we can write $\vc_1, \ldots, \vc_\ell$ instead of $\vc_1^{(j)}, \ldots, \vc_1^{(j)}$ (for any $j$).
			Therefore, \ref{enum:proofUniformStrategies:II} entails \\
				\strut\hspace{2ex} 
				$S = \mu_{\ell,k}\bigl( \hsigma_1(\vb_1), \ldots, \hsigma_\ell(\vb_1, \ldots, \vb_\ell), \vb_{\ell+1}, \ldots, \vb_k \bigr)$ \\
				\strut\hspace{6ex} 
						$=  \mu_{\ell,k}\bigl( \sigma_1(\vc_1), \ldots, \sigma_\ell(\vc_1, \ldots, \vc_\ell), \vb_{\ell+1}, \ldots, \vb_k \bigr)$.\\
			Consequently, $S \in \im_\sigma(\mu_{\ell,k})$.
			\hfill$\Diamond$
	\end{description}	

	\begin{description}
		%%%%%%%%%%%%%%%%%%%%%%%%%%%%	
		\item\underline{Claim IV:} For every $k$, $1 \leq k \leq n$, and all tuples $\vb_1 \in \fU_\cA^{|\vx_1|}, \ldots, \vb_k \in \fU_\cA^{|\vx_k|}$ there is a representative $\alpha_{k, \<\bS_0, \ldots, \bS_{k-1}\>}$ such that
			\begin{itemize}
				\item $S_{0,j} = \mu_{0,j}\bigl( \vb_1, \ldots, \vb_{j} \bigr)$ for every $j$, $0 < j \leq k$,
				\item $S_{1,j} = \mu_{1,j}\bigl(\hsigma_1(\vb_1), \vb_2, \ldots, \vb_{j}\bigr)$ for every $j$, $1 < j \leq k$,
				\item[] $\vdots$
				\item $S_{k-2,j} = \mu_{k-2,j}\bigl(\hsigma_1(\vb_1), \ldots, \hsigma_{k-2}(\vb_1, \ldots, \vb_{k-2}), \vb_{k-1}, \ldots, \vb_j \bigr)$ for every $j$, $k-2 < j \leq k$,
				\item $S_{k-1,k} = \mu_{k-1,k}\bigl( \hsigma_1(\vb_1), \ldots, \hsigma_{k-1}(\vb_1, \ldots, \vb_{k-1}), \vb_k \bigr)$.
			\end{itemize}	

		\item\underline{Proof:} We proceed by induction on $k$.
			\begin{description}
				\item Let $k=1$. 
					Consider any tuple $\vb_1 \in \fU_\cA^{|\vx_1|}$ and set $S_0 := \mu_{0,1}\bigl( \vb_1 \bigr)$. 
					Hence, $S_0 \in \im_\sigma(\mu_{0,1})$ and we thus have defined the partition $\fU_{1,\<S_0\>}$. 
					Since $\vb_1 \in \fU_{1,\<S_0\>}$, the set is nonempty and there is a representative $\alpha_{1,\<S_0\>} \in \fU_{1,\<\cS_0\>}$.

				\item Let $k>1$. Consider any sequence of tuples $\vb_1 \in \fU_\cA^{|\vx_1|}, \ldots, \vb_k \in \fU_\cA^{|\vx_k|}$. 
					By Claim~III, we have
					\begin{itemize}
						\item $S_{0,j} \in \im_{\hsigma}(\mu_{0,j}) \subseteq \im_\sigma(\mu_{0,j})$ for every $j$, $0 < j \leq k$,
						\item[] $\vdots$
						\item $S_{k-1,j} \in \im_{\hsigma}(\mu_{k-1,j}) \subseteq \im_\sigma(\mu_{k-1,j})$ for every $j$, $k-1 < j \leq k$,
					\end{itemize}
					and, therefore, we have constructed the subset $\fU_{k, \<\bS_0, \ldots, \bS_{k-1}\>} \subseteq \fU_k$ when defining representatives above.
					We next show that this set is not empty.
				
					For every $\ell$, $0 \leq \ell < k-1$, we set $\bS_\ell^{(k-1)} := S_{\ell,\ell+1} \ldots S_{\ell,k-1}$.
					By induction, there is a representative $\alpha_{k-1,\<\bS_{0}^{(k-1)}, \ldots, \bS_{k-2}^{(k-1)}\>} =: \<\vc_1, \ldots, \vc_{k-1}\>$ with $\vc_i \in \fU_\cA^{|\vx_i|}$, for every $i$.
					
					As one consequence, the definition of $\hsigma$ entails 
					\begin{itemize}
						\item $\hsigma_1(\vb_1) = \sigma_1(\vc_1) = \hsigma_1(\vc_1)$,
						\item[] $\vdots$
						\item $\hsigma_{k-1}(\vb_1, \ldots, \vb_{k-1}) = \sigma_{k-1}(\vc_1, \ldots, \vc_{k-1}) = \hsigma_{k-1}(\vc_1, \ldots, \vc_{k-1})$,			
					\end{itemize}
					leading to 
					\begin{itemize}
						\item[($*$)] $\mu_{k-1,k}\bigl( \hsigma_1(\vc_1), \ldots, \hsigma_{k-1}(\vc_1, \ldots, \vc_{k-1}), \vb_{k} \bigr)$\\
							\strut\hspace{2ex}$= \mu_{k-1,k}\bigl( \hsigma_1(\vb_1), \ldots, \hsigma_{k-1}(\vb_1, \ldots, \vb_{k-1}), \vb_{k} \bigr) = S_{k-1,k}$.
					\end{itemize}

					By definition of the $\mu_{\ell, k-1}$ and since we know $S_{0,k} \in S_{0,k-1}, \ldots, S_{k-2,k} \in S_{k-2,k-1}$, the properties of $\alpha_{k-1, \<\bS_0^{(k-1)}, \ldots, \bS_{k-2}^{(k-1)}\>} = \<\vc_1, \ldots, \vc_{k-1}\>$ entail the existence of tuples $\vd_k^{(0)}, \ldots, \vd_k^{(k-2)} \in \fU_\cA^{|\vx_k|}$ such that
					\begin{itemize}
						\item $\mu_{0,k}\bigl( \vc_1, \ldots, \vc_{k-1}, \vd_k^{(0)} \bigr) = S_{0,k}$,
						\item $\mu_{1,k}\bigl( \hsigma_1(\vc_1), \vc_2, \ldots, \vc_{k-1}, \vd_k^{(1)} \bigr) = S_{1,k}$,
						\item[] $\vdots$
						\item $\mu_{k-2,k}\bigl( \hsigma_1(\vc_1), \ldots, \hsigma_{k-2}(\vc_1, \ldots, \vc_{k-2}), \vc_{k-1}, \vd_k^{(k-2)} \bigr) = S_{k-2,k}$,
						\item $\mu_{k-1,k}\bigl( \hsigma_1(\vc_1), \ldots, \hsigma_{k-2}(\vc_1, \ldots, \vc_{k-2}), \hsigma_{k-1}(\vc_1, \ldots, \vc_{k-1}), \vb_{k} \bigr) = S_{k-1,k}$
					\end{itemize}
					(the last equation follows from ($*$)).

					Due to $S_{0,k} \in \fP^{n-k+1}\At_0, \ldots, S_{k-2,k} \in \fP^{n-k+1}\At_{k-2}$, and $S_{k-1,k} \in \fP^{n-k+1}\At_{k-1}$, Condition~\ref{enum:GBSRaxiomatic:II} of Definition~\ref{definition:GBSRaxiomatic} entails pairwise disjointness of the sets $\vars(S_{0,k}) \cap \vx, \ldots, \vars(S_{k-2,k}) \cap \vx$, and $\vars(S_{k-1,k}) \cap \vx$. 
					Consequently, we can define a new tuple $\vd'_k$ by setting 
						\[ \fd'_{k,i} := 	\begin{cases}
									\fd_{k,i}^{(j)} 	&\text{if $x_{k, i} \in \vars(S_{j,k}) \cap \vx$ with $j < k-1$,} \\
									\fb_{k,i}		&\text{if $x_{k,i} \in \vars(S_{k-1,k}) \cap \vx$,} \\
									\fb_{k,i}		&\text{otherwise. \qquad (We could use any value here.)}
								\end{cases}
						\]
					Due to the pairwise disjointness of the sets $\vars(S_{0,k}) \cap \vx, \ldots, \vars(S_{k-1,k}) \cap \vx$, Claim~II implies that for every $\ell$, $0 \leq \ell < k-1$, \\
					\strut\hspace{2ex}
						$\mu_{\ell, k}\bigl( \hsigma_1(\vc_1), \ldots, \hsigma_\ell(\vc_1, \ldots, \vc_\ell), \vc_{\ell+1}, \ldots, \vc_{k-1}, \vd'_k \bigr)\hspace{5ex}\strut$ \\
					\strut\hspace{4ex}
						$= \mu_{\ell, k}\bigl( \hsigma_1(\vc_1), \ldots, \hsigma_\ell(\vc_1, \ldots, \vc_\ell), \vc_{\ell+1}, \ldots, \vc_{k-1}, \vd_k^{(\ell)} \bigr)$\\
					\strut\hspace{4ex}
							$= S_{\ell,k}$ \\
					and \\
					\strut\hspace{2ex}
						$\mu_{k-1, k}\bigl( \hsigma_1(\vc_1), \ldots, \hsigma_{k-1}(\vc_1, \ldots, \vc_{k-1}), \vd'_k \bigr)$ \\
					\strut\hspace{4ex}
							$= \mu_{k-1, k}\bigl( \hsigma_1(\vc_1), \ldots, \hsigma_{k-1}(\vc_1, \ldots, \vc_{k-1}), \vb_{k} \bigr)$\\
					\strut\hspace{4ex}
							$= S_{k-1,k}$. \\														
					Consequently, the set $\fU_{k, \<\bS_0, \ldots, \bS_{k-1}\>}$ contains at least the tuple $\<\vc_1, \ldots, \vc_{k-1}, \vd'_k\>$.
					Therefore, there exists some representative $\alpha_{k, \<\bS_0, \ldots, \bS_{k-1}\>} \in \fU_{k, \<\bS_0, \ldots, \bS_{k-1}\>}$.
					\hfill$\Diamond$
			\end{description}						
	\end{description}

	\begin{description}
		%%%%%%%%%%%%%%%%%%%%%%%%%%%%	
		\item\underline{Claim V:} $\hsigma$ is $\mu$-uniform.
		\item\underline{Proof:} By construction of $\hsigma$.	\hfill$\Diamond$
	\end{description}

	Now let $S \in \Out_{\At,\hsigma}$, i.e.\ there exist tuples $\vb_1 \in \fU_\cA^{|\vx_1|}, \ldots, \vb_n \in \fU_\cA^{|\vx_n|}$ such that
	$S = \out_{\At,\hsigma}(\vb_1, \ldots, \vb_n)$.
	We partition $S$ into sets $S_0 := S \cap \At_0, \ldots, S_n := S \cap \At_n$ and thus obtain the fingerprints 
		$S_\ell = \mu_{\ell,n}\bigl(\hsigma_{1}(\vb_1), \ldots, \hsigma_\ell(\vb_1, \ldots, \vb_\ell), \vb_{\ell+1}, \ldots, \vb_n \bigr) \subseteq \At_\ell$ for every $\ell$, $0 \leq \ell < n$.
	Claim IV guarantees the existence of some representative $\alpha_{n,\<\bS'_0, \ldots, \bS'_{n-1}\>} = \<\vc_1, \ldots, \vc_n\>$, with $\vc_i \in \fU_\cA^{|\vx_i|}$, for every $i$, such that
		$S_\ell = \mu_{\ell,n}\bigl(\sigma_{1}(\vc_1), \ldots, \sigma_\ell(\vc_1, \ldots, \vc_\ell), \vc_{\ell+1}, \ldots, \vc_n \bigr)$ for every $\ell$, $0 \leq \ell < n$.
	
	Consider any $A \in \At$, and fix the $\ell$ for which $A \in \At_\ell$.
	We distinguish two cases.
	Suppose that $\ell < n$.
		The definition of $\alpha_{n, \<\bS'_0, \ldots, \bS'_{n-1}\>}$ and the fingerprint functions $\mu_{\ell, n}$ entail that $A \in S_\ell$ if and only if \\
			\centerline{$\cA,[\vy_1 \Mapsto \hsigma_1(\vb_1), \ldots, \vy_\ell \Mapsto \hsigma_\ell(\vb_1, \ldots, \vb_{\ell}), \vx_{\ell+1} \Mapsto \vb_{\ell+1}, \ldots, \vx_n \Mapsto \vb_n] \models A$}
		if and only if \\
			\centerline{$\cA,[\vy_1 \Mapsto \sigma_1(\vc_1), \ldots, \vy_\ell \Mapsto \sigma_\ell(\vc_1, \ldots, \vc_{\ell}), \vx_{\ell+1} \Mapsto \vc_{\ell+1}, \ldots, \vx_n \Mapsto \vc_n] \models A$.}
				
	In case of $\ell = n$, we have $A \in S_n$ if and only if \\
			\centerline{$\cA,[\vy_1 \Mapsto \hsigma_1(\vb_1), \ldots, \vy_n \Mapsto \hsigma_n(\vb_1, \ldots, \vb_n)] \models A$}
		if and only if \\
			\centerline{$\cA,[\vy_1 \Mapsto \sigma_1(\vc_1), \ldots, \vy_n \Mapsto \sigma_n(\vc_1, \ldots, \vc_n)] \models A$.}	

	In both cases, we get $A \in \out_{\At,\hsigma}(\vb_1, \ldots, \vb_n)$ if and only if $A \in \out_{\At,\sigma}(\vc_1, \ldots, \vc_n)$.
	Consequently, we have $S = \out_{\At,\hsigma}(\vb_1, \ldots, \vb_n) = \out_{\At,\sigma}(\vc_1, \ldots,$ $\vc_n) \in \Out_{\At, \sigma}$. 
	
	Altogether, it follows that $\Out_{\At, \hsigma} \subseteq \Out_{\At, \sigma}$.
\end{proof}

%%%%%%%%%%%%%%%%%%%%%%%%%%%%%%%%%%%%%%%%%%%%%%%%%%%%%%%%%%%%%%%%%%%%%%%%%%%%%%%%
\subsection*{Proof of Lemma~\ref{lemma:FiniteModelsSemanticArgumentsGBSR}}

\begin{lemma*}
	If there is a satisfying $\mu$-uniform strategy $\sigma$ for $\psi$, then there is a model $\cB \models \varphi$ such that $\fU_\cB$ contains at most $n \cdot |\vy| \cdot \bigl( \twoup{\kappa+1}{|\text{\upshape{\At}}|} \bigr)^{n^2}$ domain elements where $\kappa = \degree_\varphi$.
\end{lemma*}
\begin{proof}~
	\begin{description}
		\item \underline{Claim I:} 
			Let $\ell, k$ be two integers such that $0 \leq \ell < k < n$.
			For all tuples $\va_1, \ldots, \va_\ell, \vb_{\ell+1}, \ldots, \vb_k$ with $\va_i \in \fU_\cA^{|\vy_i|}$ and $\vb_i \in \fU_\cA^{|\vx_i|}$, for every $i$, we observe that, if $\vars(\At_\ell) \cap \vx_{k+1} = \emptyset$, then
				$\bigl| \mu_{\ell,k}(\va_1, \ldots, \va_\ell, \vb_{\ell+1}, \ldots, \vb_k) \bigr| = 1$
			and, consequently,	
				$|\im_\sigma(\mu_{\ell,k})| \leq |\im_\sigma(\mu_{\ell,k+1})|$.
				
		\item \underline{Proof:} 
			Suppose there are sets $S_1, S_2 \in \mu_{\ell,k}(\va_1, \ldots, \va_\ell, \vb_{\ell+1}, \ldots, \vb_k)$ that are distinct.
			Hence, there are tuples $\vc_{k+1}, \vd_{k+1} \in \fU_\cA^{|\vx_{k+1}|}$ such that 
				$S_1 = \mu_{\ell,k+1}(\va_1, \ldots, \va_\ell, \vb_{\ell+1}, \ldots, \vb_k, \vc_{k+1})$ 
			and 
				$S_2 = \mu_{\ell,k+1}(\va_1, \ldots, \va_\ell, \vb_{\ell+1}, \ldots, \vb_k, \vd_{k+1})$.
			But since $\vx_{k+1} \cap \vars(\At_\ell) = \emptyset$, Claim~II from the proof of Lemma~\ref{lemma:UniformStrategies} entails $S_1 = S_2$.
			This contradicts our assumption that $S_1$ and $S_2$ are distinct.
			Consequently, $\mu_{\ell,k}(\va_1, \ldots, \va_\ell, \vb_{\ell+1}, \ldots, \vb_k)$ can contain at most one set.
		
			It is easy to show that $\mu_{\ell,k}(\va_1, \ldots, \va_\ell, \vb_{\ell+1}, \ldots, \vb_k)$ is nonempty by induction on  $k < n$, starting from $k = n-1$.
			\strut\hfill$\Diamond$
	\end{description}

	\begin{description}
		\item \underline{Claim II:}
			Let $\ell, k$ be two integers such that $0 \leq \ell < k < n$.
			We have $|\im_\sigma(\mu_{\ell,k})| \leq 2^{|\im_\sigma(\mu_{\ell,k+1})|}$.
		
		\item \underline{Proof:}
			For all tuples $\vb_1, \ldots, \ldots, \vb_k$ with $\vb_i \in \fU_\cA^{|\vx_i|}$, for every $i$, and
			for every $S \in \mu_{\ell,k}(\sigma_1(\vb_1), \ldots,$ $\sigma_\ell(\vb_1, \ldots, \vb_\ell), \vb_{\ell+1}, \ldots, \vb_k)$ we know that $S = \mu_{\ell,k}(\sigma_1(\vb_1), \ldots, \sigma_\ell(\vb_1, \ldots,$ $\vb_\ell), \vb_{\ell+1}, \ldots, \vb_k,$ $\vc_{k+1})$ for some tuple $\vc_{k+1}$.
			Hence, $\mu_{\ell,k}(\sigma_1(\vb_1), \ldots, \sigma_\ell(\vb_1, \ldots, \vb_\ell), \vb_{\ell+1}, \ldots, \vb_k) \subseteq \im_\sigma(\mu_{\ell,k+1})$.\\				
			\strut\hfill$\Diamond$
	\end{description}
	
	Due to Claim I and Claim II, our assumptions about $\kappa$ entail that
	\begin{itemize}
		\item[($*$)] for all integers $\ell, k$ with $0 \leq \ell < k \leq n$
					we obtain
						$|\im_\sigma(\mu_{\ell,k})| \leq \twoup{\kappa+1}{\At_\ell}$.
	\end{itemize}
		
	%%%%%%%%%%%%%%%%%%%%%%%%%%%%%%%%%%%%%%%%%%%%%%%%%%%%%%%%%%	
	Let $\cT_\sigma$ be the \emph{target set} of $\sigma$, defined by $\cT_\sigma := \bigcup_{k = 1}^n \cT_k$, where 
		\begin{align*}
			\cT_k := \bigl\{ \fa \in \fU_\cA \bigm|\; &\text{there are tuples $\vb_1, \ldots, \vb_k$ with $\vb_i \in \fU_\cA^{|\vx_i|}$, for every $i$,} \\
				&\text{such that $\sigma_k(\vb_1, \ldots, \vb_k) = \<\ldots, \fa, \ldots\>$} \bigr\} ~. 
		\end{align*}
	Since $\sigma$ is $\mu$-uniform, we know that $\cT_\sigma$ is a finite set.
	By definition of the fingerprint functions $\mu_{\ell, k}$, we derive the following upper bounds, where we write $\him_\sigma(\mu_{i,j})$ to abbreviate $\im_\sigma(\mu_{i,i+1}) \times \im_\sigma(\mu_{i,i+2}) \times \ldots \times \im_\sigma(\mu_{i,j})$ for all $i,j$, $0\leq i < j \leq n$.\\
	\strut\hspace{1ex}\!
		$\bigl| \cT_1 \bigr| \;\;\leq\;\; |\vy_1| \cdot \bigl|\him_\sigma(\mu_{0,1})\bigr| \;\;\leq\;\; |\vy_1| \cdot \twoup{n}{|\At_0|} \leq |\vy_1| \cdot \twoup{n}{|\At|}$, \\
	\strut\hspace{1ex}
		$\bigl| \cT_2 \bigr| \;\;\leq\;\; |\vy_2| \cdot \bigl|\him_\sigma(\mu_{0,2}) \times \him_\sigma(\mu_{1,2})\bigr|$ \\ 
	\strut\hspace{7ex}
			$\leq\;\;  |\vy_2| \cdot \twoup{n}{|\At_0|} \cdot \twoup{n-1}{|\At_0|} \cdot \twoup{n-1}{|\At_1|} \;\;\leq\;\; |\vy_2| \cdot \bigl( \twoup{n}{|\At|} \bigr)^3$, \\
	\strut\hspace{8ex}
		$\vdots$ \\
	\strut\hspace{1ex}
		$\bigl| \cT_n \bigr| \;\;\leq\;\; |\vy_n| \cdot \bigl|\him_\sigma(\mu_{0,n}) \times \ldots \times \him_\sigma(\mu_{n-1,n})\bigr|$ \\
	\strut\hspace{7ex}
		$\leq\;\;  |\vy_n| \cdot \prod_{i=0}^{n-1} \prod_{j=i}^{n-1} \twoup{n-j}{|\At_j|}$ \\
	\strut\hspace{7ex}
		 $\leq\;\; |\vy_n| \cdot \bigl( \twoup{n}{|\At|} \bigr)^{n^2}$.\\
	When we combine these bounds with the bound formulated in ($*$), we may infer that $\cT_\sigma$ contains at most $\sum_{\ell = 1}^n |\vy_\ell| \cdot \prod_{i=0}^{n-1} \prod_{j=i}^{n-1} \twoup{\min(\kappa+1, n-j)}{|\At_j|} \;\leq\; n \cdot |\vy| \cdot \bigl( \twoup{\kappa+1}{|\At|} \bigr)^{n^2}$ domain elements.

	Let $\varphi_\Sk$ be the result of exhaustive Skolemization of $\varphi$, i.e.\ every existentially quantified variable $y \in \vy_k$ in $\varphi$ is replaced by the Skolem function $f_y(\vx_1, \ldots, \vx_k)$.
	Clearly, $\sigma$ induces interpretations for all the Skolem functions $f_y$ such that $\cA$ can be extended to a model $\cA'$ of $\varphi_\Sk$ using these interpretations.
	More precisely, we can construct $\cA'$ from $\cA$ by setting $\fU_{\cA'} := \fU_\cA$ and $P^{\cA'} := P^\cA$ for every predicate symbol $P$ occurring in $\varphi$.
	Moreover, for every $k$, $1 \leq k \leq n$, with $\vy_k = \<y_1, \ldots, y_{|\vy_k|}\>$ the Skolem functions $f_{y_1}, \ldots, f_{y_{|\vy_k|}}$ are defined such that
	for all tuples $\vb_1, \ldots, \vb_k$ with $\vb_i \in \fU_\cA^{|\vx_i|}$, for every $i$, we set \\
		\centerline{$\bigl\< f_{y_1}(\vb_1, \ldots, \vb_k), \ldots, f_{y_{|\vy_k|}}(\vb_1, \ldots, \vb_k) \bigr\> := \sigma_k(\vb_1, \ldots, \vb_k)$.}
	Due to $\cA \models \varphi$, we get $\cA' \models \varphi_\Sk$.
	Moreover, we observe that for every $f_y$ with $y \in \vy_k$ and all tuples $\vb_1, \ldots, \vb_k$ it holds $f_y^{\cA'}(\vb_1, \ldots, \vb_k) \in \cT_k$.
	
	We now define the interpretation $\cB$ as follows.
	\begin{itemize}
		\item As $\cB$'s domain we use $\fU_\cB := \cT_\sigma$.
		\item For every $m$-ary predicate symbol $P$ occurring in $\varphi_\Sk$ we set $P^\cB := P^{\cA'} \cap \cT_\sigma^m$.
		\item For every $m$-ary function symbol $f$ occurring in $\varphi_\Sk$ we set $f^\cB(\vc\,) := f^{\cA'}(\vc\,)$ for every $m$-tuple $\vc \in \cT_\sigma^m$.
	\end{itemize}
	Clearly, $\cB$ is a substructure of $\cA'$.
	Hence, by the Substructure Lemma, $\cB$ satisfies $\varphi_\Sk$ and thus also the original $\varphi$.
	Moreover, we can bound the number of elements in $\cB$'s domain from above by $n \cdot |\vy| \cdot \bigl( \twoup{\kappa+1}{|\At|} \bigr)^{n^2}$.	
\end{proof}

%%%%%%%%%%%%%%%%%%%%%%%%%%%%%%%%%%%%%%%%%%%%%%%%%%%%%%%%%%%%%%%%%%%%%%%%%%%%%%%%%%%
%%%%%%%%%%%%%%%%%%%%%%%%%%%%%%%%%%%%%%%%%%%%%%%%%%%%%%%%%%%%%%%%%%%%%%%%%%%%%%%%%%%
%%%%%%%%%%%%%%%%%%%%%%%%%%%%%%%%%%%%%%%%%%%%%%%%%%%%%%%%%%%%%%%%%%%%%%%%%%%%%%%%%%%

\subsection{Proof details concerning Section~\ref{section:GeneralizedAckermann}}

%%%%%%%%%%%%%%%%%%%%%%%%%%%%%%%%%%%%%%%%%%%%%%%%%%%%%%%%%%%%%%%%%%%%%%%%%%%%%%%%
\subsection*{Proof of Lemma~\ref{lemma:GAFTransformation}}

\begin{lemma*}
	If $\varphi$ belongs to GAF, we can construct a semantically equivalent sentence $\varphi'$ in standard form, in which every subformula lies within the scope of at most one universal quantifier.
	Moreover, all literals in $\varphi'$ already occur in $\varphi$ (modulo variable renaming).
\end{lemma*}
\begin{proof}
	We proceed along similar lines as in the proof of Lemma~\ref{lemma:TransformationGBSR}, i.e.\ we perform syntactic transformations based on the axioms of Boolean algebra and the equivalences stated in Lemma~\ref{lemma:Miniscoping}. 
	Once more, this will not change the set of literals occurring in the intermediate steps (modulo variable renaming), since we start from a formula in negation normal form restricted to the connectives $\wedge, \vee, \neg$.
	Analogous to the proof of Lemma~\ref{lemma:TransformationGBSR}, we (re-)transform parts of $\varphi$ repeatedly into a disjunction of conjunctions (or a conjunction of disjunctions) of subformulas which we treat as indivisible units. 
	The literals and indivisible units in the respective conjunctions (disjunctions) will be grouped in accordance with the sets $\cL_0, \cL_x$ and $\cL_{x,\idx(x)}, \ldots, \cL_{x,n}$, where needed. 
	For this purpose, it is important to note that Lemma~\ref{lemma:GAFPropertiesOne}\ref{enum:GAFPropertiesOne:I} and the definition of $\cL_0$ entail that $\cL_0$ together with the sets $\cL_x$ partition the set of all literals occurring in $\varphi$.
	Moreover, every $\cL_x$ is partitioned by the sets $\cL_{x,0}, \cL_{x, \idx(x)}, \ldots, \cL_{x,n}$, by virtue of Lemma~\ref{lemma:GAFPropertiesTwo}\ref{enum:GAFPropertiesTwo:I}, \ref{enum:GAFPropertiesTwo:II} and the definition of $\cL_{x,0}$.
	
	We first give a description of the whole transformation process and afterwards present it formally below.
	
	Recall that $\varphi$ is of the shape $\forall \vx_1 \exists \vy_1 \ldots \forall \vx_n \exists \vy_n. \psi$.
	At the beginning, we transform $\psi$ into a disjunction of conjunctions of literals $\bigvee_i \psi_i$. 
	Moreover, we rewrite every $\psi_i$ into $\chi_{i,0}^{(1)} \wedge \bigwedge_{k=1}^{n} \bigwedge_{x \in \vx_k}$ $\bigl( \chi_{i,x,0}^{(1)} \wedge \bigwedge_{j=\idx(x)}^{n}  \chi_{i,x,j}^{(1)}\bigr)$, where $\chi_{i,0}^{(1)}$ and the $\chi_{i,x,k}^{(1)}$ are conjunctions of literals. 
	$\chi_{i,0}^{(1)}$ comprises all literals in $\psi_i$ which belong to $\cL_0$, while for every $k$ the literals, which occur in $\psi_i$ and belong to $\cL_{x,k}$, are grouped as $\chi_{i,x,k}^{(1)}$, respectively.
	By Lemmas~\ref{lemma:GAFPropertiesOne}\ref{enum:GAFPropertiesOne:II} and \ref{lemma:GAFPropertiesTwo}\ref{enum:GAFPropertiesTwo:IV}, \ref{enum:GAFPropertiesTwo:V}, we know that $\vars(\chi_{i,x,0}^{(1)}) \subseteq \{x\} \cup \bigcup_{k < \idx(x)} \vy_k$ and $\vars(\chi_{i,x,j}^{(1)}) \subseteq \{x\} \cup \bigcup_{k \leq \idx(x)} \vy_k$. Moreover, the definition of $\cL_0$ entails $\vars(\chi_{i,0}^{(1)}) \subseteq \vy$.
	
	At this point, we move the existential quantifier block $\exists \vy_n$ inwards. By Lemma~\ref{lemma:GAFPropertiesOne}\ref{enum:GAFPropertiesOne:III} and \ref{enum:GAFPropertiesOne:IV}, the subformulas $\chi_{i,0}^{(1)}$ and $\chi_{i,x,0}^{(1)}$ do not share any variables from $\vy_n$. 
	Moreover, due to Lemma~\ref{lemma:GAFPropertiesTwo}\ref{enum:GAFPropertiesTwo:IV} and \ref{enum:GAFPropertiesTwo:V}, the $\chi_{i,x,k}^{(1)}$ with $k < n$ do not contain any variables from $\vy_n$. 
	Consequently, one part of the quantifier block $\exists \vy_n$, namely $\exists \vy_n \cap \vars(\chi_{i,0}^{(1)})$, binds $\chi_{i,0}^{(1)}$ (for convenience, we still write the full $\exists \vy_n$, which does not affect semantics), and another---disjoint---part, namely $\exists \vy_n\cap Y_{x,n}$, binds $\chi_{i,x,n}^{(1)}$.
	 The thus obtained sentence $\varphi''$ has the form 
	 	$\forall \vx_1 \exists \vy_1 \ldots \forall \vx_n. \bigvee_{i} \bigl( \exists \vy_n. \chi_{i,0}^{(1)} \bigr)
			\wedge \bigwedge_{k=1}^{n} \bigwedge _{x \in \vx_k}$ $\bigl( \chi_{i,x,0}^{(1)}
			\wedge \bigwedge_{j=\idx(x)}^{n-1} \chi_{i,x,j}^{(1)} 
			\wedge \exists(\vy_n \cap Y_{x,n}). \chi_{i,x,n}^{(1)} \bigr)$. 
	In further transformations we shall treat the subformulas $\bigl( \exists \vy_n. \chi_{i,0}^{(1)} \bigr)$ and $\bigl(\exists(\vy_n \cap Y_{x,n}). \chi_{i,x,n}^{(1)} \bigr)$ as indivisible units.
	
	Next, we transform the big disjunction in $\varphi''$ into a conjunction of disjunctions $\bigwedge_i \psi'_i$, rewrite the disjunctions $\psi'_i$ into subformulas $\eta_{i,0}^{(1)} \vee \bigvee_{k=1}^{n-1} \bigvee _{x \in \vx_k} \eta_{i,x}^{(1)} \vee \bigvee _{x \in \vx_n} \eta_{i,x}^{(1)}$ (similarly to what we have done above, but this time grouped in accordance with the more coarse-grained sets $\cL_0$ and $\cL_x$). 
	Having done the regrouping, we move the universal quantifier block $\forall \vx_n$ inwards. 
	The resulting formula has the shape $\forall \vx_1 \exists \vy_1 \ldots \forall \vx_{n-1} \exists \vy_{n-1}. \bigwedge_{i} \eta_{i,0}^{(1)} \vee \bigvee_{k=1}^{n-1} \bigvee _{x \in \vx_k} \eta_{i,x}^{(1)} \vee \bigvee _{x \in \vx_n} \forall x.\, \eta_{i,x}^{(1)}$. 
	In further transformations, we shall treat the subformulas $\bigl( \forall x.\, \eta_{i,x}^{(1)} \bigr)$ as indivisible units as well. 
	Moreover, we shall group them under the conjunctions $\chi_{i,0}^{(\ell)}$ or $\eta_{i,0}^{(\ell)}$, $\ell \geq 2$, respectively, since they do not contain any free occurrences of universally quantified variables $x \in \vx$ anymore. 
	This is not only convenient but also necessary, because a subformula $\bigl( \forall x.\, \eta_{i,x}^{(1)} \bigr)$ may share free variables $y \in \bigcup_{k < \\idx(x)} \vy_k$ with the subformula $\eta_{i,0}^{(1)}$. 
	Hence, when the quantifier block $\exists y$ is moved inwards later on, both $\bigl( \forall x.\, \eta_{i,x}^{(1)} \bigr)$ and some literals in $\eta_{i,0}^{(1)}$ might have to remain within the scope of $\exists y$.

	We reiterate the described process until all the quantifiers have been moved inwards in the outlined way.
	There is one more peculiarity to mention. At later stages of the transformation subformulas of the form $\chi_{i, x, j}^{(\ell)} \wedge \ldots \wedge \chi_{i, x,n}^{(\ell)}$ may appear in which the constituents $\chi_{i,x,j'}^{(\ell)}$ may share variables $y \in \vy_j$, for instance. 
	We shall abbreviate such subformulas by $\chi_{i,x,\geq j}^{(\ell)}$ and similar notations, for the sake of readability. 
	Emerging subformulas $\bigl(\exists (\vy_\ell \cap Y_{x,\ell}). \chi_{i,x,\geq j}^{(\ell)}\bigr)$ will be treated as indivisible units.
	\begin{align*}
		&\forall \vx_1 \exists \vy_1 \ldots \forall \vx_n \exists \vy_n. \psi \\
		&\semequiv\; \forall \vx_1 \exists \vy_1 \ldots \exists \vy_n. \bigvee_{i} \chi_{i,0}^{(1)}(\vy_1, \ldots, \vy_n) \\
			&\hspace{23ex}\wedge \bigwedge_{k=1}^{n}\; \bigwedge _{x \in \vx_k}\; \biggl(\chi_{i,x,0}^{(1)}(x, \vy_1, \ldots, \vy_{\idx(x)-1})\\
			&\hspace{36ex}\wedge\! \bigwedge_{j=\idx(x)}^{n}  \chi_{i,x,j}^{(1)}(x, \vy_1, \ldots, \vy_{j}) \biggr) \\
		&\semequiv\; \forall \vx_1 \exists \vy_1 \ldots \forall \vx_n. \bigvee_{i} \Bigl( \exists \vy_n. \chi_{i,0}^{(1)}(\vy_1, \ldots, \vy_n) \Bigr) \\[-1ex]	
			&\hspace{23ex}\wedge \bigwedge_{k=1}^{n}\; \bigwedge _{x \in \vx_k}\; \biggl( \chi_{i,x,0}^{(1)}(x, \vy_1, \ldots, \vy_{\idx(x)-1})\\[-1ex]
			&\hspace{36ex} \wedge \bigwedge_{j=\idx(x)}^{n-1} \chi_{i,x,j}^{(1)}(x, \vy_1, \ldots, \vy_{j}) \\[-1ex]
			&\hspace{47ex}	\wedge \exists(\vy_n \cap Y_{x,n}). \chi_{i,x,n}^{(1)}(x, \vy_1, \ldots, \vy_n) \biggr) \\[-0.5ex]
		&\semequiv\; \forall \vx_1 \exists \vy_1 \ldots \forall \vx_n. \bigwedge_{i} \eta_{i,0}^{(1)}(\vy_1, \ldots, \vy_{n-1}) \\[-1ex]
			&\hspace{23ex}\vee \bigvee_{k=1}^{n-1}\; \bigvee _{x \in \vx_k}\; \eta_{i,x}^{(1)}(x, \vy_1, \ldots, \vy_{n-1}) \\
			&\hspace{23ex}\vee \bigvee _{x \in \vx_n}\; \eta_{i,x}^{(1)}(x, \vy_1, \ldots, \vy_{n-1}) \\
		&\semequiv\; \forall \vx_1 \exists \vy_1 \ldots \forall \vx_{n-1} \exists \vy_{n-1}. \bigwedge_{i} \eta_{i,0}^{(1)}(\vy_1, \ldots, \vy_{n-1}) \\[-1ex]
			&\hspace{33ex}\vee \bigvee_{k=1}^{n-1}\; \bigvee _{x \in \vx_k}\; \eta_{i,x}^{(1)}(x, \vy_1, \ldots, \vy_{n-1}) \\[-0.5ex]
			&\hspace{33ex}\vee \bigvee _{x \in \vx_n}\; \Bigl( \forall x.\, \eta_{i,x}^{(1)}(x, \vy_1, \ldots, \vy_{n-1}) \Bigr) 
	\end{align*}
	\begin{align*}	
		&\semequiv\; \forall \vx_1 \exists \vy_1 \ldots \forall \vx_{n-1} \exists \vy_{n-1}. \\[-1ex]
			&\hspace{5ex}\bigvee_{i} \chi_{i,0}^{(2)}(\vy_1, \ldots, \vy_{n-1}) \\[-1ex]
			&\hspace{8ex}\wedge \bigwedge_{k=1}^{n-1} \bigwedge _{x \in \vx_k} \biggl( \chi_{i,x,0}^{(2)}(x, \vy_1, \ldots, \vy_{\idx(x)-1})\\[-1ex]
			&\hspace{22ex} \wedge\! \bigwedge_{j=\idx(x)}^{n-2}  \chi_{i,x,j}^{(2)}(x, \vy_1, \ldots, \vy_{j}) 
				\wedge \chi_{i,x,\geq n-1}^{(2)}(x, \vy_1, \ldots, \vy_{n-1}) \biggr) \\
		&\semequiv\; \forall \vx_1 \exists \vy_1 \ldots \forall \vx_{n-1}.\\[-1ex]
			&\hspace{5ex} \bigvee_{i} \biggl( \exists \vy_{n-1}.\; \chi_{i,0}^{(2)}(\vy_1, \ldots, \vy_{n-1}) \biggl) \\[-1ex]
			&\hspace{10ex}\wedge \bigwedge_{k=1}^{n-1}\; \bigwedge _{x \in \vx_k}\; \biggl( \chi_{i,x,0}^{(2)}(x, \vy_1, \ldots, \vy_{\idx(x)-1}) \wedge \bigwedge_{j=\idx(x)}^{n-2}  \chi_{i,x,j}^{(2)}(x, \vy_1, \ldots, \vy_{j}) \\[-1ex]
			&\hspace{27ex}	\wedge \Bigl( \exists(\vy_{n-1} \cap Y_{x,{n-1}}).\, \chi_{i,x,\geq n-1}^{(2)}(x, \vy_1, \ldots, \vy_{n-1}) \Bigr) \biggr) \\
		&\hspace{10ex}\vdots	\\	
		&\semequiv\; \forall \vx_1 \exists \vy_1.\;\;
			\bigvee_{i}\; \chi_{i,0}^{(n)}(\vy_1)
			\;\wedge \bigwedge _{x \in \vx_1}\; \chi_{i,x,0}^{(n)}(x) \wedge \chi_{i,x,\geq 1}^{(n)}(x, \vy_1)\\
		&\semequiv\; \forall \vx_1.\;\;
			\bigvee_{i}\; \Bigl(\exists \vy_1.\chi_{i,0}^{(n)}(\vy_1)\Bigr)
			\;\wedge \bigwedge _{x \in \vx_1} \chi_{i,x,0}^{(n)}(x) \wedge \exists (\vy_1 \cap Y_{x,1}).\chi_{i,x,\geq 1}^{(n)}(x, \vy_1)\\
		&\semequiv\; \forall \vx_1.\;\;
			\bigwedge_{i}\; \eta_{i,0}^{(n)}()
			\;\vee \bigvee _{x \in \vx_1} \eta_{i,x}^{(n)}(x)\\
		&\semequiv\; \bigwedge_{i}\; \eta_{i,0}^{(n)}() \;\vee \bigvee _{x \in \vx_1} \forall x.\; \eta_{i,x}^{(n)}(x)
	\end{align*}
	The final result of this transformation is the sought $\varphi'$.
	Every time a universal quantifier block $\forall \vx_j$ is moved inwards at the $\ell$-th stage, the only subformulas, which contain universal quantifiers already, are grouped into $\eta_{i,0}^{(\ell)}$. 
	Due to the disjointness properties in Lemmas~\ref{lemma:GAFPropertiesOne} and \ref{lemma:GAFPropertiesTwo}, it is guaranteed that no $\eta_{i,0}^{(\ell)}$ contains a free occurrence of some $x \in \vx$ (details have been elaborated above).
	Consequently, in the final result $\varphi'$ we do not have any nested occurrences of universal quantifiers.
			
	By appropriately renaming variables in $\varphi'$, we may restore the property that no two quantifiers in $\varphi'$ bind the same variables.
\end{proof}

\end{document}